\documentclass[nobib,nofonts,justified]{tufte-handout-mct} %

\pdfoutput=1 %

\usepackage{mmap}

\newcommand{\ccs}[1]{}
\newcommand{\ndss}[1]{}
\newcommand{\tufte}[1]{#1}

\newenvironment{fw}{\vspace{-3ex}\begin{fullwidth}}{\end{fullwidth}}
\newcommand{\brk}{}
\newcommand{\aln}{}
\newcommand{\brkaln}[1]{}
\newcommand{\brkalnaln}[1]{&}
\newcommand{\DP}{differential privacy}

\setcaptionlabelfont{\bf}

\titleformat{\section}%
  [hang]%
  {\normalfont\Large\bf}%
  {\thesection.}%
  {1em}%
  {}%
  []%

\titleformat{\subsection}%
  [hang]%
  {\normalfont\large}%
  {\thesubsection}%
  {1em}%
  {}%
  []%

\titleformat{\subsection}%
  [hang]%
  {\normalfont\large}%
  {\thesubsection}%
  {1em}%
  {}%
  []%

\titleformat{\subsubsection}%
  [hang]%
  {\normalfont\normalsize\textit}%
  {\thesubsubsection}%
  {1em}%
  {}%
  []%

\newenvironment{DIFnomarkup}{}{}

\usepackage[utf8]{inputenc}
\usepackage{newunicodechar}%
\clearpage{}%
\newunicodechar{£}{\sterling}
\newunicodechar{¥}{\yen}
\newunicodechar{¬}{\neg}
\newunicodechar{±}{\pm}
\newunicodechar{·}{\cdotp}
\newunicodechar{×}{\times}
\newunicodechar{ð}{\matheth}
\newunicodechar{÷}{\div}
\newunicodechar{Ƶ}{\Zbar}
\newunicodechar{Α}{\Alpha}
\newunicodechar{Β}{\Beta}
\newunicodechar{Γ}{\Gamma}
\newunicodechar{Δ}{\Delta}
\newunicodechar{Ε}{\Epsilon}
\newunicodechar{Ζ}{\Zeta}
\newunicodechar{Η}{\Eta}
\newunicodechar{Θ}{\Theta}
\newunicodechar{Ι}{\Iota}
\newunicodechar{Κ}{\Kappa}
\newunicodechar{Λ}{\Lambda}
\newunicodechar{Μ}{\Mu}
\newunicodechar{Ν}{\Nu}
\newunicodechar{Ξ}{\Xi}
\newunicodechar{Ο}{\Omicron}
\newunicodechar{Π}{\Pi}
\newunicodechar{Ρ}{\Rho}
\newunicodechar{Σ}{\Sigma}
\newunicodechar{Τ}{\Tau}
\newunicodechar{Υ}{\Upsilon}
\newunicodechar{Φ}{\Phi}
\newunicodechar{Χ}{\Chi}
\newunicodechar{Ψ}{\Psi}
\newunicodechar{Ω}{\Omega}
\newunicodechar{α}{\alpha}
\newunicodechar{β}{\beta}
\newunicodechar{γ}{\gamma}
\newunicodechar{δ}{\delta}
\newunicodechar{ε}{\epsilon}
\newunicodechar{ζ}{\zeta}
\newunicodechar{η}{\eta}
\newunicodechar{θ}{\theta}
\newunicodechar{ι}{\iota}
\newunicodechar{κ}{\kappa}
\newunicodechar{λ}{\lambda}
\newunicodechar{μ}{\mu}
\newunicodechar{ν}{\nu}
\newunicodechar{ξ}{\xi}
\newunicodechar{ο}{\omicron}
\newunicodechar{π}{\pi}
\newunicodechar{ρ}{\rho}
\newunicodechar{ς}{\varsigma}
\newunicodechar{σ}{\sigma}
\newunicodechar{τ}{\tau}
\newunicodechar{υ}{\upsilon}
\newunicodechar{φ}{\varphi}
\newunicodechar{χ}{\chi}
\newunicodechar{ψ}{\psi}
\newunicodechar{ω}{\omega}
\newunicodechar{ϐ}{\varbeta}
\newunicodechar{ϑ}{\vartheta}
\newunicodechar{ϕ}{\phi}
\newunicodechar{ϖ}{\varpi}
\newunicodechar{Ϙ}{\oldKoppa}
\newunicodechar{ϙ}{\oldkoppa}
\newunicodechar{Ϛ}{\Stigma}
\newunicodechar{ϛ}{\stigma}
\newunicodechar{Ϝ}{\Digamma}
\newunicodechar{ϝ}{\digamma}
\newunicodechar{Ϟ}{\Koppa}
\newunicodechar{ϟ}{\koppa}
\newunicodechar{Ϡ}{\Sampi}
\newunicodechar{ϡ}{\sampi}
\newunicodechar{ϰ}{\varkappa}
\newunicodechar{ϱ}{\varrho}
\newunicodechar{ϴ}{\varTheta}
\newunicodechar{ϵ}{\varepsilon}
\newunicodechar{϶}{\backepsilon}
\newunicodechar{―}{\horizbar}
\newunicodechar{‖}{\Vert}
\newunicodechar{‗}{\twolowline}
\newunicodechar{†}{\dagger}
\newunicodechar{‡}{\ddagger}
\newunicodechar{•}{\smblkcircle}
\newunicodechar{‥}{\enleadertwodots}
\newunicodechar{…}{\unicodeellipsis}
\newunicodechar{′}{\prime}
\newunicodechar{″}{\dprime}
\newunicodechar{‴}{\trprime}
\newunicodechar{‵}{\backprime}
\newunicodechar{‶}{\backdprime}
\newunicodechar{‷}{\backtrprime}
\newunicodechar{‸}{\caretinsert}
\newunicodechar{‼}{\Exclam}
\newunicodechar{⁀}{\tieconcat}
\newunicodechar{⁃}{\hyphenbullet}
\newunicodechar{⁄}{\fracslash}
\newunicodechar{⁇}{\Question}
\newunicodechar{⁐}{\closure}
\newunicodechar{⁗}{\qprime}
\newunicodechar{€}{\euro}
\newunicodechar{⃐}{\leftharpoonaccent}
\newunicodechar{⃑}{\rightharpoonaccent}
\newunicodechar{⃒}{\vertoverlay}
\newunicodechar{⃖}{\overleftarrow}
\newunicodechar{⃗}{\vec}
\newunicodechar{⃛}{\dddot}
\newunicodechar{⃜}{\ddddot}
\newunicodechar{⃝}{\enclosecircle}
\newunicodechar{⃞}{\enclosesquare}
\newunicodechar{⃟}{\enclosediamond}
\newunicodechar{⃡}{\overleftrightarrow}
\newunicodechar{⃤}{\enclosetriangle}
\newunicodechar{⃧}{\annuity}
\newunicodechar{ℂ}{\BbbC}
\newunicodechar{ℇ}{\Eulerconst}
\newunicodechar{ℊ}{\mscrg}
\newunicodechar{ℋ}{\mscrH}
\newunicodechar{ℌ}{\mfrakH}
\newunicodechar{ℍ}{\BbbH}
\newunicodechar{ℎ}{\Planckconst}
\newunicodechar{ℏ}{\hslash}
\newunicodechar{ℐ}{\mscrI}
\newunicodechar{ℑ}{\Im}
\newunicodechar{ℒ}{\mscrL}
\newunicodechar{ℓ}{\ell}
\newunicodechar{ℕ}{\BbbN}
\newunicodechar{℘}{\wp}
\newunicodechar{ℙ}{\BbbP}
\newunicodechar{ℚ}{\BbbQ}
\newunicodechar{ℛ}{\mscrR}
\newunicodechar{ℜ}{\Re}
\newunicodechar{ℝ}{\BbbR}
\newunicodechar{ℤ}{\BbbZ}
\newunicodechar{℧}{\mho}
\newunicodechar{ℨ}{\mfrakZ}
\newunicodechar{℩}{\turnediota}
\newunicodechar{Å}{\Angstrom}
\newunicodechar{ℬ}{\mscrB}
\newunicodechar{ℭ}{\mfrakC}
\newunicodechar{ℯ}{\mscre}
\newunicodechar{ℰ}{\mscrE}
\newunicodechar{ℱ}{\mscrF}
\newunicodechar{Ⅎ}{\Finv}
\newunicodechar{ℳ}{\mscrM}
\newunicodechar{ℴ}{\mscro}
\newunicodechar{ℵ}{\aleph}
\newunicodechar{ℶ}{\beth}
\newunicodechar{ℷ}{\gimel}
\newunicodechar{ℸ}{\daleth}
\newunicodechar{ℼ}{\Bbbpi}
\newunicodechar{ℽ}{\Bbbgamma}
\newunicodechar{ℾ}{\BbbGamma}
\newunicodechar{ℿ}{\BbbPi}
\newunicodechar{⅀}{\Bbbsum}
\newunicodechar{⅁}{\Game}
\newunicodechar{⅂}{\sansLturned}
\newunicodechar{⅃}{\sansLmirrored}
\newunicodechar{⅄}{\Yup}
\newunicodechar{ⅅ}{\mitBbbD}
\newunicodechar{ⅆ}{\mitBbbd}
\newunicodechar{ⅇ}{\mitBbbe}
\newunicodechar{ⅈ}{\mitBbbi}
\newunicodechar{ⅉ}{\mitBbbj}
\newunicodechar{⅊}{\PropertyLine}
\newunicodechar{⅋}{\upand} %
\newunicodechar{←}{\leftarrow}
\newunicodechar{↑}{\uparrow}
\newunicodechar{→}{\rightarrow}
\newunicodechar{↓}{\downarrow}
\newunicodechar{↔}{\leftrightarrow}
\newunicodechar{↕}{\updownarrow}
\newunicodechar{↖}{\nwarrow}
\newunicodechar{↗}{\nearrow}
\newunicodechar{↘}{\searrow}
\newunicodechar{↙}{\swarrow}
\newunicodechar{↚}{\nleftarrow}
\newunicodechar{↛}{\nrightarrow}
\newunicodechar{↜}{\leftwavearrow}
\newunicodechar{↝}{\rightwavearrow}
\newunicodechar{↞}{\twoheadleftarrow}
\newunicodechar{↟}{\twoheaduparrow}
\newunicodechar{↠}{\twoheadrightarrow}
\newunicodechar{↡}{\twoheaddownarrow}
\newunicodechar{↢}{\leftarrowtail}
\newunicodechar{↣}{\rightarrowtail}
\newunicodechar{↤}{\mapsfrom}
\newunicodechar{↥}{\mapsup}
\newunicodechar{↦}{\mapsto}
\newunicodechar{↧}{\mapsdown}
\newunicodechar{↨}{\updownarrowbar}
\newunicodechar{↩}{\hookleftarrow}
\newunicodechar{↪}{\hookrightarrow}
\newunicodechar{↫}{\looparrowleft}
\newunicodechar{↬}{\looparrowright}
\newunicodechar{↭}{\leftrightsquigarrow}
\newunicodechar{↮}{\nleftrightarrow}
\newunicodechar{↯}{\downzigzagarrow}
\newunicodechar{↰}{\Lsh}
\newunicodechar{↱}{\Rsh}
\newunicodechar{↲}{\Ldsh}
\newunicodechar{↳}{\Rdsh}
\newunicodechar{↴}{\linefeed}
\newunicodechar{↵}{\carriagereturn}
\newunicodechar{↶}{\curvearrowleft}
\newunicodechar{↷}{\curvearrowright}
\newunicodechar{↸}{\barovernorthwestarrow}
\newunicodechar{↹}{\barleftarrowrightarrowbar}
\newunicodechar{↺}{\acwopencirclearrow}
\newunicodechar{↻}{\cwopencirclearrow}
\newunicodechar{↼}{\leftharpoonup}
\newunicodechar{↽}{\leftharpoondown}
\newunicodechar{↾}{\upharpoonright}
\newunicodechar{↿}{\upharpoonleft}
\newunicodechar{⇀}{\rightharpoonup}
\newunicodechar{⇁}{\rightharpoondown}
\newunicodechar{⇂}{\downharpoonright}
\newunicodechar{⇃}{\downharpoonleft}
\newunicodechar{⇄}{\rightleftarrows}
\newunicodechar{⇅}{\updownarrows}
\newunicodechar{⇆}{\leftrightarrows}
\newunicodechar{⇇}{\leftleftarrows}
\newunicodechar{⇈}{\upuparrows}
\newunicodechar{⇉}{\rightrightarrows}
\newunicodechar{⇊}{\downdownarrows}
\newunicodechar{⇋}{\leftrightharpoons}
\newunicodechar{⇌}{\rightleftharpoons}
\newunicodechar{⇍}{\nLeftarrow}
\newunicodechar{⇎}{\nLeftrightarrow}
\newunicodechar{⇏}{\nRightarrow}
\newunicodechar{⇐}{\Leftarrow}
\newunicodechar{⇑}{\Uparrow}
\newunicodechar{⇒}{\Rightarrow}
\newunicodechar{⇓}{\Downarrow}
\newunicodechar{⇔}{\Leftrightarrow}
\newunicodechar{⇕}{\Updownarrow}
\newunicodechar{⇖}{\Nwarrow}
\newunicodechar{⇗}{\Nearrow}
\newunicodechar{⇘}{\Searrow}
\newunicodechar{⇙}{\Swarrow}
\newunicodechar{⇚}{\Lleftarrow}
\newunicodechar{⇛}{\Rrightarrow}
\newunicodechar{⇜}{\leftsquigarrow}
\newunicodechar{⇝}{\rightsquigarrow}
\newunicodechar{⇞}{\nHuparrow}
\newunicodechar{⇟}{\nHdownarrow}
\newunicodechar{⇠}{\leftdasharrow}
\newunicodechar{⇡}{\updasharrow}
\newunicodechar{⇢}{\rightdasharrow}
\newunicodechar{⇣}{\downdasharrow}
\newunicodechar{⇤}{\barleftarrow}
\newunicodechar{⇥}{\rightarrowbar}
\newunicodechar{⇦}{\leftwhitearrow}
\newunicodechar{⇧}{\upwhitearrow}
\newunicodechar{⇨}{\rightwhitearrow}
\newunicodechar{⇩}{\downwhitearrow}
\newunicodechar{⇪}{\whitearrowupfrombar}
\newunicodechar{⇴}{\circleonrightarrow}
\newunicodechar{⇵}{\downuparrows}
\newunicodechar{⇶}{\rightthreearrows}
\newunicodechar{⇷}{\nvleftarrow}
\newunicodechar{⇸}{\nvrightarrow}
\newunicodechar{⇹}{\nvleftrightarrow}
\newunicodechar{⇺}{\nVleftarrow}
\newunicodechar{⇻}{\nVrightarrow}
\newunicodechar{⇼}{\nVleftrightarrow}
\newunicodechar{⇽}{\leftarrowtriangle}
\newunicodechar{⇾}{\rightarrowtriangle}
\newunicodechar{⇿}{\leftrightarrowtriangle}
\newunicodechar{∀}{\forall}
\newunicodechar{∁}{\complement}
\newunicodechar{∂}{\partial}
\newunicodechar{∃}{\exists}
\newunicodechar{∄}{\nexists}
\newunicodechar{∅}{\varnothing}
\newunicodechar{∆}{\increment}
\newunicodechar{∇}{\nabla}
\newunicodechar{∈}{\in}
\newunicodechar{∉}{\notin}
\newunicodechar{∊}{\smallin}
\newunicodechar{∋}{\ni}
\newunicodechar{∌}{\nni}
\newunicodechar{∍}{\smallni}
\newunicodechar{∎}{\QED}
\newunicodechar{∏}{\prod}
\newunicodechar{∐}{\coprod}
\newunicodechar{∑}{\sum}
\newunicodechar{−}{\minus}
\newunicodechar{∓}{\mp}
\newunicodechar{∔}{\dotplus}
\newunicodechar{∕}{\divslash}
\newunicodechar{∖}{\smallsetminus}
\newunicodechar{∗}{\ast}
\newunicodechar{∘}{\vysmwhtcircle}
\newunicodechar{∙}{\vysmblkcircle}
\newunicodechar{√}{\sqrt}
\newunicodechar{∛}{\cuberoot}
\newunicodechar{∜}{\fourthroot}
\newunicodechar{∝}{\propto}
\newunicodechar{∞}{\infty}
\newunicodechar{∟}{\rightangle}
\newunicodechar{∠}{\angle}
\newunicodechar{∡}{\measuredangle}
\newunicodechar{∢}{\sphericalangle}
\newunicodechar{∣}{\mid}
\newunicodechar{∤}{\nmid}
\newunicodechar{∥}{\parallel}
\newunicodechar{∦}{\nparallel}
\newunicodechar{∧}{\wedge}
\newunicodechar{∨}{\vee}
\newunicodechar{∩}{\cap}
\newunicodechar{∪}{\cup}
\newunicodechar{∫}{\int}
\newunicodechar{∬}{\iint}
\newunicodechar{∭}{\iiint}
\newunicodechar{∮}{\oint}
\newunicodechar{∯}{\oiint}
\newunicodechar{∰}{\oiiint}
\newunicodechar{∱}{\intclockwise}
\newunicodechar{∲}{\varointclockwise}
\newunicodechar{∳}{\ointctrclockwise}
\newunicodechar{∴}{\therefore}
\newunicodechar{∵}{\because}
\newunicodechar{∶}{\mathratio}
\newunicodechar{∷}{\Colon}
\newunicodechar{∸}{\dotminus}
\newunicodechar{∹}{\dashcolon}
\newunicodechar{∺}{\dotsminusdots}
\newunicodechar{∻}{\kernelcontraction}
\newunicodechar{∼}{\sim}
\newunicodechar{∽}{\backsim}
\newunicodechar{∾}{\invlazys}
\newunicodechar{∿}{\sinewave}
\newunicodechar{≀}{\wr}
\newunicodechar{≁}{\nsim}
\newunicodechar{≂}{\eqsim}
\newunicodechar{≃}{\simeq}
\newunicodechar{≄}{\nsime}
\newunicodechar{≅}{\cong}
\newunicodechar{≆}{\simneqq}
\newunicodechar{≇}{\ncong}
\newunicodechar{≈}{\approx}
\newunicodechar{≉}{\napprox}
\newunicodechar{≊}{\approxeq}
\newunicodechar{≋}{\approxident}
\newunicodechar{≌}{\backcong}
\newunicodechar{≍}{\asymp}
\newunicodechar{≎}{\Bumpeq}
\newunicodechar{≏}{\bumpeq}
\newunicodechar{≐}{\doteq}
\newunicodechar{≑}{\Doteq}
\newunicodechar{≒}{\fallingdotseq}
\newunicodechar{≓}{\risingdotseq}
\newunicodechar{≔}{\coloneq}
\newunicodechar{≕}{\eqcolon}
\newunicodechar{≖}{\eqcirc}
\newunicodechar{≗}{\circeq}
\newunicodechar{≘}{\arceq}
\newunicodechar{≙}{\wedgeq}
\newunicodechar{≚}{\veeeq}
\newunicodechar{≛}{\stareq}
\newunicodechar{≜}{\triangleq}
\newunicodechar{≝}{\eqdef}
\newunicodechar{≞}{\measeq}
\newunicodechar{≟}{\questeq}
\newunicodechar{≠}{\ne}
\newunicodechar{≡}{\equiv}
\newunicodechar{≢}{\nequiv}
\newunicodechar{≣}{\Equiv}
\newunicodechar{≤}{\leq}
\newunicodechar{≥}{\geq}
\newunicodechar{≦}{\leqq}
\newunicodechar{≧}{\geqq}
\newunicodechar{≨}{\lneqq}
\newunicodechar{≩}{\gneqq}
\newunicodechar{≪}{\ll}
\newunicodechar{≫}{\gg}
\newunicodechar{≬}{\between}
\newunicodechar{≭}{\nasymp}
\newunicodechar{≮}{\nless}
\newunicodechar{≯}{\ngtr}
\newunicodechar{≰}{\nleq}
\newunicodechar{≱}{\ngeq}
\newunicodechar{≲}{\lesssim}
\newunicodechar{≳}{\gtrsim}
\newunicodechar{≴}{\nlesssim}
\newunicodechar{≵}{\ngtrsim}
\newunicodechar{≶}{\lessgtr}
\newunicodechar{≷}{\gtrless}
\newunicodechar{≸}{\nlessgtr}
\newunicodechar{≹}{\ngtrless}
\newunicodechar{≺}{\prec}
\newunicodechar{≻}{\succ}
\newunicodechar{≼}{\preccurlyeq}
\newunicodechar{≽}{\succcurlyeq}
\newunicodechar{≾}{\precsim}
\newunicodechar{≿}{\succsim}
\newunicodechar{⊀}{\nprec}
\newunicodechar{⊁}{\nsucc}
\newunicodechar{⊂}{\subset}
\newunicodechar{⊃}{\supset}
\newunicodechar{⊄}{\nsubset}
\newunicodechar{⊅}{\nsupset}
\newunicodechar{⊆}{\subseteq}
\newunicodechar{⊇}{\supseteq}
\newunicodechar{⊈}{\nsubseteq}
\newunicodechar{⊉}{\nsupseteq}
\newunicodechar{⊊}{\subsetneq}
\newunicodechar{⊋}{\supsetneq}
\newunicodechar{⊌}{\cupleftarrow}
\newunicodechar{⊍}{\cupdot}
\newunicodechar{⊎}{\uplus}
\newunicodechar{⊏}{\sqsubset}
\newunicodechar{⊐}{\sqsupset}
\newunicodechar{⊑}{\sqsubseteq}
\newunicodechar{⊒}{\sqsupseteq}
\newunicodechar{⊓}{\sqcap}
\newunicodechar{⊔}{\sqcup}
\newunicodechar{⊕}{\oplus}
\newunicodechar{⊖}{\ominus}
\newunicodechar{⊗}{\otimes}
\newunicodechar{⊘}{\oslash}
\newunicodechar{⊙}{\odot}
\newunicodechar{⊚}{\circledcirc}
\newunicodechar{⊛}{\circledast}
\newunicodechar{⊜}{\circledequal}
\newunicodechar{⊝}{\circleddash}
\newunicodechar{⊞}{\boxplus}
\newunicodechar{⊟}{\boxminus}
\newunicodechar{⊠}{\boxtimes}
\newunicodechar{⊡}{\boxdot}
\newunicodechar{⊢}{\vdash}
\newunicodechar{⊣}{\dashv}
\newunicodechar{⊤}{\top}
\newunicodechar{⊥}{\bot}
\newunicodechar{⊦}{\assert}
\newunicodechar{⊧}{\models}
\newunicodechar{⊨}{\vDash}
\newunicodechar{⊩}{\Vdash}
\newunicodechar{⊪}{\Vvdash}
\newunicodechar{⊫}{\VDash}
\newunicodechar{⊬}{\nvdash}
\newunicodechar{⊭}{\nvDash}
\newunicodechar{⊮}{\nVdash}
\newunicodechar{⊯}{\nVDash}
\newunicodechar{⊰}{\prurel}
\newunicodechar{⊱}{\scurel}
\newunicodechar{⊲}{\vartriangleleft}
\newunicodechar{⊳}{\vartriangleright}
\newunicodechar{⊴}{\trianglelefteq}
\newunicodechar{⊵}{\trianglerighteq}
\newunicodechar{⊶}{\origof}
\newunicodechar{⊷}{\imageof}
\newunicodechar{⊸}{\multimap}
\newunicodechar{⊹}{\hermitmatrix}
\newunicodechar{⊺}{\intercal}
\newunicodechar{⊻}{\veebar}
\newunicodechar{⊼}{\barwedge}
\newunicodechar{⊽}{\barvee}
\newunicodechar{⊾}{\measuredrightangle}
\newunicodechar{⊿}{\varlrtriangle}
\newunicodechar{⋀}{\bigwedge}
\newunicodechar{⋁}{\bigvee}
\newunicodechar{⋂}{\bigcap}
\newunicodechar{⋃}{\bigcup}
\newunicodechar{⋄}{\smwhtdiamond}
\newunicodechar{⋅}{\cdot}
\newunicodechar{⋆}{\star}
\newunicodechar{⋇}{\divideontimes}
\newunicodechar{⋈}{\bowtie}
\newunicodechar{⋉}{\ltimes}
\newunicodechar{⋊}{\rtimes}
\newunicodechar{⋋}{\leftthreetimes}
\newunicodechar{⋌}{\rightthreetimes}
\newunicodechar{⋍}{\backsimeq}
\newunicodechar{⋎}{\curlyvee}
\newunicodechar{⋏}{\curlywedge}
\newunicodechar{⋐}{\Subset}
\newunicodechar{⋑}{\Supset}
\newunicodechar{⋒}{\Cap}
\newunicodechar{⋓}{\Cup}
\newunicodechar{⋔}{\pitchfork}
\newunicodechar{⋕}{\equalparallel}
\newunicodechar{⋖}{\lessdot}
\newunicodechar{⋗}{\gtrdot}
\newunicodechar{⋘}{\lll}
\newunicodechar{⋙}{\ggg}
\newunicodechar{⋚}{\lesseqgtr}
\newunicodechar{⋛}{\gtreqless}
\newunicodechar{⋜}{\eqless}
\newunicodechar{⋝}{\eqgtr}
\newunicodechar{⋞}{\curlyeqprec}
\newunicodechar{⋟}{\curlyeqsucc}
\newunicodechar{⋠}{\npreccurlyeq}
\newunicodechar{⋡}{\nsucccurlyeq}
\newunicodechar{⋢}{\nsqsubseteq}
\newunicodechar{⋣}{\nsqsupseteq}
\newunicodechar{⋤}{\sqsubsetneq}
\newunicodechar{⋥}{\sqsupsetneq}
\newunicodechar{⋦}{\lnsim}
\newunicodechar{⋧}{\gnsim}
\newunicodechar{⋨}{\precnsim}
\newunicodechar{⋩}{\succnsim}
\newunicodechar{⋪}{\nvartriangleleft}
\newunicodechar{⋫}{\nvartriangleright}
\newunicodechar{⋬}{\ntrianglelefteq}
\newunicodechar{⋭}{\ntrianglerighteq}
\newunicodechar{⋮}{\vdots}
\newunicodechar{⋯}{\unicodecdots}
\newunicodechar{⋰}{\adots}
\newunicodechar{⋱}{\ddots}
\newunicodechar{⋲}{\disin}
\newunicodechar{⋳}{\varisins}
\newunicodechar{⋴}{\isins}
\newunicodechar{⋵}{\isindot}
\newunicodechar{⋶}{\varisinobar}
\newunicodechar{⋷}{\isinobar}
\newunicodechar{⋸}{\isinvb}
\newunicodechar{⋹}{\isinE}
\newunicodechar{⋺}{\nisd}
\newunicodechar{⋻}{\varnis}
\newunicodechar{⋼}{\nis}
\newunicodechar{⋽}{\varniobar}
\newunicodechar{⋾}{\niobar}
\newunicodechar{⋿}{\bagmember}
\newunicodechar{⌀}{\diameter}
\newunicodechar{⌂}{\house}
\newunicodechar{⌅}{\varbarwedge}
\newunicodechar{⌆}{\vardoublebarwedge}
\newunicodechar{⌈}{\lceil}
\newunicodechar{⌉}{\rceil}
\newunicodechar{⌊}{\lfloor}
\newunicodechar{⌋}{\rfloor}
\newunicodechar{⌐}{\invnot}
\newunicodechar{⌑}{\sqlozenge}
\newunicodechar{⌒}{\profline}
\newunicodechar{⌓}{\profsurf}
\newunicodechar{⌗}{\viewdata}
\newunicodechar{⌙}{\turnednot}
\newunicodechar{⌜}{\ulcorner}
\newunicodechar{⌝}{\urcorner}
\newunicodechar{⌞}{\llcorner}
\newunicodechar{⌟}{\lrcorner}
\newunicodechar{⌠}{\inttop}
\newunicodechar{⌡}{\intbottom}
\newunicodechar{⌢}{\frown}
\newunicodechar{⌣}{\smile}
\newunicodechar{⌬}{\varhexagonlrbonds}
\newunicodechar{⌲}{\conictaper}
\newunicodechar{⌶}{\topbot}
\newunicodechar{⌽}{\obar}
\newunicodechar{⌿}{\APLnotslash}
\newunicodechar{⍀}{\APLnotbackslash}
\newunicodechar{⍓}{\APLboxupcaret}
\newunicodechar{⍰}{\APLboxquestion}
\newunicodechar{⍼}{\rangledownzigzagarrow}
\newunicodechar{⎔}{\hexagon}
\newunicodechar{⎛}{\lparenuend}
\newunicodechar{⎜}{\lparenextender}
\newunicodechar{⎝}{\lparenlend}
\newunicodechar{⎞}{\rparenuend}
\newunicodechar{⎟}{\rparenextender}
\newunicodechar{⎠}{\rparenlend}
\newunicodechar{⎡}{\lbrackuend}
\newunicodechar{⎢}{\lbrackextender}
\newunicodechar{⎣}{\lbracklend}
\newunicodechar{⎤}{\rbrackuend}
\newunicodechar{⎥}{\rbrackextender}
\newunicodechar{⎦}{\rbracklend}
\newunicodechar{⎧}{\lbraceuend}
\newunicodechar{⎨}{\lbracemid}
\newunicodechar{⎩}{\lbracelend}
\newunicodechar{⎪}{\vbraceextender}
\newunicodechar{⎫}{\rbraceuend}
\newunicodechar{⎬}{\rbracemid}
\newunicodechar{⎭}{\rbracelend}
\newunicodechar{⎮}{\intextender}
\newunicodechar{⎯}{\harrowextender}
\newunicodechar{⎰}{\lmoustache}
\newunicodechar{⎱}{\rmoustache}
\newunicodechar{⎲}{\sumtop}
\newunicodechar{⎳}{\sumbottom}
\newunicodechar{⎴}{\overbracket}
\newunicodechar{⎵}{\underbracket}
\newunicodechar{⎶}{\bbrktbrk}
\newunicodechar{⎷}{\sqrtbottom}
\newunicodechar{⎸}{\lvboxline}
\newunicodechar{⎹}{\rvboxline}
\newunicodechar{⏎}{\varcarriagereturn}
\newunicodechar{⏜}{\overparen}
\newunicodechar{⏝}{\underparen}
\newunicodechar{⏞}{\overbrace}
\newunicodechar{⏟}{\underbrace}
\newunicodechar{⏠}{\obrbrak}
\newunicodechar{⏡}{\ubrbrak}
\newunicodechar{⏢}{\trapezium}
\newunicodechar{⏣}{\benzenr}
\newunicodechar{⏤}{\strns}
\newunicodechar{⏥}{\fltns}
\newunicodechar{⏦}{\accurrent}
\newunicodechar{⏧}{\elinters}
\newunicodechar{␢}{\blanksymbol}
\newunicodechar{␣}{\mathvisiblespace}
\newunicodechar{┆}{\bdtriplevdash}
\newunicodechar{▀}{\blockuphalf}
\newunicodechar{▄}{\blocklowhalf}
\newunicodechar{█}{\blockfull}
\newunicodechar{▌}{\blocklefthalf}
\newunicodechar{▐}{\blockrighthalf}
\newunicodechar{░}{\blockqtrshaded}
\newunicodechar{▒}{\blockhalfshaded}
\newunicodechar{▓}{\blockthreeqtrshaded}
\newunicodechar{■}{\mdlgblksquare}
\newunicodechar{□}{\mdlgwhtsquare}
\newunicodechar{▢}{\squoval}
\newunicodechar{▣}{\blackinwhitesquare}
\newunicodechar{▤}{\squarehfill}
\newunicodechar{▥}{\squarevfill}
\newunicodechar{▦}{\squarehvfill}
\newunicodechar{▧}{\squarenwsefill}
\newunicodechar{▨}{\squareneswfill}
\newunicodechar{▩}{\squarecrossfill}
\newunicodechar{▪}{\smblksquare}
\newunicodechar{▫}{\smwhtsquare}
\newunicodechar{▬}{\hrectangleblack}
\newunicodechar{▭}{\hrectangle}
\newunicodechar{▮}{\vrectangleblack}
\newunicodechar{▯}{\vrectangle}
\newunicodechar{▰}{\parallelogramblack}
\newunicodechar{▱}{\parallelogram}
\newunicodechar{▲}{\bigblacktriangleup}
\newunicodechar{△}{\bigtriangleup}
\newunicodechar{▴}{\blacktriangle}
\newunicodechar{▵}{\vartriangle}
\newunicodechar{▶}{\blacktriangleright}
\newunicodechar{▷}{\triangleright}
\newunicodechar{▸}{\smallblacktriangleright}
\newunicodechar{▹}{\smalltriangleright}
\newunicodechar{►}{\blackpointerright}
\newunicodechar{▻}{\whitepointerright}
\newunicodechar{▼}{\bigblacktriangledown}
\newunicodechar{▽}{\bigtriangledown}
\newunicodechar{▾}{\blacktriangledown}
\newunicodechar{▿}{\triangledown}
\newunicodechar{◀}{\blacktriangleleft}
\newunicodechar{◁}{\triangleleft}
\newunicodechar{◂}{\smallblacktriangleleft}
\newunicodechar{◃}{\smalltriangleleft}
\newunicodechar{◄}{\blackpointerleft}
\newunicodechar{◅}{\whitepointerleft}
\newunicodechar{◆}{\mdlgblkdiamond}
\newunicodechar{◇}{\mdlgwhtdiamond}
\newunicodechar{◈}{\blackinwhitediamond}
\newunicodechar{◉}{\fisheye}
\newunicodechar{◊}{\mdlgwhtlozenge}
\newunicodechar{○}{\mdlgwhtcircle}
\newunicodechar{◌}{\dottedcircle}
\newunicodechar{◍}{\circlevertfill}
\newunicodechar{◎}{\bullseye}
\newunicodechar{●}{\mdlgblkcircle}
\newunicodechar{◐}{\circlelefthalfblack}
\newunicodechar{◑}{\circlerighthalfblack}
\newunicodechar{◒}{\circlebottomhalfblack}
\newunicodechar{◓}{\circletophalfblack}
\newunicodechar{◔}{\circleurquadblack}
\newunicodechar{◕}{\blackcircleulquadwhite}
\newunicodechar{◖}{\blacklefthalfcircle}
\newunicodechar{◗}{\blackrighthalfcircle}
\newunicodechar{◘}{\inversebullet}
\newunicodechar{◙}{\inversewhitecircle}
\newunicodechar{◚}{\invwhiteupperhalfcircle}
\newunicodechar{◛}{\invwhitelowerhalfcircle}
\newunicodechar{◜}{\ularc}
\newunicodechar{◝}{\urarc}
\newunicodechar{◞}{\lrarc}
\newunicodechar{◟}{\llarc}
\newunicodechar{◠}{\topsemicircle}
\newunicodechar{◡}{\botsemicircle}
\newunicodechar{◢}{\lrblacktriangle}
\newunicodechar{◣}{\llblacktriangle}
\newunicodechar{◤}{\ulblacktriangle}
\newunicodechar{◥}{\urblacktriangle}
\newunicodechar{◦}{\smwhtcircle}
\newunicodechar{◧}{\squareleftblack}
\newunicodechar{◨}{\squarerightblack}
\newunicodechar{◩}{\squareulblack}
\newunicodechar{◪}{\squarelrblack}
\newunicodechar{◫}{\boxbar}
\newunicodechar{◬}{\trianglecdot}
\newunicodechar{◭}{\triangleleftblack}
\newunicodechar{◮}{\trianglerightblack}
\newunicodechar{◯}{\lgwhtcircle}
\newunicodechar{◰}{\squareulquad}
\newunicodechar{◱}{\squarellquad}
\newunicodechar{◲}{\squarelrquad}
\newunicodechar{◳}{\squareurquad}
\newunicodechar{◴}{\circleulquad}
\newunicodechar{◵}{\circlellquad}
\newunicodechar{◶}{\circlelrquad}
\newunicodechar{◷}{\circleurquad}
\newunicodechar{◸}{\ultriangle}
\newunicodechar{◹}{\urtriangle}
\newunicodechar{◺}{\lltriangle}
\newunicodechar{◻}{\mdwhtsquare}
\newunicodechar{◼}{\mdblksquare}
\newunicodechar{◽}{\mdsmwhtsquare}
\newunicodechar{◾}{\mdsmblksquare}
\newunicodechar{◿}{\lrtriangle}
\newunicodechar{★}{\bigstar}
\newunicodechar{☆}{\bigwhitestar}
\newunicodechar{☉}{\astrosun}
\newunicodechar{☡}{\danger}
\newunicodechar{☻}{\blacksmiley}
\newunicodechar{☼}{\sun}
\newunicodechar{☽}{\rightmoon}
\newunicodechar{☾}{\leftmoon}
\newunicodechar{♀}{\female}
\newunicodechar{♂}{\male}
\newunicodechar{♠}{\spadesuit}
\newunicodechar{♡}{\heartsuit}
\newunicodechar{♢}{\diamondsuit}
\newunicodechar{♣}{\clubsuit}
\newunicodechar{♤}{\varspadesuit}
\newunicodechar{♥}{\varheartsuit}
\newunicodechar{♦}{\vardiamondsuit}
\newunicodechar{♧}{\varclubsuit}
\newunicodechar{♩}{\quarternote}
\newunicodechar{♪}{\eighthnote}
\newunicodechar{♫}{\twonotes}
\newunicodechar{♭}{\flat}
\newunicodechar{♮}{\natural}
\newunicodechar{♯}{\sharp}
\newunicodechar{♾}{\acidfree}
\newunicodechar{⚀}{\dicei}
\newunicodechar{⚁}{\diceii}
\newunicodechar{⚂}{\diceiii}
\newunicodechar{⚃}{\diceiv}
\newunicodechar{⚄}{\dicev}
\newunicodechar{⚅}{\dicevi}
\newunicodechar{⚆}{\circledrightdot}
\newunicodechar{⚇}{\circledtwodots}
\newunicodechar{⚈}{\blackcircledrightdot}
\newunicodechar{⚉}{\blackcircledtwodots}
\newunicodechar{⚥}{\Hermaphrodite}
\newunicodechar{⚪}{\mdwhtcircle}
\newunicodechar{⚫}{\mdblkcircle}
\newunicodechar{⚬}{\mdsmwhtcircle}
\newunicodechar{⚲}{\neuter}
\newunicodechar{✓}{\checkmark}
\newunicodechar{✠}{\maltese}
\newunicodechar{✪}{\circledstar}
\newunicodechar{✶}{\varstar}
\newunicodechar{✽}{\dingasterisk}
\newunicodechar{❲}{\lbrbrak}
\newunicodechar{❳}{\rbrbrak}
\newunicodechar{➛}{\draftingarrow}
\newunicodechar{⟀}{\threedangle}
\newunicodechar{⟁}{\whiteinwhitetriangle}
\newunicodechar{⟂}{\perp}
\newunicodechar{⟃}{\subsetcirc}
\newunicodechar{⟄}{\supsetcirc}
\newunicodechar{⟅}{\lbag}
\newunicodechar{⟆}{\rbag}
\newunicodechar{⟇}{\veedot}
\newunicodechar{⟈}{\bsolhsub}
\newunicodechar{⟉}{\suphsol}
\newunicodechar{⟌}{\longdivision}
\newunicodechar{⟐}{\diamondcdot}
\newunicodechar{⟑}{\wedgedot}
\newunicodechar{⟒}{\upin}
\newunicodechar{⟓}{\pullback}
\newunicodechar{⟔}{\pushout}
\newunicodechar{⟕}{\leftouterjoin}
\newunicodechar{⟖}{\rightouterjoin}
\newunicodechar{⟗}{\fullouterjoin}
\newunicodechar{⟘}{\bigbot}
\newunicodechar{⟙}{\bigtop}
\newunicodechar{⟚}{\DashVDash}
\newunicodechar{⟛}{\dashVdash}
\newunicodechar{⟜}{\multimapinv}
\newunicodechar{⟝}{\vlongdash}
\newunicodechar{⟞}{\longdashv}
\newunicodechar{⟟}{\cirbot}
\newunicodechar{⟠}{\lozengeminus}
\newunicodechar{⟡}{\concavediamond}
\newunicodechar{⟢}{\concavediamondtickleft}
\newunicodechar{⟣}{\concavediamondtickright}
\newunicodechar{⟤}{\whitesquaretickleft}
\newunicodechar{⟥}{\whitesquaretickright}
\newunicodechar{⟦}{\lBrack}
\newunicodechar{⟧}{\rBrack}
\newunicodechar{⟨}{\langle}
\newunicodechar{⟩}{\rangle}
\newunicodechar{⟪}{\lAngle}
\newunicodechar{⟫}{\rAngle}
\newunicodechar{⟬}{\Lbrbrak}
\newunicodechar{⟭}{\Rbrbrak}
\newunicodechar{⟮}{\lgroup}
\newunicodechar{⟯}{\rgroup}
\newunicodechar{⟰}{\UUparrow}
\newunicodechar{⟱}{\DDownarrow}
\newunicodechar{⟲}{\acwgapcirclearrow}
\newunicodechar{⟳}{\cwgapcirclearrow}
\newunicodechar{⟴}{\rightarrowonoplus}
\newunicodechar{⟵}{\longleftarrow}
\newunicodechar{⟶}{\longrightarrow}
\newunicodechar{⟷}{\longleftrightarrow}
\newunicodechar{⟸}{\Longleftarrow}
\newunicodechar{⟹}{\Longrightarrow}
\newunicodechar{⟺}{\Longleftrightarrow}
\newunicodechar{⟻}{\longmapsfrom}
\newunicodechar{⟼}{\longmapsto}
\newunicodechar{⟽}{\Longmapsfrom}
\newunicodechar{⟾}{\Longmapsto}
\newunicodechar{⟿}{\longrightsquigarrow}
\newunicodechar{⤀}{\nvtwoheadrightarrow}
\newunicodechar{⤁}{\nVtwoheadrightarrow}
\newunicodechar{⤂}{\nvLeftarrow}
\newunicodechar{⤃}{\nvRightarrow}
\newunicodechar{⤄}{\nvLeftrightarrow}
\newunicodechar{⤅}{\twoheadmapsto}
\newunicodechar{⤆}{\Mapsfrom}
\newunicodechar{⤇}{\Mapsto}
\newunicodechar{⤈}{\downarrowbarred}
\newunicodechar{⤉}{\uparrowbarred}
\newunicodechar{⤊}{\Uuparrow}
\newunicodechar{⤋}{\Ddownarrow}
\newunicodechar{⤌}{\leftbkarrow}
\newunicodechar{⤍}{\rightbkarrow}
\newunicodechar{⤎}{\leftdbkarrow}
\newunicodechar{⤏}{\dbkarow}
\newunicodechar{⤐}{\drbkarow}
\newunicodechar{⤑}{\rightdotarrow}
\newunicodechar{⤒}{\baruparrow}
\newunicodechar{⤓}{\downarrowbar}
\newunicodechar{⤔}{\nvrightarrowtail}
\newunicodechar{⤕}{\nVrightarrowtail}
\newunicodechar{⤖}{\twoheadrightarrowtail}
\newunicodechar{⤗}{\nvtwoheadrightarrowtail}
\newunicodechar{⤘}{\nVtwoheadrightarrowtail}
\newunicodechar{⤙}{\lefttail}
\newunicodechar{⤚}{\righttail}
\newunicodechar{⤛}{\leftdbltail}
\newunicodechar{⤜}{\rightdbltail}
\newunicodechar{⤝}{\diamondleftarrow}
\newunicodechar{⤞}{\rightarrowdiamond}
\newunicodechar{⤟}{\diamondleftarrowbar}
\newunicodechar{⤠}{\barrightarrowdiamond}
\newunicodechar{⤡}{\nwsearrow}
\newunicodechar{⤢}{\neswarrow}
\newunicodechar{⤣}{\hknwarrow}
\newunicodechar{⤤}{\hknearrow}
\newunicodechar{⤥}{\hksearow}
\newunicodechar{⤦}{\hkswarow}
\newunicodechar{⤧}{\tona}
\newunicodechar{⤨}{\toea}
\newunicodechar{⤩}{\tosa}
\newunicodechar{⤪}{\towa}
\newunicodechar{⤫}{\rdiagovfdiag}
\newunicodechar{⤬}{\fdiagovrdiag}
\newunicodechar{⤭}{\seovnearrow}
\newunicodechar{⤮}{\neovsearrow}
\newunicodechar{⤯}{\fdiagovnearrow}
\newunicodechar{⤰}{\rdiagovsearrow}
\newunicodechar{⤱}{\neovnwarrow}
\newunicodechar{⤲}{\nwovnearrow}
\newunicodechar{⤳}{\rightcurvedarrow}
\newunicodechar{⤴}{\uprightcurvearrow}
\newunicodechar{⤵}{\downrightcurvedarrow}
\newunicodechar{⤶}{\leftdowncurvedarrow}
\newunicodechar{⤷}{\rightdowncurvedarrow}
\newunicodechar{⤸}{\cwrightarcarrow}
\newunicodechar{⤹}{\acwleftarcarrow}
\newunicodechar{⤺}{\acwoverarcarrow}
\newunicodechar{⤻}{\acwunderarcarrow}
\newunicodechar{⤼}{\curvearrowrightminus}
\newunicodechar{⤽}{\curvearrowleftplus}
\newunicodechar{⤾}{\cwundercurvearrow}
\newunicodechar{⤿}{\ccwundercurvearrow}
\newunicodechar{⥀}{\acwcirclearrow}
\newunicodechar{⥁}{\cwcirclearrow}
\newunicodechar{⥂}{\rightarrowshortleftarrow}
\newunicodechar{⥃}{\leftarrowshortrightarrow}
\newunicodechar{⥄}{\shortrightarrowleftarrow}
\newunicodechar{⥅}{\rightarrowplus}
\newunicodechar{⥆}{\leftarrowplus}
\newunicodechar{⥇}{\rightarrowx}
\newunicodechar{⥈}{\leftrightarrowcircle}
\newunicodechar{⥉}{\twoheaduparrowcircle}
\newunicodechar{⥊}{\leftrightharpoonupdown}
\newunicodechar{⥋}{\leftrightharpoondownup}
\newunicodechar{⥌}{\updownharpoonrightleft}
\newunicodechar{⥍}{\updownharpoonleftright}
\newunicodechar{⥎}{\leftrightharpoonupup}
\newunicodechar{⥏}{\updownharpoonrightright}
\newunicodechar{⥐}{\leftrightharpoondowndown}
\newunicodechar{⥑}{\updownharpoonleftleft}
\newunicodechar{⥒}{\barleftharpoonup}
\newunicodechar{⥓}{\rightharpoonupbar}
\newunicodechar{⥔}{\barupharpoonright}
\newunicodechar{⥕}{\downharpoonrightbar}
\newunicodechar{⥖}{\barleftharpoondown}
\newunicodechar{⥗}{\rightharpoondownbar}
\newunicodechar{⥘}{\barupharpoonleft}
\newunicodechar{⥙}{\downharpoonleftbar}
\newunicodechar{⥚}{\leftharpoonupbar}
\newunicodechar{⥛}{\barrightharpoonup}
\newunicodechar{⥜}{\upharpoonrightbar}
\newunicodechar{⥝}{\bardownharpoonright}
\newunicodechar{⥞}{\leftharpoondownbar}
\newunicodechar{⥟}{\barrightharpoondown}
\newunicodechar{⥠}{\upharpoonleftbar}
\newunicodechar{⥡}{\bardownharpoonleft}
\newunicodechar{⥢}{\leftharpoonsupdown}
\newunicodechar{⥣}{\upharpoonsleftright}
\newunicodechar{⥤}{\rightharpoonsupdown}
\newunicodechar{⥥}{\downharpoonsleftright}
\newunicodechar{⥦}{\leftrightharpoonsup}
\newunicodechar{⥧}{\leftrightharpoonsdown}
\newunicodechar{⥨}{\rightleftharpoonsup}
\newunicodechar{⥩}{\rightleftharpoonsdown}
\newunicodechar{⥪}{\leftharpoonupdash}
\newunicodechar{⥫}{\dashleftharpoondown}
\newunicodechar{⥬}{\rightharpoonupdash}
\newunicodechar{⥭}{\dashrightharpoondown}
\newunicodechar{⥮}{\updownharpoonsleftright}
\newunicodechar{⥯}{\downupharpoonsleftright}
\newunicodechar{⥰}{\rightimply}
\newunicodechar{⥱}{\equalrightarrow}
\newunicodechar{⥲}{\similarrightarrow}
\newunicodechar{⥳}{\leftarrowsimilar}
\newunicodechar{⥴}{\rightarrowsimilar}
\newunicodechar{⥵}{\rightarrowapprox}
\newunicodechar{⥶}{\ltlarr}
\newunicodechar{⥷}{\leftarrowless}
\newunicodechar{⥸}{\gtrarr}
\newunicodechar{⥹}{\subrarr}
\newunicodechar{⥺}{\leftarrowsubset}
\newunicodechar{⥻}{\suplarr}
\newunicodechar{⥼}{\leftfishtail}
\newunicodechar{⥽}{\rightfishtail}
\newunicodechar{⥾}{\upfishtail}
\newunicodechar{⥿}{\downfishtail}
\newunicodechar{⦀}{\Vvert}
\newunicodechar{⦁}{\mdsmblkcircle}
\newunicodechar{⦂}{\typecolon}
\newunicodechar{⦃}{\lBrace}
\newunicodechar{⦄}{\rBrace}
\newunicodechar{⦅}{\lParen}
\newunicodechar{⦆}{\rParen}
\newunicodechar{⦇}{\llparenthesis}
\newunicodechar{⦈}{\rrparenthesis}
\newunicodechar{⦉}{\llangle}
\newunicodechar{⦊}{\rrangle}
\newunicodechar{⦋}{\lbrackubar}
\newunicodechar{⦌}{\rbrackubar}
\newunicodechar{⦍}{\lbrackultick}
\newunicodechar{⦎}{\rbracklrtick}
\newunicodechar{⦏}{\lbracklltick}
\newunicodechar{⦐}{\rbrackurtick}
\newunicodechar{⦑}{\langledot}
\newunicodechar{⦒}{\rangledot}
\newunicodechar{⦓}{\lparenless}
\newunicodechar{⦔}{\rparengtr}
\newunicodechar{⦕}{\Lparengtr}
\newunicodechar{⦖}{\Rparenless}
\newunicodechar{⦗}{\lblkbrbrak}
\newunicodechar{⦘}{\rblkbrbrak}
\newunicodechar{⦙}{\fourvdots}
\newunicodechar{⦚}{\vzigzag}
\newunicodechar{⦛}{\measuredangleleft}
\newunicodechar{⦜}{\rightanglesqr}
\newunicodechar{⦝}{\rightanglemdot}
\newunicodechar{⦞}{\angles}
\newunicodechar{⦟}{\angdnr}
\newunicodechar{⦠}{\gtlpar}
\newunicodechar{⦡}{\sphericalangleup}
\newunicodechar{⦢}{\turnangle}
\newunicodechar{⦣}{\revangle}
\newunicodechar{⦤}{\angleubar}
\newunicodechar{⦥}{\revangleubar}
\newunicodechar{⦦}{\wideangledown}
\newunicodechar{⦧}{\wideangleup}
\newunicodechar{⦨}{\measanglerutone}
\newunicodechar{⦩}{\measanglelutonw}
\newunicodechar{⦪}{\measanglerdtose}
\newunicodechar{⦫}{\measangleldtosw}
\newunicodechar{⦬}{\measangleurtone}
\newunicodechar{⦭}{\measangleultonw}
\newunicodechar{⦮}{\measangledrtose}
\newunicodechar{⦯}{\measangledltosw}
\newunicodechar{⦰}{\revemptyset}
\newunicodechar{⦱}{\emptysetobar}
\newunicodechar{⦲}{\emptysetocirc}
\newunicodechar{⦳}{\emptysetoarr}
\newunicodechar{⦴}{\emptysetoarrl}
\newunicodechar{⦵}{\circlehbar}
\newunicodechar{⦶}{\circledvert}
\newunicodechar{⦷}{\circledparallel}
\newunicodechar{⦸}{\obslash}
\newunicodechar{⦹}{\operp}
\newunicodechar{⦺}{\obot}
\newunicodechar{⦻}{\olcross}
\newunicodechar{⦼}{\odotslashdot}
\newunicodechar{⦽}{\uparrowoncircle}
\newunicodechar{⦾}{\circledwhitebullet}
\newunicodechar{⦿}{\circledbullet}
\newunicodechar{⧀}{\olessthan}
\newunicodechar{⧁}{\ogreaterthan}
\newunicodechar{⧂}{\cirscir}
\newunicodechar{⧃}{\cirE}
\newunicodechar{⧄}{\boxdiag}
\newunicodechar{⧅}{\boxbslash}
\newunicodechar{⧆}{\boxast}
\newunicodechar{⧇}{\boxcircle}
\newunicodechar{⧈}{\boxbox}
\newunicodechar{⧉}{\boxonbox}
\newunicodechar{⧊}{\triangleodot}
\newunicodechar{⧋}{\triangleubar}
\newunicodechar{⧌}{\triangles}
\newunicodechar{⧍}{\triangleserifs}
\newunicodechar{⧎}{\rtriltri}
\newunicodechar{⧏}{\ltrivb}
\newunicodechar{⧐}{\vbrtri}
\newunicodechar{⧑}{\lfbowtie}
\newunicodechar{⧒}{\rfbowtie}
\newunicodechar{⧓}{\fbowtie}
\newunicodechar{⧔}{\lftimes}
\newunicodechar{⧕}{\rftimes}
\newunicodechar{⧖}{\hourglass}
\newunicodechar{⧗}{\blackhourglass}
\newunicodechar{⧘}{\lvzigzag}
\newunicodechar{⧙}{\rvzigzag}
\newunicodechar{⧚}{\Lvzigzag}
\newunicodechar{⧛}{\Rvzigzag}
\newunicodechar{⧜}{\iinfin}
\newunicodechar{⧝}{\tieinfty}
\newunicodechar{⧞}{\nvinfty}
\newunicodechar{⧟}{\dualmap}
\newunicodechar{⧠}{\laplac}
\newunicodechar{⧡}{\lrtriangleeq}
\newunicodechar{⧢}{\shuffle}
\newunicodechar{⧣}{\eparsl}
\newunicodechar{⧤}{\smeparsl}
\newunicodechar{⧥}{\eqvparsl}
\newunicodechar{⧦}{\gleichstark}
\newunicodechar{⧧}{\thermod}
\newunicodechar{⧨}{\downtriangleleftblack}
\newunicodechar{⧩}{\downtrianglerightblack}
\newunicodechar{⧪}{\blackdiamonddownarrow}
\newunicodechar{⧫}{\mdlgblklozenge}
\newunicodechar{⧬}{\circledownarrow}
\newunicodechar{⧭}{\blackcircledownarrow}
\newunicodechar{⧮}{\errbarsquare}
\newunicodechar{⧯}{\errbarblacksquare}
\newunicodechar{⧰}{\errbardiamond}
\newunicodechar{⧱}{\errbarblackdiamond}
\newunicodechar{⧲}{\errbarcircle}
\newunicodechar{⧳}{\errbarblackcircle}
\newunicodechar{⧴}{\ruledelayed}
\newunicodechar{⧵}{\setminus}
\newunicodechar{⧶}{\dsol}
\newunicodechar{⧷}{\rsolbar}
\newunicodechar{⧸}{\xsol}
\newunicodechar{⧹}{\xbsol}
\newunicodechar{⧺}{\doubleplus}
\newunicodechar{⧻}{\tripleplus}
\newunicodechar{⧼}{\lcurvyangle}
\newunicodechar{⧽}{\rcurvyangle}
\newunicodechar{⧾}{\tplus}
\newunicodechar{⧿}{\tminus}
\newunicodechar{⨀}{\bigodot}
\newunicodechar{⨁}{\bigoplus}
\newunicodechar{⨂}{\bigotimes}
\newunicodechar{⨃}{\bigcupdot}
\newunicodechar{⨄}{\biguplus}
\newunicodechar{⨅}{\bigsqcap}
\newunicodechar{⨆}{\bigsqcup}
\newunicodechar{⨇}{\conjquant}
\newunicodechar{⨈}{\disjquant}
\newunicodechar{⨉}{\bigtimes}
\newunicodechar{⨊}{\modtwosum}
\newunicodechar{⨋}{\sumint}
\newunicodechar{⨌}{\iiiint}
\newunicodechar{⨍}{\intbar}
\newunicodechar{⨎}{\intBar}
\newunicodechar{⨏}{\fint}
\newunicodechar{⨐}{\cirfnint}
\newunicodechar{⨑}{\awint}
\newunicodechar{⨒}{\rppolint}
\newunicodechar{⨓}{\scpolint}
\newunicodechar{⨔}{\npolint}
\newunicodechar{⨕}{\pointint}
\newunicodechar{⨖}{\sqint}
\newunicodechar{⨗}{\intlarhk}
\newunicodechar{⨘}{\intx}
\newunicodechar{⨙}{\intcap}
\newunicodechar{⨚}{\intcup}
\newunicodechar{⨛}{\upint}
\newunicodechar{⨜}{\lowint}
\newunicodechar{⨝}{\Join}
\newunicodechar{⨞}{\bigtriangleleft}
\newunicodechar{⨟}{\zcmp}
\newunicodechar{⨠}{\zpipe}
\newunicodechar{⨡}{\zproject}
\newunicodechar{⨢}{\ringplus}
\newunicodechar{⨣}{\plushat}
\newunicodechar{⨤}{\simplus}
\newunicodechar{⨥}{\plusdot}
\newunicodechar{⨦}{\plussim}
\newunicodechar{⨧}{\plussubtwo}
\newunicodechar{⨨}{\plustrif}
\newunicodechar{⨪}{\minusdot}
\newunicodechar{⨫}{\minusfdots}
\newunicodechar{⨬}{\minusrdots}
\newunicodechar{⨭}{\opluslhrim}
\newunicodechar{⨮}{\oplusrhrim}
\newunicodechar{⨯}{\vectimes}
\newunicodechar{⨰}{\dottimes}
\newunicodechar{⨱}{\timesbar}
\newunicodechar{⨲}{\btimes}
\newunicodechar{⨳}{\smashtimes}
\newunicodechar{⨴}{\otimeslhrim}
\newunicodechar{⨵}{\otimesrhrim}
\newunicodechar{⨶}{\otimeshat}
\newunicodechar{⨷}{\Otimes}
\newunicodechar{⨸}{\odiv}
\newunicodechar{⨹}{\triangleplus}
\newunicodechar{⨺}{\triangleminus}
\newunicodechar{⨻}{\triangletimes}
\newunicodechar{⨼}{\intprod}
\newunicodechar{⨽}{\intprodr}
\newunicodechar{⨾}{\fcmp}
\newunicodechar{⨿}{\amalg}
\newunicodechar{⩀}{\capdot}
\newunicodechar{⩁}{\uminus}
\newunicodechar{⩂}{\barcup}
\newunicodechar{⩃}{\barcap}
\newunicodechar{⩄}{\capwedge}
\newunicodechar{⩅}{\cupvee}
\newunicodechar{⩆}{\cupovercap}
\newunicodechar{⩇}{\capovercup}
\newunicodechar{⩈}{\cupbarcap}
\newunicodechar{⩉}{\capbarcup}
\newunicodechar{⩊}{\twocups}
\newunicodechar{⩋}{\twocaps}
\newunicodechar{⩌}{\closedvarcup}
\newunicodechar{⩍}{\closedvarcap}
\newunicodechar{⩎}{\Sqcap}
\newunicodechar{⩏}{\Sqcup}
\newunicodechar{⩐}{\closedvarcupsmashprod}
\newunicodechar{⩑}{\wedgeodot}
\newunicodechar{⩒}{\veeodot}
\newunicodechar{⩓}{\Wedge}
\newunicodechar{⩔}{\Vee}
\newunicodechar{⩕}{\wedgeonwedge}
\newunicodechar{⩖}{\veeonvee}
\newunicodechar{⩗}{\bigslopedvee}
\newunicodechar{⩘}{\bigslopedwedge}
\newunicodechar{⩙}{\veeonwedge}
\newunicodechar{⩚}{\wedgemidvert}
\newunicodechar{⩛}{\veemidvert}
\newunicodechar{⩜}{\midbarwedge}
\newunicodechar{⩝}{\midbarvee}
\newunicodechar{⩞}{\doublebarwedge}
\newunicodechar{⩟}{\wedgebar}
\newunicodechar{⩠}{\wedgedoublebar}
\newunicodechar{⩡}{\varveebar}
\newunicodechar{⩢}{\doublebarvee}
\newunicodechar{⩣}{\veedoublebar}
\newunicodechar{⩤}{\dsub}
\newunicodechar{⩥}{\rsub}
\newunicodechar{⩦}{\eqdot}
\newunicodechar{⩧}{\dotequiv}
\newunicodechar{⩨}{\equivVert}
\newunicodechar{⩩}{\equivVvert}
\newunicodechar{⩪}{\dotsim}
\newunicodechar{⩫}{\simrdots}
\newunicodechar{⩬}{\simminussim}
\newunicodechar{⩭}{\congdot}
\newunicodechar{⩮}{\asteq}
\newunicodechar{⩯}{\hatapprox}
\newunicodechar{⩰}{\approxeqq}
\newunicodechar{⩱}{\eqqplus}
\newunicodechar{⩲}{\pluseqq}
\newunicodechar{⩳}{\eqqsim}
\newunicodechar{⩴}{\Coloneq}
\newunicodechar{⩵}{\eqeq}
\newunicodechar{⩶}{\eqeqeq}
\newunicodechar{⩷}{\ddotseq}
\newunicodechar{⩸}{\equivDD}
\newunicodechar{⩹}{\ltcir}
\newunicodechar{⩺}{\gtcir}
\newunicodechar{⩻}{\ltquest}
\newunicodechar{⩼}{\gtquest}
\newunicodechar{⩽}{\leqslant}
\newunicodechar{⩾}{\geqslant}
\newunicodechar{⩿}{\lesdot}
\newunicodechar{⪀}{\gesdot}
\newunicodechar{⪁}{\lesdoto}
\newunicodechar{⪂}{\gesdoto}
\newunicodechar{⪃}{\lesdotor}
\newunicodechar{⪄}{\gesdotol}
\newunicodechar{⪅}{\lessapprox}
\newunicodechar{⪆}{\gtrapprox}
\newunicodechar{⪇}{\lneq}
\newunicodechar{⪈}{\gneq}
\newunicodechar{⪉}{\lnapprox}
\newunicodechar{⪊}{\gnapprox}
\newunicodechar{⪋}{\lesseqqgtr}
\newunicodechar{⪌}{\gtreqqless}
\newunicodechar{⪍}{\lsime}
\newunicodechar{⪎}{\gsime}
\newunicodechar{⪏}{\lsimg}
\newunicodechar{⪐}{\gsiml}
\newunicodechar{⪑}{\lgE}
\newunicodechar{⪒}{\glE}
\newunicodechar{⪓}{\lesges}
\newunicodechar{⪔}{\gesles}
\newunicodechar{⪕}{\eqslantless}
\newunicodechar{⪖}{\eqslantgtr}
\newunicodechar{⪗}{\elsdot}
\newunicodechar{⪘}{\egsdot}
\newunicodechar{⪙}{\eqqless}
\newunicodechar{⪚}{\eqqgtr}
\newunicodechar{⪛}{\eqqslantless}
\newunicodechar{⪜}{\eqqslantgtr}
\newunicodechar{⪝}{\simless}
\newunicodechar{⪞}{\simgtr}
\newunicodechar{⪟}{\simlE}
\newunicodechar{⪠}{\simgE}
\newunicodechar{⪡}{\Lt}
\newunicodechar{⪢}{\Gt}
\newunicodechar{⪣}{\partialmeetcontraction}
\newunicodechar{⪤}{\glj}
\newunicodechar{⪥}{\gla}
\newunicodechar{⪦}{\ltcc}
\newunicodechar{⪧}{\gtcc}
\newunicodechar{⪨}{\lescc}
\newunicodechar{⪩}{\gescc}
\newunicodechar{⪪}{\smt}
\newunicodechar{⪫}{\lat}
\newunicodechar{⪬}{\smte}
\newunicodechar{⪭}{\late}
\newunicodechar{⪮}{\bumpeqq}
\newunicodechar{⪯}{\preceq}
\newunicodechar{⪰}{\succeq}
\newunicodechar{⪱}{\precneq}
\newunicodechar{⪲}{\succneq}
\newunicodechar{⪳}{\preceqq}
\newunicodechar{⪴}{\succeqq}
\newunicodechar{⪵}{\precneqq}
\newunicodechar{⪶}{\succneqq}
\newunicodechar{⪷}{\precapprox}
\newunicodechar{⪸}{\succapprox}
\newunicodechar{⪹}{\precnapprox}
\newunicodechar{⪺}{\succnapprox}
\newunicodechar{⪻}{\Prec}
\newunicodechar{⪼}{\Succ}
\newunicodechar{⪽}{\subsetdot}
\newunicodechar{⪾}{\supsetdot}
\newunicodechar{⪿}{\subsetplus}
\newunicodechar{⫀}{\supsetplus}
\newunicodechar{⫁}{\submult}
\newunicodechar{⫂}{\supmult}
\newunicodechar{⫃}{\subedot}
\newunicodechar{⫄}{\supedot}
\newunicodechar{⫅}{\subseteqq}
\newunicodechar{⫆}{\supseteqq}
\newunicodechar{⫇}{\subsim}
\newunicodechar{⫈}{\supsim}
\newunicodechar{⫉}{\subsetapprox}
\newunicodechar{⫊}{\supsetapprox}
\newunicodechar{⫋}{\subsetneqq}
\newunicodechar{⫌}{\supsetneqq}
\newunicodechar{⫍}{\lsqhook}
\newunicodechar{⫎}{\rsqhook}
\newunicodechar{⫏}{\csub}
\newunicodechar{⫐}{\csup}
\newunicodechar{⫑}{\csube}
\newunicodechar{⫒}{\csupe}
\newunicodechar{⫓}{\subsup}
\newunicodechar{⫔}{\supsub}
\newunicodechar{⫕}{\subsub}
\newunicodechar{⫖}{\supsup}
\newunicodechar{⫗}{\suphsub}
\newunicodechar{⫘}{\supdsub}
\newunicodechar{⫙}{\forkv}
\newunicodechar{⫚}{\topfork}
\newunicodechar{⫛}{\mlcp}
\newunicodechar{⫝}{\forksnot}
\newunicodechar{⫞}{\shortlefttack}
\newunicodechar{⫟}{\shortdowntack}
\newunicodechar{⫠}{\shortuptack}
\newunicodechar{⫡}{\perps}
\newunicodechar{⫢}{\vDdash}
\newunicodechar{⫣}{\dashV}
\newunicodechar{⫤}{\Dashv}
\newunicodechar{⫥}{\DashV}
\newunicodechar{⫦}{\varVdash}
\newunicodechar{⫧}{\Barv}
\newunicodechar{⫨}{\vBar}
\newunicodechar{⫩}{\vBarv}
\newunicodechar{⫪}{\barV}
\newunicodechar{⫫}{\Vbar}
\newunicodechar{⫬}{\Not}
\newunicodechar{⫭}{\bNot}
\newunicodechar{⫮}{\revnmid}
\newunicodechar{⫯}{\cirmid}
\newunicodechar{⫰}{\midcir}
\newunicodechar{⫱}{\topcir}
\newunicodechar{⫲}{\nhpar}
\newunicodechar{⫳}{\parsim}
\newunicodechar{⫴}{\interleave}
\newunicodechar{⫵}{\nhVvert}
\newunicodechar{⫶}{\threedotcolon}
\newunicodechar{⫷}{\lllnest}
\newunicodechar{⫸}{\gggnest}
\newunicodechar{⫹}{\leqqslant}
\newunicodechar{⫺}{\geqqslant}
\newunicodechar{⫻}{\trslash}
\newunicodechar{⫼}{\biginterleave}
\newunicodechar{⫽}{\sslash}
\newunicodechar{⫾}{\talloblong}
\newunicodechar{⫿}{\bigtalloblong}
\newunicodechar{⬒}{\squaretopblack}
\newunicodechar{⬓}{\squarebotblack}
\newunicodechar{⬔}{\squareurblack}
\newunicodechar{⬕}{\squarellblack}
\newunicodechar{⬖}{\diamondleftblack}
\newunicodechar{⬗}{\diamondrightblack}
\newunicodechar{⬘}{\diamondtopblack}
\newunicodechar{⬙}{\diamondbotblack}
\newunicodechar{⬚}{\dottedsquare}
\newunicodechar{⬛}{\lgblksquare}
\newunicodechar{⬜}{\lgwhtsquare}
\newunicodechar{⬝}{\vysmblksquare}
\newunicodechar{⬞}{\vysmwhtsquare}
\newunicodechar{⬟}{\pentagonblack}
\newunicodechar{⬠}{\pentagon}
\newunicodechar{⬡}{\varhexagon}
\newunicodechar{⬢}{\varhexagonblack}
\newunicodechar{⬣}{\hexagonblack}
\newunicodechar{⬤}{\lgblkcircle}
\newunicodechar{⬥}{\mdblkdiamond}
\newunicodechar{⬦}{\mdwhtdiamond}
\newunicodechar{⬧}{\mdblklozenge}
\newunicodechar{⬨}{\mdwhtlozenge}
\newunicodechar{⬩}{\smblkdiamond}
\newunicodechar{⬪}{\smblklozenge}
\newunicodechar{⬫}{\smwhtlozenge}
\newunicodechar{⬬}{\blkhorzoval}
\newunicodechar{⬭}{\whthorzoval}
\newunicodechar{⬮}{\blkvertoval}
\newunicodechar{⬯}{\whtvertoval}
\newunicodechar{⬰}{\circleonleftarrow}
\newunicodechar{⬱}{\leftthreearrows}
\newunicodechar{⬲}{\leftarrowonoplus}
\newunicodechar{⬳}{\longleftsquigarrow}
\newunicodechar{⬴}{\nvtwoheadleftarrow}
\newunicodechar{⬵}{\nVtwoheadleftarrow}
\newunicodechar{⬶}{\twoheadmapsfrom}
\newunicodechar{⬷}{\twoheadleftdbkarrow}
\newunicodechar{⬸}{\leftdotarrow}
\newunicodechar{⬹}{\nvleftarrowtail}
\newunicodechar{⬺}{\nVleftarrowtail}
\newunicodechar{⬻}{\twoheadleftarrowtail}
\newunicodechar{⬼}{\nvtwoheadleftarrowtail}
\newunicodechar{⬽}{\nVtwoheadleftarrowtail}
\newunicodechar{⬾}{\leftarrowx}
\newunicodechar{⬿}{\leftcurvedarrow}
\newunicodechar{⭀}{\equalleftarrow}
\newunicodechar{⭁}{\bsimilarleftarrow}
\newunicodechar{⭂}{\leftarrowbackapprox}
\newunicodechar{⭃}{\rightarrowgtr}
\newunicodechar{⭄}{\rightarrowsupset}
\newunicodechar{⭅}{\LLeftarrow}
\newunicodechar{⭆}{\RRightarrow}
\newunicodechar{⭇}{\bsimilarrightarrow}
\newunicodechar{⭈}{\rightarrowbackapprox}
\newunicodechar{⭉}{\similarleftarrow}
\newunicodechar{⭊}{\leftarrowapprox}
\newunicodechar{⭋}{\leftarrowbsimilar}
\newunicodechar{⭌}{\rightarrowbsimilar}
\newunicodechar{⭐}{\medwhitestar}
\newunicodechar{⭑}{\medblackstar}
\newunicodechar{⭒}{\smwhitestar}
\newunicodechar{⭓}{\rightpentagonblack}
\newunicodechar{⭔}{\rightpentagon}
\newunicodechar{〒}{\postalmark}
\newunicodechar{〔}{\lbrbrak}
\newunicodechar{〕}{\rbrbrak}
\newunicodechar{〘}{\Lbrbrak}
\newunicodechar{〙}{\Rbrbrak}
\newunicodechar{〰}{\hzigzag}
\clearpage{}%

\usepackage{pdftexcmds} %
\usepackage{fnpct}
\setfnpct{add-punct-marks=!?:;} %

\usepackage{ellipsis}

\usepackage{fancyvrb}

\usepackage[normalem]{ulem} %
\usepackage{pifont} %
\usepackage{xcolor}
\usepackage{fancyhdr}
\usepackage{xspace}

\usepackage{graphicx} %

\usepackage{afterpage}
\usepackage{wrapfig}

\usepackage[shortlabels]{enumitem}
\usepackage{alphalph} %

\usepackage{booktabs} %

\usepackage{array}
\newcolumntype{L}[1]{>{\raggedright\let\newline\\\arraybackslash\hspace{0pt}}m{#1}}
\newcolumntype{C}[1]{>{\centering\let\newline\\\arraybackslash\hspace{0pt}}m{#1}}
\newcolumntype{R}[1]{>{\raggedleft\let\newline\\\arraybackslash\hspace{0pt}}m{#1}}

\usepackage{tabularx} %

\usepackage{environ}

\NewEnviron{tabx}[1]{\renewcommand{\arraystretch}{1.2} %
\tabularx{\textwidth}{@{}#1@{}}
\toprule
\BODY
\bottomrule
}[\endtabularx]

\usepackage{multirow}

\ifdefined\final
 \newcommand{\authorsnote}[2]{}
\else
 \newcommand*{\authorsnote}[2]{\textcolor{#1}{[#2]}}
\fi

\usepackage{amsmath}
\usepackage{amsfonts}
\usepackage{amssymb}
\usepackage{mathtools} %

\usepackage{amsthm}

\theoremstyle{definition}
\newtheorem{thm}     {Theorem}
\newtheorem{lem}     {Lemma}
\newtheorem{prp}     {Proposition}

\newtheorem{dfn}     {Definition}

\newtheoremstyle{taskstyle}%
  {.32\baselineskip±.15\baselineskip}%
  {.32\baselineskip±.15\baselineskip}%
  {\it}%
  {}%
  {\bf}%
  {. }%
  { }%
  {}%
\theoremstyle{taskstyle}

\newcounter{subtasknumber}
\newcounter{subtasklabel}
\renewcommand{\thesubtasklabel}{\textbf{\thetasknumber}}

\newcommand{\mc}{\mathcal}
\newcommand{\msf}{\mathsf}

\newcommand{\given}{\mathrel{|}}
\DeclareMathOperator{\Cr}{Cr}
\DeclareMathOperator{\Fr}{Fr}

\newcommand{\doo}{\mathrm{do}}

\usepackage[T1]{fontenc}

\usepackage{stmaryrd}
\newcommand*{\llb}{\llbracket}
\newcommand*{\rrb}{\rrbracket}

\newcommand{\Expect}{{\rm I\kern-.3em E}}

\newcommand{\bigsum}[1]{
  \mathop{\oalign{$\displaystyle\sum$\hfill\cr
    $\begin{subarray}{l}\mathrlap{#1}\end{subarray}$\hfill\cr}\:}
}%

\newcommand{\smldots}{\text{\scriptsize{...}}}
\usepackage{tikz}
\usetikzlibrary{cd}

\renewcommand{\Fr}{\textstyle\Pr}
\newcommand{\changes}[1]{\underline{#1}}
\newcommand{\bad}{\relax{}}
\newcommand{\groupheader}{\relax}

\DeclareMathOperator{\rp}{RP}
\newcommand*{\mrp}{\bar{RP}}

\newtheorem{chr}[dfn]{Definition}
\newtheorem{view}[dfn]{Definition}
\newtheorem{prpchr}[prp]{Proposition}
\newtheorem{prpview}[prp]{Proposition}

\newcommand{\olangle}{}
\newcommand{\orangle}{}

 \newcommand{\egcite}[1]{ (e.g.,~\cite{#1})}

\newcommand*{\aff}[1]{{\normalsize {#1}}}

\title{Differential Privacy as a Causal Property}
\author{Michael Carl Tschantz, \aff{International Computer Science Institute}\\ %
Shayak Sen, \aff{Carnegie Mellon University} %
\\
Anupam Datta, \aff{Carnegie Mellon University} %
}

\begin{document}

\maketitle

\begin{abstract}
We present associative and causal views of differential privacy.
Under the associative view, the possibility of dependencies between
data points precludes a simple statement of differential privacy's
guarantee as conditioning upon a single changed data point.
However, we show that a simple characterization of differential privacy as
limiting the effect of a single data point does exist under the
causal view, without independence assumptions about data points.
We believe this characterization resolves disagreement and confusion
in prior work about the consequences of differential privacy.
The associative view needing assumptions boils down to the contrapositive of the maxim that correlation
doesn’t imply causation: differential privacy ensuring a lack
of (strong) causation does not imply a lack of (strong) association.
Our characterization also
opens up the possibility of applying results from statistics,
experimental design, and science about causation while studying
differential privacy.
\end{abstract}

\setlength{\parindent}{0.8pc} %

\section{Introduction}
\label{sec:intro}

Differential Privacy\ndss{ (\DP{})} is a precise mathematical property of an algorithm
requiring that it produce almost identical distributions of outputs
for any pair of possible input databases that differ in a single data
point.
Despite the popularity of \DP{} in the research
community, unease with the concept remains.
For example, Cuff and Yu's paper states ``an intuitive understanding can be elusive''
and \tufte{expresses a preference}\ndss{recommends} that \DP{} be
related to more familiar concepts based on statistical associations, such as
mutual information~\cite[p.\,2]{cuff16ccs}.
This and numerous other works exploring similar connections between \DP{} and
statistical association each makes assumptions about the data points
(e.g.,
\cite[p.\,9]{alivim11icalp} %
\cite[p.\,32]{clarkson15mathstructcompsci},
\cite[p.\,4]{barthe11csf},
\cite[p.\,14]{mcgregor11eccc}, %
\cite[p.\,6]{ghosh17arxiv}).

The use of such assumptions
has led to
some papers stating that \DP{}
implicitly requires some assumption:
that it requires the data points to be independent
(e.g.,
\cite[p.\,2]{kifer11sigmod},
\cite[p.\,1]{kifer12pods},
\cite[p.\,2]{li13ccs},
\cite[p.\,3]{he14sigmod}, %
\cite[p.\,7]{chen14vldbj},
\cite[p.\,232]{zhu15tifs},
\cite[p.\,1]{liu16ndss}),
that the adversary must know all but one data point, the
so-called \emph{strong adversary assumption} (e.g.,~\cite[p.\,2]{cuff16ccs}, \cite[p.\,10]{li13ccs}),
or that either assumption will do (e.g.,~\cite[\S1.2]{yang15sigmod}).
(Appendix~\ref{app:history-camp2} provides quotations.)
Conversely,
other works assert that no such assumption exists\egcite{bassily13focs,kasiviswanathan14jpc,mcsherry16blog1,mcsherry16blog2}.
How can such disagreements arise about a precise mathematical property of an algorithm?

We put to rest both the nagging feeling that
\DP{} should be expressed in more basic terms and the
disagreement about whether it makes various implicit assumptions.
We do so by showing that \DP{} is better
understood as a causal property than as an associative one.
We show that \DP{} constrains effect sizes, a basic concept
from empirical science about how much changing one variable changes another.
This view does not require any independence or adversary assumptions.

Furthermore, %
we show that the difference between the two views over whether \DP{} makes assumptions is precisely captured as the difference between association and causation.
That some fail to get what they want out of \DP{} (without making an assumption) comes from the contrapositive of the maxim \emph{correlation doesn't imply causation}: \DP{} ensuring a lack of (strong) causation does not imply a lack of (strong) association.
Given the common confusion of association and causation, and that \DP{} does not make its causal nature explicit in its mathematical statement,
we believe our work explains how disagreement could have arose in the research literature about the what assumptions \DP{} requires.

\subsection{Motivating Example and Intuition}
\label{sec:mot-example}

To provide more details,
let us consider an example of using \DP{}
inspired by Kifer and Machanavajjhala~\cite{kifer11sigmod}.
Suppose Ada and her son Byron are considering participating in a
differentially private survey with $n-2$ other people.
The survey collects a data point from each participant about their
health status with respect to a genetic disease.
Since Ada and Byron are closely related, their data points are
closely related.
This makes them wonder whether the promise of \DP{}
becomes watered down for them, a worrying prospect given
the sensitivity of their health statuses.

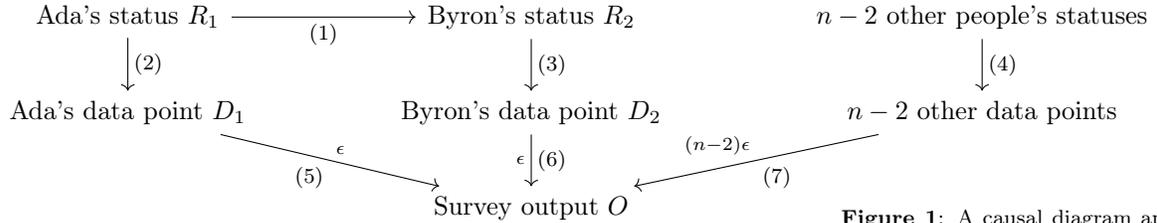
\begin{figure*}[t]
\centering
\begin{tikzcd}[math mode=false, row sep=2em, column sep=5em]
  Ada's status $R_1$ \arrow[r, "{\footnotesize (1)}"']\arrow[d, "{\footnotesize (2)}"] & Byron's status $R_2$ \arrow[d, "{\footnotesize (3)}"] & $n-2$ other people's statuses\arrow[d, "{\footnotesize (4)}"]\\
  Ada's data point $D_1$ \arrow[dr,"{\footnotesize (5)}"',"$\epsilon$"]
  &  Byron's data point $D_2$ \arrow[d, "$\epsilon$"', "{\footnotesize (6)}"]   &   $n-2$ other data points \arrow[dl,"${\footnotesize (n-2)}\epsilon$"',"{\footnotesize (7)}"]\\
  & Survey output $O$ %
\end{tikzcd}
\caption{A causal diagram approximating the process through which the output of a statistical query is generated and used. The arrows represent direct causal effects.
Indirect cause effects can be inferred from taking the transitive closure of the arrows.
$\epsilon$ labels causal effects bounded by $\epsilon$-differential privacy.
(1)--(7) serve as labels naming arrows.}
\label{fig:causal-privacy}
\end{figure*}

Figure~\ref{fig:causal-privacy} summarizes what would happen if both Ada and Byron participate in the survey.
In it,
each solid arrow
represents a causal relationship where the quantity at the start of the arrow causally affects
the quantity at the end of the arrow.
For example, Arrow~(1) represents that Ada's genetics has a causal
effect on her son Byron's genetics.
We use an arrow since causation is directional: Byron's genetics does not have a causal effect on Ada's.
Arrow~(2) represents a mechanism by which
Ada provides her status to the survey.  This information becomes a data point in the survey's data set, that is,
a row in a database.
This database comprises Ada's data point, Byron's data point,
and $n-2$ other people's data points.
Arrows~(5), (6), and (7) together represent the algorithm that
computes the survey's result, that is, the output produced from the database using a differentially private algorithm.

As mentioned, Ada's status also affects the status of her son Byron, shown with Arrow~(1).
Therefore, their statuses are statistically associated (i.e., not probabilistically independent).
While causation is directional, such associations are not:
seeing Byron's status reveals information about Ada's
status despite not causing Ada's status.
Furthermore,
Ada's and Byron's data points will be
statistically associated
because they have a common cause, Ada's status.
Thus, seeing Byron's data point reveals information about Ada's
status and data point\tufte{, and seeing Ada's data point reveals
information about Byron's status and data point}.
Since both Ada's and Byron's data points reveal information about Ada's status,
the output can be informed by two data points about Ada's status.
This double dose of information is what gives Ada pause about participating.
Furthermore, much the same applies to Byron.

\tufte{Now, let us consider what \DP{} provides in this example.}
In the words of Kasiviswanathan and Smith~\cite[p.\,2]{kasiviswanathan14jpc}, \DP{} intuitively
ensures that
\begin{flushright}
\begin{tabular}{rr@{}}
  \begin{minipage}{0.85\columnwidth}
changing a single individual’s data in the database leads
to a small change in the distribution on outputs.
\end{minipage}
& $(*)$\\
\end{tabular}
\end{flushright}
This intuitive consequence of \DP{}, denoted as ``$(*)$'',
does not make explicit the notion of \emph{change} intended.
It implicitly compares the distribution over the output, a random variable $O$,
in two hypothetical worlds, the pre- and post-change worlds.
If we focus on the individual Ada and let $D_1$ be a random variable representing her data point as it changes
values from $d_1$ to $d'_1$, then the comparison is between
$\Fr[O{=}o\ \text{when}\ D_1{=}d_1]$
and $\Fr[O{=}o\ \text{when}\ D_1{=}d_1']$.
The part of this characterization of \DP{} that is informal is the
notion of \emph{when}, which leaves the notion of \emph{change} imprecise.
Our paper
contrasts various interpretations of \emph{change} and \emph{when}.

The most obvious interpretation is that of \emph{conditioning} upon two different values for the changed status.
This interpretation implies an approximation of statistical independence
between an individual's data point and the output:
$\Fr[O{=}o \given D_1{=}d_1] \approx \Fr[O{=}o \given D_1{=}d'_1]$.
Presuming the data points are truthful, such an approximate independence implies (up to a factor) an approximate independence that compares probabilities over a status with or without knowing the output, that is,
$\Fr[R_1{=}r_1\given O{=}o] \approx \Fr[R_1{=}r_1]$.
In this case,
observing the output reveals little about an individual's status,
explaining this interpretation's appeal.

However, as discussed above,
both Ada's and Byron's data points reveal information about each of their statuses
since associations depend upon the full breadth of causal relations.
This double dose of information about their statuses means that \DP{} does not actually imply this appealing form of approximate independence.
Thus, attempts to interpret \DP{} in terms of
conditioning fail to hold in the presence of the associations between
the data points.
Those desiring an associative guarantee from \DP{} must
rule out such double doses of information, for example, by assuming that the data points
lack any associations or that the adversary already knows all but
one data point, making such associations uninformative.

Now, let us instead consider interpreting \DP{} in
terms of \emph{causal interventions}.
This interpretation models artificially altering the value of random variables, as in a randomized experiment.
The key difference between intervening upon a random variable and
conditioning upon it is that while intervening tracks causal effects
by accounting for how the intervention may cause other variables to
change, it does not depend upon all the associations in the database
since such interventions break them.
Thus, while associative definitions using conditioning depends upon the distribution producing data points, causal ones can screen off this distribution to examine the behavior of just the \DP{} algorithm itself by intervening upon all its outputs.

For example, suppose Byron is born without the genetic disease and a scientist flips a coin and ensures that Byron has the disease if it comes up heads and ensures that he does not if it comes up tails.
(While the technology to execute this experiment is currently wanting, it is conceptually possible.)
Since Bryon starts without the disease, the tails outcome does nothing and can be viewed a control treatment while the heads outcome causes a change.
If it comes up heads, the scientist could measure various things about Byron to see what changed from giving him the disease.
In particular, Byron's data point and the output computed from it would change.
On the other hand, nothing would change about Ada since causation is directional.
(Section~\ref{sec:cause} makes this more precise.)
In fact, after the randomization, Bryon's status and data point no longer reveals any information about Ada's status since the randomization broke the association between their statuses.

The scientist can measure the size of any changes to compute an \emph{effect size}.
The effect size for Byron's data point would be large since the data point is supposed to be equal to the status, but the effect size for the output will be small since it is computed by an algorithm with $\epsilon$-differential privacy.
If we instead consider intervening on Ada's status, we find two paths to the output:
one via Ada's data point (Arrows~(2) and (5)) and another via Byron's (Arrows~(1), (3), and (6)).
These two paths mean that the effect size could be as much as double that of changing Byron's status.
Thus, \DP{} cannot be interpreted as limiting the effect of changing Ada's status to just $\epsilon$ in size.

Recall that the intuitive characterization $(*)$ of \DP{} referred to data points, not statuses:
``changing a single individual’s data in the database\ldots''%
~\cite[p.\,2]{kasiviswanathan14jpc}.
So, let us consider intervening upon the data points instead.
Each data point is piped directly into the differentially private algorithm and has no other effects.
Thus, \DP{} does bound the effect size at $\epsilon$ for Ada's data point without making any assumptions about the\tufte{ relationships between} statuses.
For this reason, we believe
\DP{} is better understood as a bound on effect sizes
than as a bound on associations.

We believe that ease of conflating associative and causal properties explains the disagreement in the research literature.
(See Appendix~\ref{app:history} for a history of this disagreement.)
Our observation also reduces the benefits and drawbacks of these implicitly associative and causal views of privacy to those known from studying association and causation in general.
For example, the causal view only requires looking at the system itself (causation is an inherent property of systems)
while the associative view requires looking at the input distribution as well.
This difference explains why papers implicitly with the associative view discuss the distribution over data points despite the definition of \DP{} and the implicitly causal papers do not mention it.

The causal characterization also requires us to distinguish between an
individual's attributes ($R_i$s) and the data that is input to an algorithm
($D_i$s), and intervenes on the latter. Under the assumption that individuals
report their true statuses, the associative interpretation does not require this distinction
since conditioning on one is identical to conditioning on the other. This
distinction captures an aspect of the difference between protecting ``secrets about
you'' ($R_i$) and protecting ``secrets from you'' ($D_i$) pointed out by
McSherry~\cite{mcsherry16blog1,mcsherry16blog2}, where \DP{}
protects the latter in a causal sense.
An individual's attribute $R_i$ is \emph{about} him and its
value is often outside of his control.
On the other hand, an individual's data point $D_i$, at least in the
setting typically envisioned for \DP{}, is under his
control and is volunteered by the individual, making it \emph{from} him.

\subsection{Overview}
\label{sec:overview}

Our main goal is to demonstrate that \DP{} can be understood as a causal property without needing the sorts of assumptions made to view it as an associative property.
We lay out the associative view by surveying definitions presented in prior work to show its awkward fit for \DP{} and how it leads to suggestions that \DP{} makes assumptions.
We then turn to the causal view, replacing conditioning with interventions in the associative
definitions. Doing so reveals three key insights; we find that the causal definitions (1) work without such assumptions, (2) provides a tight characterization of \DP{}, and (3) explains how \DP{} maps to a concept found throughout statistics and science, namely to a measure of effect sizes.

We start our analysis with the associative view, which uses conditioning (Section~\ref{sec:dp-associative}).
We first consider conditioning upon all the data points instead of just the changed one.
After dealing with some annoyances involving the inability to condition on zero-probability data points, we get a precise characterization of \DP{} (Definition~\ref{view:dp-utc}).
However, this associative definition does not correspond well to the intuitive characterization $(*)$ of differential privacy's key consequences:
whereas the above-quoted characterization refers to just the changed data point,
this associative definition refers to them all, thereby blurring the characterization's focus on change.

Next, we modify the associative definition to condition upon just the single changed data point (Definition~\ref{view:dp-ucb}). %
The resulting definition prohibits more than an $\epsilon$ degree of correlation between the data point and the output, hereby limiting what can be learned about the data point.
While this definition is not implied by \DP{} on its own, %
it is implied with an additional assumption of independence between data points (Definition~\ref{view:dp-ica}).
We believe that this explains the claim found in some papers that \DP{} implicitly assumes independence.

However, we do not share this feeling since the independence assumption is not required to get \DP{} to imply the intuitive consequence $(*)$ quoted above when interpreting \emph{change} as a causal intervention instead of as associative conditioning.
After reviewing the core concepts of causal modeling
(Section~\ref{sec:cause}), we consider intervening upon all the
data points (Section~\ref{sec:dp-causal-whole}).
As with conditioning upon all the data points, a definition intervening on all the data points (Definition~\ref{dfn:dp-tc})
characterizes \DP{} (Proposition~\ref{prp:dp-tc}) but
without the intuitive focus on a single data point that we desire.

We then consider characterizing \DP{} as intervening upon a single point (Definition~\ref{chr:dp-spc} of Section~\ref{sec:dp-causal-single}).
A benefit of this causal characterization is that it is implied by \DP{} without any assumptions about independence (Proposition~\ref{prpchr:dp-spc}).
An additional benefit is that, unlike the associative characterizations, we do not need side conditions limiting the characterization to data points with non-zero probabilities.
This benefit follows from causal interventions being defined for zero-probability events unlike conditioning upon them.
These two benefits lead us to believe that \DP{} is better viewed as a causal property than as an associative one.

In addition to considering the consequences of \DP{} through the lenses of association and causation, we also consider how these two approaches can provide definitions equivalent to \DP{}.
Table~\ref{tbl:overview-eqv} shows our key results about definitions that are either equivalent to \DP{} or might be mistaken as such, which, in the sections below, we weave in with our aforementioned results about characterizations of the consequences of \DP{}.

When intervening upon all data points, we get equivalence for free from Definition~\ref{dfn:dp-tc} that we already explored as a characterization of the consequences of \DP{}.
This free equivalence does not occur for conditioning upon all data points since the side condition ruling out zero-probability data points means those data points are not constrained.
Since \DP{} is a restriction on all data points, to get an equivalence, the definition must check all data points.
To achieve this, we further require that the definition hold on all distributions over the data points, not just the naturally occurring distribution.
(Alternatively, we could require the definition to hold for any one distribution with non-zero probabilities for all data points, such as the uniform distribution.)
We also make similar alterations to the definitions looking at a single data point.

\begin{table*}[t]
\tufte{\let\Caption=\caption\renewcommand{\caption}[1]{\Caption[][2ex]{#1}}}
\centering\tufte{\small}
\caption{%
Differential Privacy and Variations upon It.
\normalfont
The left-most column gives the number of its definition later in the text.
The point of comparison is the quantity computed for every pair of values $d_i$ and $d'_i$ for $\changes{d}_i$
to check whether the point of comparison's values for $d_i$ and for $d'_i$ are within a factor of $e^\epsilon$ of one another.
The check is for all values of the index $i$. %
Some of the definitions only perform the comparison when the probability of the changed data point $D_i$ having the value $d_i$ (and $d'_i$, the changed value) is non-zero under $\mc{P}$.
Others only perform the comparison when all the data points $D$ having the values $d$ (and $d'$ for changed value of $D_i$) has non-zero probability.
$\doo$ denotes a causal intervention instead of standard conditioning~\cite{pearl09book}.
The definitions vary in whether they require performing these comparisons for just the actual probability distribution over data points $\mc{P}$ or over all such distributions.
In one case (Definition~\ref{view:dp-ica}), the comparison just applies to distributions where the data points are independent of one another.}
\label{tbl:overview-eqv}
\renewcommand{\arraystretch}{1.5} %
\tufte{\footnotesize}
\begin{tabular}{@{}clllc@{\,}l@{}}
\toprule
Num. & $\mc{P}$ & Conditions on population distribution $\mc{P}$ & Point of comparison (should be stable as $\changes{d}_i$ changes) & \multicolumn{2}{l}{Relation}\\ %
\midrule
\multicolumn{6}{@{}l@{}}{\groupheader{Original Differential Privacy}}\\
\ref{dfn:dp}     & n/a               & & $\Fr_{\mc{A}}[\mc{A}(\langle d_1,\ldots,\changes{d_i},\ldots,d_n\rangle){=}o]$ & is & DP\\
\midrule
\multicolumn{6}{@{}l@{}}{\groupheader{Associative Variants}}\\
\ref{view:dp-utc} & $\forall$      & $\Fr_{\mc{P}}[D_1{=}d_1,\ldots,D_i{=}\changes{d_i},\ldots,D_n{=}d_n] > 0$ & $\Fr_{\mc{P},\mc{A}}[O{=}o \given D_1{=}d_1,\ldots,D_i{=}\changes{d_i},\ldots,D_n{=}d_n]$ & $\leftrightarrow$ & DP\\
\ref{view:dp-ucb} & $\forall$      & $\Fr_{\mc{P}}[D_i{=}d_i] > 0$ & $\Fr_{\mc{P},\mc{A}}[O{=}o \given D_i{=}\changes{d_i}]$ & \bad{$\rightarrow$} & \bad{DP}\\
\ref{view:dp-ica} & $\forall$ indep.\@ $D_i$ & $\Fr_{\mc{P}}[D_i{=}d_i] > 0$ & $\Fr_{\mc{P},\mc{A}}[O{=}o \given D_i{=}\changes{d_i}]$ & $\leftrightarrow$ & DP\\
\midrule
\multicolumn{6}{@{}l@{}}{\groupheader{Causal Variants}}\\
\ref{view:dp-tc-univ} & $\forall$ & & $\Fr_{\mc{P},\mc{A}}[O{=}o \given \doo(D_1{=}d_1,\ldots,D_i{=}\changes{d_i},\ldots,D_n{=}d_n)]$ & $\leftrightarrow$ & DP\\
\ref{dfn:dp-tc}  & given               & & $\Fr_{\mc{P},\mc{A}}[O{=}o \given \doo(D_1{=}d_1,\ldots,D_i{=}\changes{d_i},\ldots,D_n{=}d_n)]$ & $\leftrightarrow$ & DP\\
\ref{chr:dp-spc}& given     & & $\Fr_{\mc{P},\mc{A}}[O{=}o \given \doo(D_i{=}\changes{d_i})]$ & $\leftarrow$ & DP\\
\ref{view:dp-spca}& $\forall$      & & $\Fr_{\mc{P},\mc{A}}[O{=}o \given \doo(D_i{=}\changes{d_i})]$ & $\leftrightarrow$ & DP\\
\bottomrule
\end{tabular}
\end{table*}

Having shown that \DP{} can be viewed as a causal property,
we then consider how this view can inform our understanding of it.
We relate \DP{} to a previously studied
notion of effect size and discuss how this more general notion
can make discussions about privacy more clear (Section~\ref{sec:brp}).
In particular, \DP{} is a bound on the measure
of effect size called \emph{relative probabilities} (also known as
\emph{relative risk} and \emph{risk ratio}).
That is, \DP{} bounds the relative probabilities for
the effects of each data point upon the output.
Since not all research papers are in agreement
about what counts as an individual's data point,
spelling out exactly which random variables have bounded relative probabilities may be more clear than simply asserting that
\DP{} holds for some implicit notion of data point.

We then consider in more detail the relationship between our work and
that of Kasiviswanathan and Smith~\cite{kasiviswanathan14jpc}
(Section~\ref{sec:knowledge}).
In short, Kasiviswanathan and Smith provide a Bayesian interpretation
of \DP{} whereas we provide a complementary causal one.

As we elaborate in the conclusion (Section~\ref{sec:conc}), these
results open up the possibility of using all the methods developed for
working with causation to work with \DP{}.
Furthermore, it explains why researchers have found uses for
\DP{} out side of privacy\egcite{dwork12itcs,dwork15stoc,dwork15nips,dwork15science,lecuyer18arxiv}:
they are really trying to limit effect sizes.

\section{Prior Work}
\label{sec:prior}

The paper coining the term ``differential privacy'' recognized that causation is key to understanding \DP{}: ``it will not be the presence of her data that causes [the disclosure of sensitive information]''~\cite[p.\,8]{dwork06icalp}.  Despite this causal view being present in the understanding of \DP{} from the beginning, we believe we are first to make it mathematically explicit and precise, and to compare it explicitly with the associative view.

Tschantz~et~al.~\cite{tschantz15csf} reduces probabilistic noninterference (a notion of having no flow of information) to having no casual effect at all.
We observe that \DP{} with $\epsilon = 0$ is identical to noninterference, implying that the $\epsilon = 0$ case of \DP{} could be reduced to causal effects.
Our work generalizes from non-interference to \DP{} and thereby differs in having additional bookkeeping to track the size of the effect for handling the $\epsilon > 0$ case, where an effect may be present but must be bounded. Importantly, this generalization allows us to compare the causal and associative views of \DP{}, not a focus of ~\cite{tschantz15csf}.

Our work is largely motivated by wanting to explain
the difference between two lines of research papers that have emerged from \DP{}.
The first line, associated with the inventors of \DP{}, emphasizes differential privacy's ability to ensure that data providers are no worse off for providing data\egcite{dwork06icalp,kasiviswanathan14jpc,mcsherry16blog1,mcsherry16blog2}.
The second line, which formed in response to limitations in differential privacy's guarantee, emphasizes that an adversary should not be able to learn anything sensitive about the data providers from the system's outputs\egcite{kifer11sigmod,kifer12pods,li13ccs,kifer14database,he14sigmod,chen14vldbj,zhu15tifs,liu16ndss}.
The second line notes that \DP{} fails to provide this guarantee when the data points from different data providers are associated with one another unless one assumes that the adversary knows all but one data point.
McSherry provides an informal description of the differences between the two lines~\cite{mcsherry16blog1}. %
While not necessary for understanding our technical development,
Appendix~\ref{app:history} provides a history of the two views of \DP{}.

Kasiviswanathan and Smith look at a different way of comparing the two views of \DP{}, which they call \emph{Semantic Privacy}~\cite{kasiviswanathan14jpc}.
They study the Bayesian probabilities that an adversary seeing the system's outputs would assign to a sensitive property.
Whereas other works looking at an adversary's beliefs, such as Pufferfish~\cite{kifer14database}, bounds the change in the adversary's probabilities before and after seeing the output, Kasiviswanathan and Smith bound the change between
adversary's probabilities after seeing the output for two difference inputs, much as \DP{} compares output distributions for two different inputs.
They conclude that this posterior-to-posterior comparison captures the epistemic consequences of \DP{}, unlike the anterior-to-posterior comparison made by Pufferfish-like definitions, since \DP{} bounds it without additional assumptions, such as independent data points.
Our causal definitions (Def.\,\ref{view:dp-tc-univ}--\ref{view:dp-spca})
instead expose differential privacy's causal nature with a
modification of Pearl's causal framework as a frequentist effect size
and we do not use any Bayesian probabilities in our causal definitions.
We view their Bayesian non-causal characterization of \DP{} as complimentary to our
frequentist causal characterization, with theirs focused on an adversary's knowledge and ours on physical constraints.
(We conjecture that a Bayesian causal characterization should be possible, but
leave that to future work.)
Besides the conceptual difference, our characterization is tighter in that we show an exact equivalence between our central definition (Def.\,\ref{view:dp-spca}) and \DP{} in that each implies the other with the same value of $\epsilon$, whereas their implications hold for an increased value of $\epsilon$.
Section~\ref{sec:knowledge} considers their work in more detail.

Others have explored
how assumptions about the data or adversary enables
alternative reductions of \DP{} to information flow
properties.
Clarkson and Schneider prove an equivalence between \DP{} and an information-theoretic notion of information suppression while making the strong adversary assumption~\cite[p.\,32]{clarkson15mathstructcompsci}. %
After making the strong adversary assumption, Cuff and Yu have argued that \DP{} can be viewed a constraint on mutual
information~\cite[p.\,2]{cuff16ccs}, but McSherry points out that the
connection is rather weak~\cite{mcsherry17blog1}.
Alvim~et~al.\@ bound the min-entropy and mutual
information in terms of $\epsilon$ under assumptions about the data's
distribution~\cite[p.\,9]{alivim11icalp}.
Ghosh and Kleinberg provide inferential privacy bounds
for \DP{} mechanisms under assumptions about restricted background knowledge~\cite[p.\,6]{ghosh17arxiv}.
We avoid such assumptions and our causal version of \DP{} (Def.\,\ref{view:dp-spca}) is equivalent to the original, not merely a bound.

Instead of looking at how much an adversary learns about a single data point, Barthe and K\"opf bound how much adversary learns, in terms of min entropy, about the whole database from a differentially private output, while sometimes making the strong adversary assumption~\cite[p.\,4]{barthe11csf}.
They prove that as the database increases size, the bound increases as well.
McGregor et al.\@ similarly bound the amount of information leaked, in terms of mutual information, about the whole database by a differentially private protocol (the information cost), while sometimes assuming independent data points~\cite[p.\,14]{mcgregor11eccc}.
We focus on privacy consequences to individuals, that is, on one data point at a time.

Other papers have provided flexible or convenient associative definitions not limited to attempting to capture \DP{}.
For example, Pufferfish is a flexible framework for stating associative privacy properties~\cite{kifer14database}.
Lee and Clifton explore bounding the probability that the adversary can assign to an individual being in a data set~\cite{lee12kdd}.
While such probabilities are more intuitive than the $\epsilon$ of \DP{}, their central definition implicitly makes a strong adversary assumption~\cite[Def.\,4]{lee12kdd}.

\section{Differential Privacy as Association}
\label{sec:dp-associative}

Dwork provides a well known expression of \DP{}~\cite[p.\,8]{dwork06icalp}\ndss{.}\tufte{:
\begin{quote}
\noindent\textbf{Definition 2.} \emph{A randomized function $\mc{K}$ gives \emph{$ε$-differential privacy} if for all data sets $D_1$ and $D_2$ differing on at most one element, and all $S ⊆ \mathit{Range}(\mc{K})$,}
\begin{align*}
\Pr[\mc{K}(D_1) ∈ S] &≤ \exp(ε) × \Pr[\mc{K}(D_2) ∈ S] \quad (1)
\end{align*}
\end{quote}}
In our notation, it becomes
\begin{view}[Differential Privacy]
\label{dfn:dp}
A randomized algorithm $\mc{A}$ is
\emph{$ε$-differentially private} %
if
for all $i$, for all data points
$d_1, \ldots, d_n$ in $\mc{D}^n$
and $d'_i$ in $\mc{D}$,
and for all output values $o$,
\[ \Fr_{\mc{A}}[\mc{A}(\langle d_1,\ldots,d_n\rangle){=}o] ≤ e^ε \Fr_{\mc{A}}[\mc{A}(\langle d_1, \ldots, d'_i, \ldots, d_n\rangle){=}o] \]
\end{view}
This formulation differs from Dwork's formulation in four minor ways.
First, for simplicity, we restrict ourselves to only considering programs producing outputs over a finite domain, allowing us to use notationally simpler discrete probabilities.
Second, we change some variable names.
Third, we explicitly represent that the probabilities are over the
randomization within the algorithm $\mc{A}$, which should be
understood as physical probabilities, or \emph{frequencies}, not as
epistemic probabilities, or Bayesian \emph{credences}.
Fourth, we use the \emph{bounded} formulation of \DP{},
in which we presume a maximum number $n$ of individuals potentially providing data.
In this formulation, it is important that one of the possible values
for data points is the null data point containing no information to
represent an individual deciding to not participate.

Both Dwork's expression of and our re-expression of \DP{} make
discussing the concerns about dependencies between data points raised by
some papers difficult since it does not mention any distribution over
data points.
This omission is a reflection of the standard view that \DP{} does not depend upon that distribution.
However, to have a precise discussion of this issue, we should introduce notation for denoting the data points.
We use Yang~et~al.'s expression of \DP{} as a starting point~\cite[p.\,749]{yang15sigmod}:
\begin{quote}
\noindent\textit{Definition 4.} (Differential Privacy) A randomized mechanism $\mc{M}$ satisfies \textit{$ε$-differential privacy}, or $ε$-DP, if
\begin{align*}
DP(\mc{M}) &:= \sup_{i, \mathbf{x}_{-i}, x_i, x'_i, S} \log \frac{\Pr(r ∈ S \given x_i, \mathbf{x}_{-i})}{\Pr(r ∈ S \given x'_i, \mathbf{x}_{-i})} ≤ ε. \tufte{\quad (1)}
\end{align*}
\end{quote}
We rewrite this definition in our notation as follows:
\begin{view}[Strong Adversary Differential Privacy]
\label{view:dp-utc}
A randomized algorithm $\mc{A}$ is
\emph{$ε$-strong adversary differentially private} %
if for all population distributions $\mc{P}$,
for all $i$, for all data points
$d_1, \ldots, d_n$ in $\mc{D}^n$
and $d'_i$ in $\mc{D}$,
and for all output values $o$,
if
\begin{align}
\Fr_{\mc{P}}[D_1{=}d_1, \ldots,D_i{=}d_i, \ldots, D_n{=}d_n] &> 0 \\
\text{and}\quad \Fr_{\mc{P}}[D_1{=}d_1, \ldots,D_i{=}d'_i, \ldots, D_n{=}d_n] &> 0
\end{align}
then
\begin{multline}
\Fr_{\mc{P},\mc{A}}[O{=}o \given D_1{=}d_1, \ldots,D_i{=}d_i, \ldots, D_n{=}d_n] \\≤ e^ε * \Fr_{\mc{P},\mc{A}}[O{=}o \given D_1{=}d_1, \ldots,D_i{=}d'_i, \ldots, D_n{=}d_n]
\label{eqn:conditioning-on-all}
\end{multline}
where $O = \mc{A}(D)$ and $D = \langle D_1, \ldots, D_n\rangle$.
\end{view}

This formulation differs from Yang~et~al.'s formulation in the following ways.
As before, we change some variable names and only consider programs producing outputs over a finite domain.
Also, rather than using shorthand, we write out variables explicitly and denote the distributions from which they are drawn.
For example, for what they denoted as $\Pr(r ∈ S \given x'_i, \mathbf{x}_{-i})$, we write
$\Pr_{\mc{P},\mc{A}}[O{=}o  \given D_1{=}d_1, \ldots, D_i{=}d'_i, \ldots, D_n{=}d_n]$,
where the data points $D_1,\ldots,D_n$ are drawn from the population distribution $\mc{P}$
and the output $O$ uses the algorithm's internal randomization $\mc{A}$.
This allows explicitly discussion of how the data points $D_1,\ldots,D_n$ may be correlated in the population $\mc{P}$ from which they come.

Finally, we explicitly deal with data points potentially having
a probability of zero under $\mc{P}$.
We ensure that we only attempt to calculate the conditional probability for
databases with non-zero probability.
This introduces a new problem: if the probability distribution
$\mc{P}$ over databases assigns zero probability to a data point value $d_i$, we
will never examine the algorithm's behavior for it.
While the algorithm's behavior on zero-probability events may be of little
practical concern, it would allow the algorithm $\mc{A}$ to violate
\DP{}.
(See Appendix~\ref{app:chr-csb} for an example.)
To remove this possibility, we quantify over all probability
distributions, which will include some with non-zero probability for
every combination of data points.

Alternately, we could have used just one distribution that assigns
non-zero probability to all possible input data points.
We instead quantify over all distributions
to make it clear that \DP{} implies a property
for all population distributions $\mc{P}$.
While the population distribution $\mc{P}$ is needed to compute the probabilities
used by Definition~\ref{view:dp-utc} and will change the probability of outcomes,
whether or not $\mc{A}$ has \DP{} does not actually
depend upon the distribution beyond whether it assigns non-zero probability to data points.
This lack of dependence explains why \DP{} is typically
defined without reference to a population distribution $\mc{P}$ and
typically only mentions the algorithm's randomization $\mc{A}$.

For us, the population distribution $\mc{P}$ serves to link the algorithm
to the data on which it is used, explaining the consequences of the algorithm
for that population.
Since the concerns of Yang~et~al.\@ and others deal with differential
privacy's behavior on populations with correlated data points, having
this link proves useful.
The following theorem shows that its introduction does not alter
the concept.

\begin{prpview}\label{prpview:dp-utc}
Definitions~\ref{dfn:dp} and~\ref{view:dp-utc} are equivalent.
\end{prpview}
\begin{proof}
\ndss{\newcommand{\eqnswitch}[2][]{$#2$#1}}
\tufte{\newcommand{\eqnswitch}[2][]{\[#2\]}}
Assume Definition~\ref{dfn:dp} holds.
Consider any
population $\mc{P}$,
index $i$,
data points
$\langle d_1, \ldots, d_n\rangle$ in $\mc{D}^n$
and $d'_i$ in $\mc{D}$,
and output $o$
such that the following holds:
\eqnswitch{\Fr_{\mc{P}}[D_1{=}d_1,\ldots, D_n{=}d_n] > 0}
and
\eqnswitch[.]{\Fr_{\mc{P}}[D_1{=}d_1, \ldots, D_i{=}d'_i, \ldots, D_n{=}d_n] > 0}
Since Definition~\ref{dfn:dp} holds,
\begin{fw}
\begin{align}
\Fr_{\mc{A}}[\mc{A}(\langle d_1, \ldots, d_n\rangle){=}o] &\leq e^\epsilon \Fr_{\mc{A}}[\mc{A}(\langle d_1, \ldots, d'_i, \ldots, d_n\rangle){=}o]\ndss{\notag}
\end{align}
\end{fw}
Letting $O = \mc{A}(D)$ and $D = \langle D_1, \ldots, D_n\rangle$, the above implies
\begin{fw}
\begin{align}
\aln\Fr_{\mc{P},\mc{A}}[O{=}o \given D_1{=}d_1,\ldots, D_n{=}d_n] \brk &\ndss{\:}\leq e^\epsilon * \Fr_{\mc{P},\mc{A}}[O{=}o \given D_1{=}d_1, \ldots, D_i{=}d'_i, \ldots, D_n{=}d_n]\ndss{\notag}
\end{align}
\end{fw}
Thus, Definition~\ref{view:dp-utc} holds.

Assume Definition~\ref{view:dp-utc} holds.
Let $\mc{P}$ be a population that is i.i.d.\@ and assigns non-zero probabilities to all the sequences of $n$ data points.
Consider any
index $i$,
data points
$\langle d_1, \ldots, d_n\rangle$ in $\mc{D}^n$
and $d'_i$ in $\mc{D}$,
and output $o$.
$\mc{P}$ is such that
\eqnswitch{\Fr_{\mc{P}}[D_1{=}d_1,\ldots, D_n{=}d_n] > 0}
and
\eqnswitch{\Fr_{\mc{P}}[D_1{=}d_1, \ldots, D_i{=}d'_i, \ldots, D_n{=}d_n] > 0}
both hold.
Thus, since Definition~\ref{view:dp-utc} holds for $\mc{P}$,
\begin{fw}
\begin{align}
\aln\Fr_{\mc{P},\mc{A}}[O{=}o \given D_1{=}d_1,\ldots, D_n{=}d_n] \brk&\ndss{\:}\leq e^\epsilon * \Fr_{\mc{P},\mc{A}}[O{=}o \given D_1{=}d_1, \ldots, D_i{=}d'_i, \ldots, D_n{=}d_n]\ndss{\notag}
\end{align}
\end{fw}
where $O = \mc{A}(D)$ and $D = \langle D_1, \ldots, D_n\rangle$.
Thus,
\begin{align}
\Fr_{\mc{A}}[\mc{A}(\langle d_1, \ldots, d_n\rangle){=}o] &\leq e^\epsilon \Fr_{\mc{A}}[\mc{A}(\langle d_1, \ldots, d_i, \ldots, d_n\rangle){=}o]\ndss{\notag}
\end{align}
Thus, Definition~\ref{dfn:dp} holds.
\end{proof}

The standard intuition provided for the formulation of differential
privacy found in Definition~\ref{view:dp-utc} is a Bayesian one
in which we think of $\mc{P}$ as being prior information held
by an adversary trying to learn about $D_i$.
We condition upon and fix all the values of
$D_1,\ldots,D_n$ except $D_i$ to model a ``strong adversary'' that knows every
data point except $D_i$, whose value varies in \eqref{eqn:conditioning-on-all}.
As the value of $D_i$ varies, we compare the probabilities of output values $o$.
These probabilities can be thought of as measuring what the adversary
knows about $D_i$ given all the other data points.
The bigger the change in the probabilities as the value of $D_i$
varies, the bigger the flow of information from $D_i$ to $O$.

The origins of this characterization of \DP{} go back to the original
work of Dwork~et~al., who instead call strong adversaries
``informed adversaries''~\cite[App.\,A]{dwork06crypto}.
However, their characterization is somewhat different than what
is now viewed as the strong adversary characterization.
This new characterization has since shown up in numerous places.
For example,
Alvim and Andr\'{e}s rewrite \DP{} this
way~\cite[p.\,5]{alivim11icalp} while %
Yang~et~al.~\cite[Def.\,4]{yang15sigmod} and %
Cuff and Yu~\cite[Def.\,1]{cuff16ccs} even define it thus. %

Despite this intuition, there's no mathematical requirement that we
interpret the probabilities in terms of an adversary's Bayesian
beliefs and we could instead treat them as frequencies over some
population.
In Section~\ref{sec:knowledge}, we return to this issue where
we explicitly mix the two interpretations.
Either way, we term Definition~\ref{view:dp-utc} to be an
\emph{associative characterization} of \DP{} since
\eqref{eqn:conditioning-on-all} compares probabilities that differ in
the value of $D_i$ that is conditioned upon.

While it may seem intuitive that ensuring privacy against such a
``strong'' adversary would imply privacy against other ``weaker''
adversaries that know less, it turns out that the name is misleading.
Suppose we measure privacy in terms of the association between $D_i$ and $O$,
which captures what an adversary learns, as in \eqref{eqn:conditioning-on-all}.
Depending upon the circumstances,
either a more informed ``stronger'' adversary or
a less informed ``weaker'' adversary
will learn more from a data release~\cite{cormode11kdd,kifer11sigmod}.
Intuitively, if the released data is esoteric information and only the informed
adversary has enough context to make use of it, it will learn more.
If, on the other hand, the released data is more basic information
relating something that the
informed adversary already knows but the uninformed one does not, then
the ``weaker'' uninformed one will learn more.

One way to make this issue more precise is to model how informed an adversary is
by the number of data points it knows, that is, the number conditioned
upon.
This leads to Yang~et~al.'s definition of \emph{Bayesian Differential Privacy}~\cite[Def.\,5]{yang15sigmod}.
Despite the name, its probabilities can be interpreted either as Bayesian credences or as frequencies.
For simplicity, we state their definition for just the extreme case where the adversary knows zero data points:

\begin{view}[Bayesian$_0$ Differential Privacy]
\label{view:dp-ucb}
A randomized algorithm $\mc{A}$ is
\emph{$ε$-Bayesian$_0$ differentially private}
if for all population distributions $\mc{P}$,
for all $i$, for all data points
$d_i$ and $d'_i$ in $\mc{D}$,
and for all output values $o$,
if
$\Fr_{\mc{P}}[D_i{=}d_i] > 0$ and
$\Fr_{\mc{P}}[D_i{=}d'_i] > 0$
then
\begin{align}
\Fr_{\mc{P},\mc{A}}[O{=}o \given D_i{=}d_i] ≤ e^ε * \Fr_{\mc{P},\mc{A}}[O{=}o \given D_i{=}d'_i]
\end{align}
where $O = \mc{A}(D)$ and $D = \langle D_1, \ldots, D_n\rangle$.
\end{view}

One might expect that \DP{} would provide
Bayesian Differential Privacy from hearing informal descriptions of
them.
However, Yang~et~al.\@ prove that Bayesian Differential Privacy implies
\DP{} but is strictly stronger~\cite[Thm.\,2]{yang15sigmod}.
Indeed, it was already known that limiting the association between
$D_i$ and the output $O$ requires limiting the associations between
$D_i$ and the other data points~\cite{cormode11kdd,kifer11sigmod}.
Doing so, Yang~et~al.\@ proved that \DP{} implies
Bayesian Differential Privacy under the assumption that the data
points are independent of one another~\cite[Thm.\,1]{yang15sigmod}.
We state the resulting qualified form of \DP{} as
follows:
\begin{view}[Independent Bayesian$_0$ Differential Privacy]
\label{view:dp-ica}
A randomized algorithm $\mc{A}$ is
\emph{$ε$-Bayesian$_0$ differentially private for independent data points}
if for all population distributions $\mc{P}$
such that
for all $i$ and $j$ where $i \neq j$, $D_i$ is independent of $D_j$ conditioned upon the other data points
the following holds:
for all data points
$d_i$ and $d'_i$ in $\mc{D}$,
and for all output values $o$,
if
$\Fr_{\mc{P}}[D_i{=}d_i] > 0$ and
$\Fr_{\mc{P}}[D_i{=}d'_i] > 0$
then
\begin{align}
\Fr_{\mc{P},\mc{A}}[O{=}o \given D_i{=}d_i] ≤ e^ε * \Fr_{\mc{P},\mc{A}}[O{=}o \given D_i{=}d'_i]
\label{eq:view:dp-ica:1}
\end{align}
where $O = \mc{A}(D)$ and $D = \langle D_1, \ldots, D_n\rangle$.
\end{view}

On all the above math, everyone is in agreement, which we summarize in Figure~\ref{fig:ass-dp} and below:
\begin{enumerate}[(a)] %
\item Differential privacy and Strong Adversary Differential Privacy are equivalent,
\item Differential privacy and Independent Bayesian Differential Privacy are equivalent,
\item Bayesian Differential Privacy and related associative properties are strictly stronger than Differential Privacy,
\item\label{enm:strong} If we limit ourselves to strong adversaries, \DP{} and Bayesian Differential Privacy become equivalent, and
\item\label{enm:indep} If we limit ourselves to independent data points, \DP{} and Bayesian Differential Privacy become equivalent.
\end{enumerate}
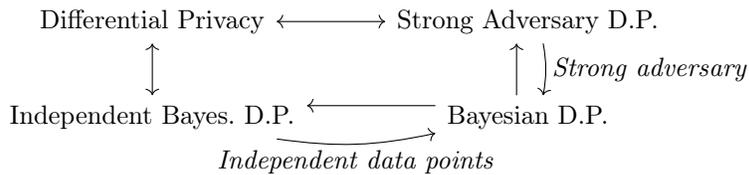
\begin{figure}%
\ndss{\small}
\centering
\begin{DIFnomarkup}
\begin{tikzcd}[math mode=false, row sep=2em, column sep=3em]
Differential Privacy\arrow[r,leftrightarrow]\arrow[d,leftrightarrow] & Strong Adversary D.P.\arrow[d, xshift=1ex, bend left=10, "\emph{Strong adversary}"]\\
Independent Bayes.\@ D.P.\arrow[r, shift right=0ex, bend right=10, "\emph{Independent data points}"']     & Bayesian D.P.\arrow[u, xshift=-1ex]\arrow[l, yshift=1ex]
\end{tikzcd}
\end{DIFnomarkup}
\caption{Relationships between Differential Privacy and Associative
  Characterizations of It.  Arrows show implications.  Curved, labeled
  arrows show, in italics, assumptions required for the implication.  For
  differential privacy to imply Bayesian Differential Privacy, one
  of two assumptions must be made.}
\label{fig:ass-dp}
\end{figure}
More controversially, some papers have pointed to these facts
to say that \DP{} makes
implicit assumptions.
Some have taken \ref{enm:strong} to imply that \DP{}
has an implicit assumption of a strong adversary.
For example, Cuff and Yu's paper states~\cite[p.\,2]{cuff16ccs}: %
\begin{quote}
  The definition of $(\epsilon, \delta)$-DP involves a notion of neighboring
  database instances. Upon examination, one realizes that this has the
  effect of assuming that the adversary has already learned about
  all but one entry in the database and is only trying to gather
  additional information about the remaining entry. We refer to this
  as the strong adversary assumption, which is implicit in the
  definition of differential privacy.
\end{quote}
Others have focused on \ref{enm:indep} and independent data points.
For example,
Liu~et~al.'s paper asserts~\cite[p.\,1]{liu16ndss}:
\begin{quote}
  To provide its guarantees, DP mechanisms assume that the data tuples
  (or records) in the database, each from a different user, are all
  independent.
\end{quote}
Appendix~\ref{app:history-camp2} provides more examples.

Those promoting the original view of \DP{} have
re-asserted that \DP{} was never intended to prevent
all associative, or inferential, privacy threats and that doing so is
impossible~\cite{bassily13focs,kasiviswanathan14jpc,mcsherry16blog1,mcsherry16blog2}.
However, this assertion raises the question: if \DP{} is not
providing some form of association-based inferential privacy, what is
it providing?

\section{A Primer on Causation}
\label{sec:cause}

We believe that the right way of thinking about \DP{}
is that it is providing a causal guarantee.
Before justifying this claim, we will review a framework for precisely reasoning about causation based upon Pearl's~\cite{pearl09book}.
We choose Pearl's since it is the most well known in computer science, but our results can be translated into other frameworks.

To explain causation, let us return to the example of Section~\ref{sec:mot-example}.
Suppose that the statistic being computed is the number of data points showing the genetic disease.
A possible implementation of such a differentially private count algorithm $\mc{A}$ for a fixed number of three data points is
\newcommand{\pos}{\mathtt{pos}}
\newcommand{\negit}{\mathtt{neg}}
\begin{align*}
\mathtt{def}&\ \texttt{prog}_{\mc{A}}(D_1,D_2,D_3):\\
&D := \langle D_1, D_2, D_3\rangle\\
&O := \texttt{Lap}(1/ε) + \sum_{i=1}^3 (1\ \texttt{if}\ D[i] == \pos\ \texttt{else}\ 0)
\end{align*}
It takes in $3$ data points as inputs, representing the statuses reported by survey participants.
It stores them in a database $D$ and
then uses the Laplace Mechanism to provide a differentially private
count of the number of data points recording the status as positive ($\pos$)~\cite[Example~1]{dwork06crypto}.

One could use a tool like the GNU Project Debugger (GDB)
to check the value of a variable as the program executes.
We can think of this as making an observation.
If you observed that $D[3]$ is negative ($\negit$), you would know that $D_3$ and the
third input were also $\negit$.
In a probabilistic setting,
conditioning would
carry out this update in knowledge.

One could also use GDB to intervene on the program's execution and
alter $D[3]$ to be $\pos$.
This would probabilistically increase the output's value.
But would one learn from this that $D_3$ is $\pos$ and no longer $\negit$?
No, since the program uses assignments and not
equalities to shift the value of the right-hand-side variable into the left-hand-side variable.
$D_3$ is a (partial) cause of $D$, but not the other way around.
Altering the value of $D[3]$ only affects variables that it assigns a value to, those they assign values to, and so forth, that is, the ones it causes.
In this example, that is only $O$.
This reflects the difference between association and causation.

More formally, to develop a causal interpretation of \DP{}, we start by replacing the equation $O = \mc{A}(D)$ with a stronger claim.
Such equations say nothing about why this relation holds.
We use a stronger causal relation asserting that the value of the output $O$ is caused by the value of the input $D$, that is, we use a \emph{structural equation}.
We will denote this structural equation by $O := \mc{A}(D)$ since it is closer to an assignment than equality due to its directionality.
To make this more precise, let $\doo(D{=}d)$ denote an intervention setting the value of $D$ to $d$ (Pearl's \emph{do} notation~\cite{pearl09book}).
Using this notation, $\Fr[O{=}o \given \doo(D{=}d)]$ represents what the probability of $O = o$ would be if the value of $D$ were set to $d$ by intervention.
Similar to normal conditioning on $D = d$, $\Fr[O{=}o \given \doo(D{=}d)]$ might not equal $\Fr[O{=}o]$.
However, $\Fr[D{=}d \given \doo(O{=}o)]$ will surely equal $\Fr[D{=}d]$ since $O$ is downstream of $D$, and, thus, changing $O$ has no effects on $D$.

Similarly, we replace $D = \langle D_1, D_2, D_3\rangle$ with $D := \langle D_1, D_2, D_3\rangle$.
That is, we consider the value of the whole database to be caused by the values of the data points and nothing more.
Furthermore, we require that $D_1, D_2, D_3$ only cause $D$ and do not have any other effects.
In particular, we do not allow $D_i$ to affect $D_j$ for $i \neq j$.
Looking at our example program $\texttt{prog}_{\mc{A}}$, this is the case.

This requirement might seem to prevent one person's attribute from affecting another's, for example, preventing a mother's genetic condition from affecting her child's genetic condition.
This is not the case since $D_1, D_2, D_3$ represent the data points provided as inputs to the algorithm and not the actual attributes themselves.
One could model the actual attributes, such as genetics itself, as random variables $R_1, R_2, R_3$ where $D_i := R_i$ for all $i$ and allow $R_i$ to affect $R_j$ without changing how intervening on the $D_i$s works.
For example, $\texttt{prog}_{\mc{A}}$ might be called in the following context:
\begin{align*}
\mathtt{def}\ &\mathtt{prog}_{\text{status}}(R_1,R_3):\\
&R_2 := R_1\\
&D_1 := R_1\\
&D_2 := R_2\\
&D_3 := R_3\\
&\texttt{prog}_{\mc{A}}(D_1,D_2,D_3)
\end{align*}
which does not say how the inputs $R_1$ or $R_3$ are set but does model that
$R_2$ is assigned $R_1$.
We can graphically represent these relationships as a \emph{graphical model}, similar to the one in Figure~\ref{fig:causal-privacy} with $n-2 = 1$ and an intermediate variable $D$ representing the database put between the data points and the output.
Note that while $D_1$ and $D_2$ are associated, equal in fact, neither
causes the other and they can be changed independently of one another,
which can be seen from neither being downstream from the other.

To make the above intuitions about causation formal, we use\tufte{ a slight modification of Pearl's causal models.\footnote{The models we use are suggested by Pearl for handling ``inherent'' randomness~\cite[p.\,220]{pearl09book} and differs from the model he typically uses~\cite[Def.\,7.1.6]{pearl09book} by allowing randomization in the structural equations $F_V$.  We find this randomization helpful for modeling the randomization within the algorithm $\mc{A}$.}
Pearl uses} \emph{structural equation models} (SEMs). %
An SEM $\mc{M} = \langle \mathcal{V}_{\msf{en}}, \mathcal{V}_{\msf{ex}}, \mathcal{E}\rangle$ includes a set of \emph{variables} partitioned into \emph{endogenous} (or dependent) variables $\mathcal{V}_{\msf{en}}$ and \emph{background} (or exogenous, or independent) variables $\mathcal{V}_{\msf{ex}}$.
You can think of the endogenous variables as being those assigned values by the programs above and the background variables as being those provided as inputs to the programs. %
$\mc{M}$ also includes a set $\mathcal{E}$ of structural equations, corresponding to the assignments.
Each endogenous variable $X$ has a structural equation $X := F_X(\vec{Y})$ where
$F_X$ is a possibly randomized function and
$\vec{Y}$ is a list of other variables, modeling the direct causes of $X$.
To avoid circularity, $\vec{Y}$ may not include $X$.
We call the variables $\vec{Y}$ the \emph{parents} of $X$, denoted as $\msf{pa}(X)$.

We limit ourselves to \emph{recursive} SEMs, those in which the variables may be ordered such that all background variables come before all endogenous variables and no variable has a parent that comes before it in the ordering.
We may view such SEMs as similar to a program where the background variables are inputs to the program and the ordering determines the order of assignment statements in the program.
We can make this precise by computing the values of endogenous
variables from the values of the background variables using a method
similar to assigning a semantics to a program.

The only difference is that, rather then a single value, the inputs are assigned probability distributions over values, which allows us to talk about the probabilities of the endogenous variables taking on a value.
Let a \emph{probabilistic SEM} $\langle \mc{M}, \mc{P}\rangle$ be an SEM $\mc{M}$ with a probability distribution $\mc{P}$ over its background variables.
We can raise the structural equations (assignments) to work over $\mc{P}$ instead of a concrete assignment of values.
(Appendix~\ref{app:causation-details} provides details.)

Finally, to define causation, let $\mc{M}$ be an SEM, $Z$ be an endogenous variable of $\mc{M}$, and $z$ be a value that $Z$ can take on.
Pearl defines the \emph{sub-model} $\mc{M}[Z{:=}z]$ to be the SEM that results from replacing the equation $Z := F_Z(\vec{Z})$ in $\mathcal{E}$ of $\mc{M}$ with the equation $Z := z$.
You can think of this as using GDB to assign a value to a variable or as aspect-oriented programming jumping into a function to alter a variable.
The sub-model $\mc{M}[Z{:=}z]$ shows the \emph{effect} of setting $Z$ to $z$.
Let $\Fr_{\olangle \mc{M}, \mc{P}\orangle}[Y{=}y \given \doo(Z{:=}z)]$ be $\Fr_{\olangle \mc{M}[Z{:=}z], \mc{P}\orangle}[Y{=}y]$.
This is well defined even when $\Fr_{\olangle \mc{M}, \mc{P}\orangle}[Z{=}z] = 0$ as long as $z$ is within in the range of values $\mc{Z}$ that $Z$ can take on.

Returning to our example,
let $\mc{M}^{\mc{A}}_{\text{st}}$ be an SEM representing $\texttt{prog}_{\text{status}}$
and $\mc{P}$ be the naturally occurring distribution of data points.
$\Fr_{\olangle \mc{M}^{\mc{A}}_{\text{st}}, \mc{P}\orangle}[O{=}o]$ is the probability of the algorithm's output being $o$ under $\mc{P}$ and coin flips internal to $\mc{A}$.
$\Fr_{\olangle \mc{M}^{\mc{A}}_{\text{st}}, \mc{P}\orangle}[O{=}o \given D_i{=}\pos]$ is that probability conditioned upon seeing $D_1 = \pos$.
$\Fr_{\olangle \mc{M}^{\mc{A}}_{\text{st}}, \mc{P}\orangle}[O{=}o \given \doo(D_1{=}\pos)]$ is that probability given an intervention setting the value of $D_1$ to $\pos$, which is
$\Fr_{\olangle \mc{M}^{\mc{A}}_{\text{st}}[D_i{:=}\pos], \mc{P}\orangle}[O{=}o]$.
$\mc{M}^{\mc{A}}_{\text{st}}[D_1{:=}\pos]$ is the program with the line assigning $R_1$ to $D_1$ replaced with $D_1 := \pos$.
$\Fr_{\olangle \mc{M}^{\mc{A}}_{\text{st}}, \mc{P}\orangle}[O{=}o \given \doo(D_1{=}\pos)]$ depends upon how the intervention on $D_1$ will flow downstream to\tufte{ $D$ and then} $O$.

This probability differs from the conditional probability in that setting $D_1$ to $\pos$ provides no information about $D_j$ for $j \neq 1$, whereas if $D_1$ and $D_j$ are associated, then seeing the value $D_1$ does provide information about $D_j$.
Intuitively, this lack of information is because the artificial setting of $D_1$ to $\pos$ has no causal influence on $D_j$ due to the data points not affecting one another and the artificial setting, by being artificial, tells us nothing about the associations found in the naturally occurring world.
On the other hand, artificially setting the attribute itself $R_1$ to $\pos$ will provide information about $D_2$ since $R_1$ has an effect on $D_2$ in addition to $D_1$. %
A second difference is that $\Fr_{\olangle \mc{M}^{\mc{A}}_{\text{st}}, \mc{P}\orangle}[O{=}o \given \doo(D_i{=}d_i)]$ is defined even when $\Fr_{\olangle \mc{M}^{\mc{A}}_{\text{st}}, \mc{P}\orangle}[D_i{=}d_i] = 0$.

Importantly, interventions on a data point $D_i$ do not model modifying the attributes they record nor affect other inputs.
Instead, interventions on $D_i$ model changing the values provided as inputs to the algorithm, which can be changed without affecting the attributes or other inputs.
This corresponds to an \emph{atomicity} property: the inputs $D_i$ are causally isolated from one another and they can be intervened upon separately.

Making the distinction between the inputs $D_i$ and the attributes $R_i$ might seem nitpicky, but it is key to understanding \DP{}.
Recall that its motivation is to make people comfortable with truthfully sharing data instead of withholding it or lying, which is an acknowledgment that the inputs people provide might not be the same as the attributes they describe.
Furthermore, that changing inputs do not change attributes or other inputs is a reflection of how the program works.
It is not an implicit or hidden assumption of independence; it is a fact about the program analyzed.

\section{Differential Privacy as Causation}
\label{sec:dp-causal}

Due to differential privacy's behavior on associated inputs and its requirement of considering zero-probability database values, \DP{} is not a straightforward property about the independence or the degree of association of the database and the algorithm's output.
The would-be conditioning upon zero-probability values corresponds to a form of counterfactual reasoning asking what the algorithm would have performed had the database taken on a particular value that it might never actually take on.
Experiments with such counterfactuals, which may never naturally occur, form the core of causation.
The behavior of \DP{} on associated inputs corresponds to the atomicity property found in causal reasoning, that one can change the value of an input without changing the values of other inputs.
With these motivations, we will show that \DP{} is equivalent to a causal property that makes the change in a single data point explicit.

\subsection{With the Whole Database}
\label{sec:dp-causal-whole}

We first
show an equivalence between \DP{} and a causal property on the whole database to echo Strong Adversary Differential Privacy (Def.~\ref{view:dp-utc}).
To draw out the parallels between the associative and causal properties, we quantify over all populations as we did in Definition~\ref{view:dp-utc}, but as we will see, doing so is not necessary.

Let $\mc{M}^{\mc{A}}$ be an SEM modeling a slightly modified version of
$\texttt{prog}_{\text{status}}$ that lacks the first assignment and treats all of any fixed number of attributes $R_i$ as inputs (i.e., as exogenous variables) with $D_i := R_i$.
(Appendix~\ref{app:causation-details} provides details.)
We could instead use a version of $\mc{M}^{\mc{A}}$ that also accounts for $D_i$ possibly being assigned a value other than $R_i$ to model withholding an attribute's actual value.
While the proofs would become more complex, the results would remain the same since we only intervene on the $D_i$ and not the $R_i$.

\begin{view}[Universal Whole Database Intervention D.P.]%
\label{view:dp-tc-univ}
A randomized algorithm $\mc{A}$ is
\emph{$ε$-differentially private as universal intervention on the whole database}
if for all population distributions $\mc{P}$, for all $i$, for all data points
$d_1, \ldots, d_n$ in $\mc{D}^n$
and $d'_i$ in $\mc{D}$,
and for all output values $o$,
\begin{fw}
\begin{align}
\aln\Fr_{\olangle\mc{M}^\mc{A},\mc{P}\orangle}[O{=}o \given \doo(D_1{=}d_1,\ldots, D_n{=}d_n)] \brkalnaln{} ≤ e^ε * \Fr_{\olangle\mc{M}^\mc{A},\mc{P}\orangle}[O{=}o \given \doo(D_1{=}d_1, \ldots, D_i{=}d'_i, \ldots, D_n{=}d_n)] \ndss{\notag}
\end{align}
\end{fw}
where $O := \mc{A}(D)$ and $D := \langle D_1, \ldots, D_n\rangle$.
\end{view}

\begin{prpview}\label{prpview:dp-tc-univ}
Definitions~\ref{dfn:dp} and~\ref{view:dp-tc-univ} are equivalent.
\end{prpview}
\begin{proof}

  Pearl's Property~1 says that conditioning upon all the parents
  of a variable and causally intervening upon them all yields the same
  probability~\cite[p.\,24]{pearl09book}.
  Intuitively, this is for the same reason that Strong Adversary
  Differential Privacy is equivalent to \DP{}: it blocks
  other paths of influence from one data point to the output via another data point by
  fixing all the data points.

  We can apply Property~1 since all the $D_i$s are being intervened upon and they make up all the parents of $D$.
  We can apply it again on $D$ and $O$.
  We then get that
$\Fr_{\olangle\mc{M}^\mc{A},\mc{P}\orangle}[O{=}o \given \doo(D_1{=}d_1,\ldots, D_n{=}d_n)]$ is equal to $\Fr_{\olangle\mc{M}^\mc{A},\mc{P}\orangle}[O{=}o \given D_1{=}d_1,\ldots, D_n{=}d_n]$,
that is to Strong Adversary Differential Privacy,
which we already know to be equivalent to \DP{} by Proposition~\ref{prpview:dp-utc}.
\tufte{\par Alternately, we can grind out the calculations.
See Lemma~\ref{lem:doing-on-all} in Appendix~\ref{app:causation-details}.}
\end{proof}

Notice that this causal property is simpler than the associative one in that it does not need qualifications around zero-probability data points because we can causally fix data points to values with zero probability.
In fact, the population distribution $\mc{P}$ did not matter at all since intervening upon all the data points makes it irrelevant, intuitively by overwriting it.
For this reason, we could instead look at any population, such as the naturally occurring one
(or even elide it from the definition altogether, as in Definition~\ref{dfn:dp}, if we are not too picky about formalism).
Next, we state such a simplified definition.

\begin{view}[Whole Database Intervention D.P.]%
\label{dfn:dp-tc}
Given a population distribution $\mc{P}$,
a randomized algorithm $\mc{A}$ is
\emph{$ε$-differentially private as intervention on the whole database for $\mc{P}$}
if for all $i$, for all data points
$d_1, \ldots, d_n$ in $\mc{D}^n$
and $d'_i$ in $\mc{D}$,
and for all output values $o$,
\begin{fw}
\begin{align}
\aln\Fr_{\olangle\mc{M}^\mc{A},\mc{P}\orangle}[O{=}o \given \doo(D_1{=}d_1,\ldots, D_n{=}d_n)] \brk &\ndss{\:}≤ e^ε * \Fr_{\olangle\mc{M}^\mc{A},\mc{P}\orangle}[O{=}o \given \doo(D_1{=}d_1, \ldots, D_i{=}d'_i, \ldots, D_n{=}d_n)]\ndss{\notag}
\end{align}
\end{fw}
where $O := \mc{A}(D)$ and $D := \langle D_1, \ldots, D_n\rangle$.
\end{view}

\begin{prp}\label{prp:dp-tc}
Definitions~\ref{dfn:dp} and~\ref{dfn:dp-tc} are equivalent.
\end{prp}
\begin{proof}
The proof follows in the same manner as Proposition~\ref{prpview:dp-tc-univ} since that proof applies to all population distributions $\mc{P}$.
\end{proof}

\subsection{With a Single Data Point}
\label{sec:dp-causal-single}

Definitions~\ref{view:dp-tc-univ} and~\ref{dfn:dp-tc}, by fixing every data point, do not capture the local nature of the decision facing a single potential survey participant.
We can define a notion similar to \DP{} that uses a causal intervention on a single data point as follows:
\begin{chr}[Data-point Intervention D.P.]%
\label{chr:dp-spc}
Given a population $\mc{P}$, a randomized algorithm $\mc{A}$ is
\emph{$ε$-differentially private as intervention on a data point for $\mc{P}$}
if for all $i$, for all data points
$d_i$ and $d'_i$ in $\mc{D}$,
and for all output values $o$,
\begin{align}
\Fr_{\mc{M}^\mc{A},\mc{P}}[O{=}o \given \doo(D_i{=}d_i)] ≤ e^ε \Fr_{\mc{M}^\mc{A},\mc{P}}[O{=}o \given \doo(D_i{=}d'_i)] \ndss{\notag}
\end{align}
where $O := \mc{A}(D)$ and $D := \langle D_1, \ldots, D_n\rangle$.
\end{chr}

This definition is strictly weaker than \DP{}.
The reason is similar to why we had to quantify over all distributions $\mc{P}$ with Strong Adversary Differential Privacy.
In both cases, we can give a counterexample with a population $\mc{P}$ that hides the effects of a possible value of the data point by assigning the value a probability of zero.
For the associative definition, the counterexample involves only a single data point (Appendix~\ref{app:chr-csb}).
However, for this causal definition, the counterexample has to have two data points.
The reason is that, since the $\doo$ operation acts on a single data point at a time, it can flush out the effects of a single zero-probability value but not the interactions between two zero-probability values.

\begin{prpchr}
\label{prpchr:dp-spc}
Definition~\ref{dfn:dp} implies Definition~\ref{chr:dp-spc}, but not the other way around.
\end{prpchr}
\begin{proof}
W.l.o.g., assume $i = n$.
Assume Definition~\ref{dfn:dp} holds:
\begin{fw}
\begin{align}
\aln\Fr_{\mc{A}}[\mc{A}(\langle d_1, \ldots, d_{n-1}, d_n\rangle){=}o] \brkaln{XXXXXXXXX} ≤ e^ε * \Fr_{\mc{A}}[\mc{A}(\langle d_1, \ldots, d_{n-1}, d'_n\rangle){=}o] \ndss{\notag}
\end{align}
\end{fw}
 for all $o$ in $\mc{O}$, $\langle d_1, \ldots, d_n\rangle$ in $\mc{D}^n$, and $d'_n$ in $\mc{D}$.
This implies that for any $\mc{P}$,
\begin{fw}
\begin{align}
\aln\Fr_{\mc{P}}\left[ \wedge_{i=1}^{n-1} D_i{=}d_i \right] * \Fr_{\mc{A}}[\mc{A}(\langle d_1, \ldots, d_{n-1}, d_n\rangle){=}o] \brkalnaln{}%
≤ e^ε * \Fr_{\mc{P}}\left[ \wedge_{i=1}^{n-1} D_i{=}d_i \right] * \Fr_{\mc{A}}[\mc{A}(\langle d_1, \ldots, d_{n-1}, d'_n\rangle){=}o] \ndss{\notag}
\end{align}
\end{fw}
for all $o$ in $\mc{O}$, $d_1, \ldots, d_n$ in $\mc{D}^n$, and $d'_n$ in $\mc{D}$.
Thus,
\begin{fw}
\ndss{\small}
\begin{align}
\aln\bigsum{\langle d_1,\ldots,d_{n-1}\rangle \in \mc{D}^{n-1}} \Fr_{\mc{P}}\left[ \wedge_{i=1}^{n-1} D_i{=}d_i \right] * \Fr_{\mc{A}}[\mc{A}(\langle d_1, \ldots, d_{n-1}, d_n\rangle){=}o] \brkalnaln{x}≤ \bigsum{\langle d_1,\ldots,d_{n-1}\rangle \in \mc{D}^{n-1}} e^ε * \Fr_{\mc{P}}\left[ \wedge_{i=1}^{n-1} D_i{=}d_i \right] * \Fr_{\mc{A}}[\mc{A}(\langle d_1, \ldots, d_{n-1}, d'_n\rangle){=}o]\ndss{\notag}\\
\aln\bigsum{\langle d_1,\ldots,d_{n-1}\rangle \in \mc{D}^{n-1}} \Fr_{\mc{P}}\left[ \wedge_{i=1}^{n-1} D_i{=}d_i \right] * \Fr_{\mc{A}}[\mc{A}(\langle d_1, \ldots, d_{n-1}, d_n\rangle){=}o] \brkalnaln{x}≤ e^ε \tufte{*} \bigsum{\langle d_1,\ldots,d_{n-1}\rangle \in \mc{D}^{n-1}} \Fr_{\mc{P}}\left[ \wedge_{i=1}^{n-1} D_i{=}d_i \right] * \Fr_{\mc{A}}[\mc{A}(\langle d_1, \ldots, d_{n-1}, d'_n \rangle){=}o]\ndss{\notag}\\
\aln\Fr_{\mc{M}^{\mc{A}},\mc{P}}[O{=}o \given \doo(D_n{=}d_n)] \tufte{&}≤ e^ε * \Fr_{\mc{M}^{\mc{A}},\mc{P}}[O{=}o \given \doo(D_n{=}d'_n)]\ndss{\notag}
\end{align}
\end{fw}
where the last line follows from Lemma~\ref{lem:doing-on-one} in Appendix~\ref{app:causation-details}.

Definition~\ref{chr:dp-spc} is, however, weaker than \DP{}.
Consider the case of a database holding two data points whose value could be $0$, $1$, or $2$.
Suppose the population $\mc{P}$ is such that $\Fr_{\mc{P}}[D_1{=}2] = 0$ and $\Fr_{\mc{P}}[D_2{=}2] = 0$.
Consider an algorithm $\mc{A}$ such that
\begin{fw}
\begin{align}
\Fr_{\mc{A}}[\mc{A}(\langle 2,2\rangle){=}0] &= 1   & \Fr_{\mc{A}}[\mc{A}(\langle 2,2\rangle){=}1] &= 0 \ndss{\notag}\\
\Fr_{\mc{A}}[\mc{A}(\langle d_1,d_2\rangle){=}0] &= 1/2 & \Fr_{\mc{A}}[\mc{A}(\langle d_1,d_2\rangle){=}1] &= 1/2
\ndss{\notag}
\tufte{&&\text{when $d_1 \neq 2$ or $d_2 \neq 2$}}%
\end{align}
\ndss{when $d_1 \neq 2$ or $d_2 \neq 2$.}
\end{fw}
The algorithm does not satisfy Definition~\ref{dfn:dp} due to its behavior when both of the inputs are $2$.
However, using Lemma~\ref{lem:doing-on-one} in Appendix~\ref{app:causation-details},
\begin{DIFnomarkup}
\end{DIFnomarkup}
\[ \Fr_{\mc{M}^{\mc{A}}, \mc{P}}[O{=}o \given \doo(D_1{=}d'_1)] = 1/2 \]
for all $o$ and $d'_1$ since $\Fr_{\mc{P}}[D_2{=}2] = 0$.
A similar result holds switching the roles of $D_1$ and $D_2$.
Thus, the algorithm satisfies Definition~\ref{chr:dp-spc} for $\mc{P}$ but not Definition~\ref{dfn:dp}.
\end{proof}

Despite being only implied by, not equivalent to, \DP{}, Definition~\ref{chr:dp-spc} captures the intuition behind
the characterization~$(*)$ of \DP{} that
``changing a single individual’s data in the database leads to a small change in the distribution on outputs''~\cite[p.\,2]{kasiviswanathan14jpc}.
To get an equivalence, we can quantify over all populations as we did to get an equivalence for association, but this time we need not worry about zero-probability data points or independence.
This simplifies the definition and makes it a more natural characterization of \DP{}.

\begin{view}[Universal Data-point Intervention D.P.]%
\label{view:dp-spca}
A randomized algorithm $\mc{A}$ is
\emph{$ε$-differentially private as universal intervention on a data point}
if
for all population distributions $\mc{P}$,
for all $i$, for all data points
$d_i$ and $d'_i$ in $\mc{D}$,
and for all output values $o$,
\begin{align}
\Fr_{\mc{M}^{\mc{A}}, \mc{P}}[O{=}o \given \doo(D_i{=}d_i)] ≤ e^ε \Fr_{\mc{M}^{\mc{A}}, \mc{P}}[O{=}o \given \doo(D_i{=}d'_i)] \ndss{\notag}
\end{align}
where $O := \mc{A}(D)$ and $D := \langle D_1, \ldots, D_n\rangle$.
\end{view}

\begin{prpview}\label{prpview:dp-spca}
Definitions~\ref{dfn:dp} and~\ref{view:dp-spca} are equivalent.
\end{prpview}
\begin{proof}
That Definition~\ref{dfn:dp} implies \ref{view:dp-spca} follows from Proposition~\ref{prpchr:dp-spc}.

Assume Definition~\ref{view:dp-spca} holds.
W.l.o.g.\@, assume $i = n$.
Then, for all $\mc{P}$, $o$ in $\mc{O}$, and $d'_n$ in $\mc{D}$,
\begin{fw}
\begin{align}
\aln\Fr_{\mc{M}^{\mc{A}}, \mc{P}}[O{=}o \given \doo(D_i{=}d_i)] ≤ e^ε * \Fr_{\mc{M}^{\mc{A}}, \mc{P}}[O{=}o \given \doo(D_i{=}d'_i)] \ndss{\notag}\\
\aln\bigsum{\langle d_1,\ldots,d_{n-1}\rangle \in \mc{D}^{n-1}} \Fr_{\mc{P}}\left[ \wedge_{i=1}^{n-1} D_i{=}d_i \right] * \Fr_{\mc{A}}[\mc{A}(\langle d_1, \ldots, d_{n-1}, d_n\rangle){=}o]  \brkalnaln{}≤ e^ε  \bigsum{\langle d_1,\ldots,d_{n-1}\rangle \in \mc{D}^{n-1}} \Fr_{\mc{P}}\left[ \wedge_{i=1}^{n-1} D_i{=}d_i \right]  \Fr_{\mc{A}}[\mc{A}(\langle d_1, \ldots, d_{n-1}, d'_n\rangle){=}o]\label{eqn:dp-spca-expanded}
\end{align}
\end{fw}
follows from Lemma~\ref{lem:doing-on-one} in Appendix~\ref{app:causation-details}.

For any $d^\dagger_1, \ldots, d^\dagger_{n-1}$ in $\mc{D}^{n-1}$,
let $\mc{P}^{d^\dagger_1, \ldots, d^\dagger_{n-1}}$ be such that
\begin{align}
\Fr_{\mc{P}^{d^\dagger_1, \ldots, d^\dagger_{n-1}}}\left[ \wedge_{i=1}^{n-1} D_i{=}d^\dagger_i \right] &= 1
\end{align}

For any $d^\dagger_1, \ldots, d^\dagger_n$ in $\mc{D}^n$ and $d'_n$ in $\mc{D}$,
\eqref{eqn:dp-spca-expanded} implies
{%
\begin{align}
&\bigsum{\langle d_1,\ldots,d_{n-1}\rangle \in \mc{D}^{n-1}} \Fr_{\mc{P}^{d^\dagger_1,\ldots,d^\dagger_{n-1}}}\left[ \wedge_{i=1}^{n-1} D_i{=}d_i \right] \Fr_{\mc{A}}[\mc{A}(\langle d_1, \smldots, d_{n-1}, d^\dagger_n\rangle){=}o] \notag\\
&≤ e^ε \bigsum{\langle d_1,\ldots,d_{n-1}\rangle \in \mc{D}^{n-1}} \Fr_{\mc{P}^{d^\dagger_1,\ldots,d^\dagger_{n-1}}}\left[ \wedge_{i=1}^{n-1} D_i{=}d_i \right] \ndss{\notag\\[-2.5ex] & \phantom{XXXXXXXXXXXX}} * \Fr_{\mc{A}}[\mc{A}(\langle d_1, \smldots, d_{n-1}, d'_n\rangle){=}o] \ndss{\notag}
\end{align}
}
Thus,
\begin{align}
\aln\Fr_{\mc{A}}[\mc{A}(\langle d^\dagger_1, \ldots, d^\dagger_{n-1}, d^\dagger_n\rangle){=}o]  \brkalnaln{XXXXXXXX}≤ e^ε \Fr_{\mc{A}}[\mc{A}(\langle d^\dagger_1, \ldots, d^\dagger_{n-1}, d'_n\rangle){=}o] \ndss{\notag}
\end{align}
since both sides has a non-zero probability for
\[ \Fr_{\mc{P}^{d^\dagger_1,\ldots,d^\dagger_{n-1}}}\left[ \wedge_{i=1}^{n-1} D_i{=}d_i \right] \]
at only the\tufte{ single} sequence of data point values $d^\dagger_1, \ldots, d^\dagger_{n-1}$.
\end{proof}

\section{Bounding Effects: Generalizing D.P., Understanding Alternatives}
\label{sec:brp}

To recap, we have shown that reasoning about \DP{} as a
causal property is more straightforward than reasoning about it
as an associative property.
Still, one might wonder, Why express \DP{} in either
form? Why not just stick with its even simpler expression in
terms of functions in Definition~\ref{dfn:dp}?

In this section, we show what is gained by the causal view.  We show
that \DP{} bounds a general notion of
\emph{effect size}.
Essentially, \DP{} limits the causal consequences of a decision
to contribute data to a data set.
If the consequences are small,
then an individual will need less encouragement (e.g., financial incentives)
to set aside privacy concerns.

We show that this general notion can also capture alternative privacy
definitions, including some arising from concerns over dependent data points.
A common causal framework allows us to precisely compare these definitions.

\subsection{Bounded Relative Probability (\textsc{brp})}

Generalizing from the decision to participate in a data set, we define a
more general notation for any two random variables $X$ and $Y$.  To do so, we
need a description of how $X$ and $Y$ relate to one another.  Recall that a
probabilistic SEM $\langle \mc{M}, \mc{P}\rangle$ shows the causal and
statistical relations between random variables by providing a list of
structural equations $\mc{M}$ and a distribution $\mc{P}$ over variables not defined
in terms of others (exogenous variables).  (See
Appendix~\ref{app:causation-details} for details.)

We will measure the size of the effects of $X$ on $Y$ using
\emph{relative probabilities}, better known as \emph{relative risk} and
as \emph{risk ratio} with clinical studies of risks in mind.
For three (binary) propositions $\rho$, $\phi$, and $\psi$,
let
\begin{align}
\rp_{\olangle \mc{M},\mc{P}\orangle}(\rho, \phi, \psi) &= \frac{\Fr_{\olangle \mc{M},\mc{P}\orangle}[\rho \given \doo(\phi)]}{\Fr_{\olangle \mc{M},\mc{P}\orangle}[\rho \given \doo(\psi)]} \ndss{\notag}
\end{align}
denote the relative probability.
(Some authors also allow using conditioning instead of interventions.)
For two random variables $X$ and $Y$, we can characterize the maximum
effect of $X$ on $Y$ as
\begin{align}
\mrp_{\mc{M},\mc{P}}(Y, X) &= \max_{y, x_1, x_2} \rp_{\mc{M},\mc{P}}(Y{=}y, X{=}x_1, X{=}x_2)  \ndss{\notag}
\end{align}
Expanding these definitions out shows that $\epsilon$-differential
privacy places a bound on the maximum of the maximum relative
probabilities:
\begin{align}
\max_{\mc{P}, i} \mrp_{\mc{M}, \mc{P}}(O, D_i) \leq e^\epsilon  \ndss{\notag}
\end{align}
where $\mc{M}$ describes the differentially private algorithm $\mc{A}$.
Note that our use of maximization is similar Yang~et~al.~\cite[p.\,749, Def.~4]{yang15sigmod}, which we
quote in Section~\ref{sec:dp-associative}.

With this in mind, we propose to use $\mrp$ for a general purpose
effect-size restriction:
\begin{dfn}[BRP]
A causal system described by $\mc{M}$ has
\emph{$\epsilon$-bounded relative probability (\textsc{brp}) for $X$ to $Y$}
iff
\[ \max_{\mc{P}} \mrp_{\mc{M}, \mc{P}}(Y, X) \leq e^\epsilon \]
\end{dfn}

Differential privacy is equivalent to requiring $\epsilon$-\textsc{brp}
for all data points $D_i$.

\subsection{Composition}

\textsc{Brp} enjoys many of the same properties as \DP{}.
Recall that \DP{} has additive sequential composition for using two differentially private algorithms one after the next, even if the second is selected using the output of the first~\cite{mcsherry2007mechanism}.
Similarly, \textsc{brp} has additive sequential composition for two random variables.

To model the second output $Z$ depending upon the first $Y$, but not the other way around,
we say random variables $X$, $Y$, and $Z$ are \emph{in sequence} if
$X$ may affect $Y$ and $Z$, and $Y$ may affect $Z$, but $Z$
may not affect $X$ nor $Y$, and $Y$ may not affect $X$.
That is,
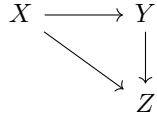
\begin{center}
\begin{DIFnomarkup}
\begin{tikzcd}[math mode=true, row sep=2em, column sep=3em]
   X \arrow[dr]\arrow[r] & Y\arrow[d] \\
                         & Z
\end{tikzcd}
\end{DIFnomarkup}
\end{center}

To model that the second output $Z$ could be computed with one of any of a set of algorithms but that each of algorithm has a bounded effect from $X$ to $Z$, we look at $Z$'s behavior in sub-models $\mc{M}[Y := y]$ where each setting of $Y$ corresponds to a selecting one available algorithm.

\begin{thm}\label{thm:brp-compo}
For any SEM $\mc{M}$ such that $X$, $Y$, and $Z$ are in sequence and the parents of $Z$ are $\{X, Y\}$,
if $X$ has $\epsilon_1$-\textsc{brp} to $Y$ in $\mc{M}$
and $\epsilon_2$-\textsc{brp} to $Z$ in $\mc{M}[Y := y]$ for all $y$ in $\mc{Y}$, then
$X$ has $(\epsilon_1+\epsilon_2)$-\textsc{brp} to $\langle Y, Z\rangle$ in $\mc{M}$.
\end{thm}
\begin{proof}
Consider any probability distribution $\mc{P}$, $x$ and $x'$ in $\mc{X}$, $y$ in $\mc{Y}$, and $z$ in $\mc{Z}$.
Since the effect of $X$ on $Y$ is bounded by $\epsilon_1$-\textsc{brp},
\[ \Fr_{\mc{M},\mc{P}}[Y{=}y \given \doo(X{=}x)] \leq e^{\epsilon_1} \Fr_{\mc{M},\mc{P}}[Y{=}y \given \doo(X{=}x')] \]

Since the parents of $Z$ are $\{X, Y\}$, Pearl's Property~1~\cite[p.\,24]{pearl09book}
shows that
for any $y$ such that $\Fr_{\mc{M},\mc{P}}[Y{=}y] > 0$,
\begin{multline*}
\Fr_{\mc{M},\mc{P}}[Z{=}z \given Y{=}y, \doo(X{=}x)] \brk= \Fr_{\mc{M},\mc{P}}[Z{=}z \given \doo(Y{=}y), \doo(X{=}x)]
\end{multline*}
Since there's $\epsilon_2$-\textsc{brp} from $X$ to $Z$ in $M[Y := y]$ for all $y$, this implies that
\begin{multline*}
\Fr_{\mc{M},\mc{P}}[Z{=}z \given Y{=}y, \doo(X{=}x)] \brk\leq e^{\epsilon_2} \Fr_{\mc{M},\mc{P}}[Z{=}z \given Y{=}y, \doo(X{=}x')]
\end{multline*}

Thus, %
\begin{align*}
&\Fr_{\mc{M},\mc{P}}[\langle Y, Z\rangle = \langle y, z\rangle \given \doo(X{=}x)]\\
&= \Fr_{\mc{M},\mc{P}}[Z {=} z \given Y{=}y, \doo(X{=}x)] \tufte{*} \Fr_{\mc{M},\mc{P}}[Y{=}y \given \doo(X{=}x)]\\
&\leq e^{\epsilon_2} \Fr_{\mc{M},\mc{P}}[Z {=} z \given Y{=}y, \doo(X{=}x')]\brkaln{XXXXXXXXXXX} * e^{\epsilon_1} \Fr_{\mc{M},\mc{P}}[Y{=}y \given \doo(X{=}x')]\\
&= e^{\epsilon_1+\epsilon_2} \Fr_{\mc{M},\mc{P}}[\langle Y, Z\rangle = \langle y, z\rangle \given \doo(X{=}x')] \qedhere
\end{align*}
\end{proof}

We can generalize this theorem for $Z$ having additional parents by requiring \textsc{brp} for all of their values as well.

The special case of this theorem where $\epsilon_2 = 0$ is known as the \emph{postprocessing} condition:
\begin{center}
\begin{tikzcd}[math mode=true, row sep=2em, column sep=3em]
  X \arrow[r, "\epsilon"] & Y \arrow[r] & Z
\end{tikzcd}
\end{center}
For this causal diagram, Theorem~\ref{thm:brp-compo} ensures that if the arrow from $X$ to $Y$ is $\epsilon$-\textsc{brp},
then any subsequent consequence $Z$ of $Y$ is also going to be $\epsilon$-\textsc{brp}.
This captures the central intuition behind \DP{} and $\textsc{brp}$ that
they limit any downstream causal consequences of a variable $X$.

\subsection{Application}

While the explicit causal reasoning in \textsc{brp} can sharpen our
intuitions about privacy, \textsc{brp} is not itself a privacy
definition.
Only some choices of variables to bound yield reasonable privacy
guarantees.
Below, we use \textsc{brp} to express some of well known variations of \DP{}.
Doing so both shows some reasonable ways of using \textsc{brp} to provide privacy guarantees and demonstrates that \textsc{brp} provides a common framework for precisely stating and comparing these variations.

First, consider the randomized response method of providing privacy in which each survey
participant adds noise to his own response before
responding~\cite{warner65asa}.
Let each person's actual attribute be $R_i$, let the noisy response he provides be $D_i$, and let $O$ be the output computed from all the $D_i$.
Unlike with (standard) \DP{}, the causal path from $D_i$ to $O$ has unbounded \textsc{brp}, may not
contain any random algorithms, and misses the privacy protection
altogether.
Similarly, the path from $R_i$ and $O$ has unbounded \textsc{brp} due to the possibility of the $R_i$ having effects upon one another.
However, the randomized response method does ensure $\epsilon$-\textsc{brp} from $R_i$ to $D_i$ for all $i$ where $\epsilon$ depends upon the amount of noise added to $D_i$.

Second, we consider \emph{group privacy}, the idea that a group of individuals may be so closely related that their privacy is intertwined.
Differential privacy approaches group privacy by summing the privacy losses, measured in terms of $\epsilon$, of each individual in the group~\cite[p.\,9]{dwork06icalp}.
Similarly, we can add the relative probabilities of multiple random variables to get a total effect size.
Alternately, \textsc{brp} can easily be extended to measure simultaneous joint interventions by using multiple instances of the $\doo$ operator.
The total effect size may be larger than the joint effect size since,
in cases where the intervened upon variables affect one another, interventions on a downstream variable can mask interventions on its parents.
Returning to the example of Section~\ref{sec:mot-example}, the total effect for both Ada's attribute $R_1$ and Byron's $R_2$ is $3\epsilon$ with $2\epsilon$ of that coming from $R_1$.
However, the joint effect is $2\epsilon$ since $R_1$ achieved half of its effect via $R_2$.
In examples like this where the variables correspond to different moral entities, the total effect size strikes us as more reasonable since it accounts for both Ada and Byron experiencing a privacy loss.
If on the other hand, the variables correspond to a single topic about a single person, such as weight and waist size, then the joint effect size seems more reasonable.
However, we see this choice as under explored since it does not emerge for \DP{} given that data points cannot not affect one another.

Third,
we consider a line of papers providing definitions of privacy that
account for dependencies between data points, but which are ambiguous about association versus causation~\cite{chen14vldbj,zhu15tifs,liu16ndss}.
For example, Liu~et~al.\@ use the word ``cause''
in a central definition of their work~\cite[Def.\,3]{liu16ndss},
but do no causal modeling, instead using a joint
probability distribution to model just associations %
in their adversary model~\cite[\S{}3]{liu16ndss}.
Using causal modeling and \textsc{brp}
would allow them to actually model causation instead of approximating it with associations, or, if associations really is what they wish to model, would provide a foil making their goals more clear.

Fourth, as a more complex example, Kifer and Machanavajjhala consider applying
\DP{} to social networks~\cite[\S{}3]{kifer11sigmod}.
They note that \DP{} applied to a network is typically
taken to mean either considering nodes or edges labeled with an individual's id $i$ in the network as that individual's data
point $D_i$, but that participation in a social network is likely to
leave far more evidence than just %
those nodes and edges.
They consider an example in which Bob joins a network and introduces
Alice and Charlie to one another, leading them to create an edge
between them that does not involve Bob.
Arguably, protecting Bob's privacy requires counting this edge as Bob's as well despite
neither edge nor node \DP{} doing so.
To capture this requirement, they distinguish between differential
privacy's deleting of data points from a data set and their desire to
``hide the \emph{evidence of participation}''~\cite[\S{}2.2.1]{kifer11sigmod}.

Because ``It is difficult to formulate a general and
formal definition for what evidence of participation
means''~\cite[\S{}3, p.\,5]{kifer11sigmod}, they use correlations
in its place for modeling public health and census records~\cite[\S\S{}2.1.3, 2.2, 4.1 \& 4.3.1]{kifer11sigmod}.
However, for modeling social networks, they use statistical models that
they interpret as providing ``a measure of the causal influence of Bob's edge'', that is, informal causal models~\cite[\S{}3, p.\,6]{kifer11sigmod}.

We believe that the causal framework presented herein provides the
necessary mathematical tools to precisely reason about evidence of participation.
Causal models would allow them to precisely state which
aspects of the system they wish to protect, for example, by requiring that
Bob's joining the network should have a bounded effect upon a data
release.
While accurately modeling a social process is a difficult task, at least
the requirement is clearly stated, allowing us to return to empirical
work.
Furthermore, such formalism can allow for multiple models to be
considered and we can demand privacy under each of them, and erring on
the side of safety by over-estimating effect sizes remains an option.

Finally, causal modeling can make the choices between privacy notions more clear.
The distinction between \emph{direct} and \emph{indirect}
effects~\cite{pearl01uai} can model the difference between node
privacy, which only captures the direct effects of joining a social network, and all of
the evidence of participation, which includes hard-to-model indirect
effects.
Edge privacy captures the direct effect of posting additional content.
Given that Facebook has reached near universal membership but worries about disengagement, this effect might be the more concerning one from paractical perspective.

\section{Restrictions on Knowledge}
\label{sec:knowledge}

Privacy is often thought of as preventing an adversary from learning sensitive information.
To make this intuition precise, we can model an adversary's beliefs
using Bayesian probabilities, or \emph{credences}.
We denote them with $\Cr$, instead of $\Fr$, which
we have been using to denote
natural frequencies over outcomes without regard to any agent's beliefs. %
We denote the adversary's background
knowledge as $B$. The knowledge of an adversary about the database
$D$ after observing the output can be expressed as $\Cr[D{=}d \given O{=}o, B]$. A natural
privacy property, termed \emph{statistical nondisclosure} by Dalenius~\cite{dalenius77statistik},
requires that
$\Cr[D{=}d \given O{=}o, B] = \Cr[D{=}d \given B]$, that is, that the beliefs
about the database before and after observing the output are the same.

This requirement limiting the difference between prior and posterior
beliefs
has been shown to be impossible to
achieve under arbitrary background knowledge by Dwork and Naor, even for approximate relaxations of
statistical nondisclosure, as long as the output provides some information~\cite{dwork08jpc}.
As \DP{} also falls under the purview of
this impossibility result, it only provides this associative guarantee under restrictive
background knowledge assumptions, such as independent data points or strong adversaries.
To see the need for assumptions, consider that statistical nondisclosure implies $0$-Bayesian$_0$ Differential Privacy (Def.\,\ref{view:dp-ucb}) since both are equivalent to requiring independence between $D$ and $O$ in the case where the adversary's background information is the true distribution over data points.
We believe such a need underlies the view that \DP{} only works with assumptions (Appendix~\ref{app:history-camp2}).
Kasiviswanathan and Smith's Semantic Privacy
is a property about the adversary's ability to do inferences
that does not require such assumptions~\cite{kasiviswanathan14jpc}.
It requires that the probability that
the adversary assigns to the input data points does not change much whether
an individual $i$ provides data or not.
The probability assigned by the adversary when each person provides his data point is
\begin{DIFnomarkup}
\begin{align}
\Cr[D{=}d \given O{=}o, B]
&= \frac{\Fr_{\mc{A}}[\mc{A}(d){=}o]*\Cr[D{=}d \given B]}{\sum_{d'} \Fr_{\mc{A}}[\mc{A}(d'){=}o] * \Cr[D{=}d' \given B]} \ndss{\notag}
\end{align}
\end{DIFnomarkup}
where $D{=}d$ is shorthand for $\bigwedge_{j = 1}^{n} D_j{=}d_j$ with $d = \langle d_1, \ldots, d_n\rangle$.
The probability where person $i$ does not provide data or provides fake data is
\begin{DIFnomarkup}
\[
\frac{\Fr_{\mc{A}}[\mc{A}(d_{-i}d'_i){=}o]*\Cr[D{=}d \given B]}{\sum_{d'} \Fr_{\mc{A}}[\mc{A}(d_{-i}d'_i){=}o] * \Cr[D{=}d' \given B]}
\]
\end{DIFnomarkup}
where
$d'_i$ is the value (possibly the null value) provided instead of the real value
and $d_{-1}d'_i$ is shorthand for $d$ with its $i$th component replaced with $d'_i$.
While we leave fully formalizing the combining of Bayesian credences and frequentist probabilities to future work, intuitively, this probability is $\Cr_{\mc{M}^{\mc{A}},\mc{P}}[D{=}d \given O{=}o, \doo(D_i{=}d'_i), B]$ in our causal notation.

Kasiviswanathan and Smith prove that \DP{} and
Semantic Privacy are closely
related~\cite[Thm.\,2.2]{kasiviswanathan14jpc}.
In essence, they show that \DP{} ensures that
$\Cr[D{=}d \given O{=}o, B]$
and
$\Cr[D{=}d \given O{=}o, \doo(D_i {=} d'_i), B]$
are close in nearly the same sense as it ensures that
$\Fr[O{=}o]$
and
$\Fr[O{=}o \given \doo(D_i{=}d'_i)]$
are close.
That is, it guarantees that an adversary's beliefs will not change
much relative to whether you decide to provide data or not, providing an inference-based view of \DP{}.

To gain intuition about these results, let us consider
the findings of Wang and Kosinski~\cite{wang2017gaydar}, which
show the possibility of training a neural network to predict a person's sexual orientation from a photo of their face.
If this model had been produced with \DP{}, then each study participant would know that their participation had little to do with the model's final form or success.
However, inferential threats would remain.
An adversary can use the model and a photo of an individual to infer the individual's sexual orientation, whether that individual participated in the study or not.
Less obviously, an adversary might have some background knowledge allowing it to repurpose the model to predict people's risks of certain health conditions.
Such difficult to predict associations may already be used for marketing~\cite{hill12forbes} (cf.\,\cite{piatetsky14kdnuggets}).

An individual facing the option of participating in such a study
may attempt to reason about how likely such repurposing is.
Doing so requires the difficult task of characterizing the adversary's background knowledge since Dwork and Naor's proof shows that the possibility cannot be categorically eliminated~\cite{dwork08jpc}.
Furthermore, if the individual decides that the study is too risky, merely declining to participate will do little to mitigate the risk since \DP{} ensures that the individual's data would have had little effect on the model.
Rather, the truly concerned individual would have to lobby others to not participate. %
For this reason, both the causal and associative views of privacy have
their uses, with the causal view being relevant to a single
potential participant's choice and the associative, to the
participants collectively.
One can debate whether such collective properties are privacy per se or some other value since it goes beyond protecting personal data~\cite{mcsherry16blog1}.

\section{Conclusion and Further Implications}
\label{sec:conc}

Although it is possible to view \DP{} as an associative property with an independence assumption,
we have shown that it is cleaner to view \DP{} as a causal property without such an assumption.
We believe that this difference in goals helps to explain why one line of research claims that \DP{} requires an assumption of independence while another line denies it: the assumption is not required but does yield stronger conclusions.

We believe these results have implications beyond explaining the differences between these two lines.
Having shown a precise sense in which \DP{} is a causal property, we can use the results of statistics, experimental design, and science about causation while studying \DP{}.
For example, various papers have attempted to reverse engineer or test whether a system has \DP{}~\cite{tang17arxiv,ding18ccs,bichsel18ccs}.
Authors of follow up works may leverage by pre-existing experimental methods and statistical analyses for measuring effect sizes
that apply with or without access to causal models.
\tufte{In more detail,
Tang et al.\ studied Apple's claim that MacOS uses \DP{} and attempted to reverse engineer the degree $\epsilon$ of privacy used by Apple from the compiled code and configuration files~\cite{tang17arxiv}.
 Consider a version of this problem in which the system purportedly providing \DP{} is a server controlled by some other entity.
 In this case, the absence of code and configuration files necessitates a blackbox investigation of the system.
 From the outside, we can study whether such a system has \DP{} as advertised by using experiments and significance testing~\cite{fisher35doe} similar to how Tschantz et al.\@'s prior work uses it for studying information flow as a causal property~\cite{tschantz15csf}.  %
 Indeed, Ding~et~al.\@ recently used significance testing~\cite{ding18ccs}
 and  Bichsel~et~al.\@ confidence intervals~\cite{bichsel18ccs}
 to find violations of \DP{},
 without naming their connections to causation or effect sizes.
 Alternately, using the associative view, we could approach the problem using observational studies.%
}

In the opposite direction, the
natural sciences can use \DP{} as an effect-size metric, inheriting all the pleasing properties known of \DP{}.
For example, \DP{} composes cleanly with itself, both in sequence and in parallel~\cite{mcsherry09sigmod}.
The same results would also apply to the effect-size metric that \DP{} suggests.

Finally, showing that \DP{} is in essence a measure of effect sizes explains why it, or properties based upon it, has shown up in areas other than privacy, including fairness~\cite{dwork12itcs}, ensuring statistical validity~\cite{dwork15stoc,dwork15nips,dwork15science}, and adversarial machine learning~\cite{lecuyer18arxiv}.
While it may be surprising that privacy is related to such a diverse set of areas, it is not surprising that causation is, given the central role the concept plays in science.
What is actually happening is that causal reasoning is making its importance felt in each of these areas, including in privacy.
That it has implicitly shown up in at least four areas of research suggests that causal reasoning should play a more explicit role in computer science.

\ndss{\paragraph*{Acknowledgements}
We thank
Arthur Azevedo de Amorim,
Deepak Garg,
Ashwin Machanavajjhala,
and Frank McSherry
for comments on this work.
We received funding from
the NSF
(Grants 1514509, %
1704845, %
and 1704985) %
and DARPA (FA8750-16-2-0287). %
The opinions in this paper are those of the authors and do not
necessarily reflect the opinions of any funding sponsor or the United
States Government.}
\tufte{\paragraph{Acknowledgements.}
We thank
Arthur Azevedo de Amorim,
Ashwin Machanavajjhala,
and Frank McSherry
for comments on earlier versions of this work.
We thank Deepak Garg for conversations about causation.
We gratefully acknowledge funding support from
the National Science Foundation
(Grants 1514509, %
1704845, %
and 1704985) %
and DARPA (FA8750-16-2-0287). %
The opinions in this paper are those of the authors and do not
necessarily reflect the opinions of any funding sponsor or the United
States Government.}

\appendix
\tufte{\clearpage\centerline{\Large\bf Appendices}\nopagebreak}

\section{Two Views of Differential Privacy:\tufte{\\} A Brief History}
\label{app:history}

Throughout this paper, we have mentioned two lines of work about \DP{}.
The historically first line, associated with its creators, views \DP{} as not requiring additional assumptions, such as independent data points or an adversary that already knows all but one data point.
The historically second line views such assumptions as needed by or implicit in \DP{}.
Here, we briefly recount the history of the two lines. 
\tufte{We start with their historical antecedents that pre-date \DP{}.  Having not participated in differential privacy's formative years, we welcome refinements to our account.}
\subsection{Before Differential Privacy}

The idea of a precise framework for mathematically modeling the conditions under which an adversary does not learn something perhaps starts with Shannon's work on \emph{perfect security} in 1949~\cite{shannon49bell}.
In 1984, this idea led to Goldwasser and Silvio's cryptographic notion of \emph{semantic security}, which relaxes Shannon's requirement by applying to only polynomially computationally bounded adversaries~\cite{goldwasser84jcss}
(with antecedents in their earlier 1982 work~\cite{goldwasser82stoc}).

Apparently independently, the statistics community also considered limiting what an adversary would learn.
One early work cited by \DP{} papers (e.g.,~\cite{dwork06icalp}) is Dalenius's 1977 paper on \emph{statistical disclosure}~\cite{dalenius77statistik}.
Dalenius defines statistical disclosures in terms of a \emph{frame} of \emph{objects}, for example, a sampled population of people~\cite[\S 4.1]{dalenius77statistik}.
The objects have data related to them~\cite[\S 4.2]{dalenius77statistik}.
A survey releases some \emph{statistics} over such data for the purpose of fulfilling some \emph{objective}~\cite[\S 4.3]{dalenius77statistik}.
Finally, the adversary may have access to \emph{extra-objective data}, which is auxiliary information other than the statistics released as part of the survey.
Dalenius defines a \emph{statistical disclosure} as follows~\cite[\S 5]{dalenius77statistik}:
\begin{quote}
If the release of the statistics $S$ makes it possible to determine the value $D_K$ more accurately than is possible without access to $S$, a disclosure has taken place [\ldots]
\end{quote}
where $D_K$ is the value of the attribute $D$ held by the object (e.g., person) $K$.
The attribute $D$ and object $K$ may be used in the computation of $S$ or not.
The extra-objective data may be used in computing the estimate of $D_K$.

As pointed out by Dwork~\cite{dwork06icalp}, Dalenius's work is both similar to and different from the aforementioned work on cryptosystems.
The most obvious difference is looking at databases and statistics instead of cryptosystems and messages.
However, the more significant difference is the presence of the \emph{objective} with a \emph{benefit}, or the need for \emph{utility} in Dwork's nomenclature.
That is, the released statistics is to convey some information to the public;
whereas, the encrypted message, the cryptosystem's analog to the statistic, only needs to convey information to the intended recipient.
Dalenius recognized that this additional need makes the elimination of statistical disclosures ``not operationally feasible'' and ``would place unreasonable restrictions on the kind of statistics that can be released''~\cite[\S 18]{dalenius77statistik}.

Even before the statistics work on statistical nondisclosure,
statistical research by S.~L.~Warner in 1965 introduced the \emph{randomized response} method of providing \DP{}~\cite{warner65asa}.
(His work is more similar to the local formulation of \DP{}~\cite{dwork06eurocrypt}.)
The randomized response model and statistical disclosure can be viewed as the prototypes of the first and second lines of reseach respectively, although these early works appear to have had little impact on the actual formation of the lines of reseach over a quarter century later.

\subsection{Differential Privacy}

In March~2006, Dwork, McSherry, Nissim, and Smith presented a paper containing the first modern instance of \DP{} under the name of ``$\epsilon$-indistinguishable''~\cite{dwork06crypto}.
The earliest use of the term ``differential privacy'' comes from a paper by Dwork presented in July~2006~\cite{dwork06icalp}.
This paper of Dwork explicitly rejects the view that \DP{} provides associative or inferential privacy~\cite[p.\,8]{dwork06icalp}:
\begin{quote}
Note that a bad disclosure can still occur [despite \DP{}], but [\DP{}] assures the individual that it will not be the presence of her data that causes it, nor could the disclosure be avoided through any action or inaction on the part of the user.
\end{quote}
and further contains a proof that preventing Dalenius's statistical disclosures while releasing useful statistics is impossible.
(The proof was joint work with Naor, with whom Dwork later further developed the impossibility result~\cite{dwork08jpc}.)
\tufte{
She defines \DP{} as follows, a presentation that implicitly promotes a causal view~\cite[pp.\,8--9]{dwork06icalp}:
 \begin{quote}
 As noted in the example of Terry Gross’ height, an auxiliary information generator with information about someone not even in the database can cause a privacy breach to this person. In order to sidestep this issue we change from absolute guarantees about disclosures to relative ones: any given disclosure will be, within a small multiplicative factor, just as likely whether or not the individual participates in the database. As a consequence, there is a nominally increased risk to the individual in participating, and only nominal gain to be had by concealing or misrepresenting one’s data. Note that a bad disclosure can still occur, but our guarantee assures the individual that it will not be the presence of her data that causes it, nor could the disclosure be avoided through any action or inaction on the part of the user.
 
 \noindent\textbf{Definition 2.} \emph{A randomized function $\mc{K}$ gives \emph{$ε$-differential privacy} if for all data sets $D_1$ and $D_2$ differing on at most one element, and all $S ⊆ \mathit{Range}(\mc{K})$,}
 \begin{align*}
 \Pr[\mc{K}(D_1) ∈ S] &≤ \exp(ε) × \Pr[\mc{K}(D_2) ∈ S] \quad (1)
 \end{align*}
 
 A mechanism $\mc{K}$ satisfying this definition addresses concerns that any participant might have about the leakage of her personal information $x$: even if the participant removed her data from the data set, no outputs (and thus consequences of outputs) would become significantly more or less likely. For example, if the database were to be consulted by an insurance provider before deciding whether or not to insure Terry Gross, then the presence or absence of Terry Gross in the database will not significantly affect her chance of receiving coverage.
\end{quote}}
Later works further expound upon their position~\cite{dwork06eurocrypt,kasiviswanathan08arxiv}.

\subsection{Questions Raised about Differential Privacy}
\label{app:history-camp2}

In 2011, papers started to question whether \DP{} actually provides a meaningful notion of privacy~\cite{cormode11kdd,kifer11sigmod,gehrke11tcc}.
These papers point to the fact that a released statistic can enable inferring sensitive information about a person, similar to the attacks Dalenius wanted to prevent~\cite{dalenius77statistik}, even when that statistic was computed using a differentially private algorithm.
While the earlier work on \DP{} acknowledged this limitation, these papers provide examples where correlations, or more generally associations, between data points can enable inferences that some people might not expect to be possible under \DP{}.
These works kicked off a second line of research (including, e.g., \cite{kifer12pods,kifer14database,he14sigmod,chen14vldbj,zhu15tifs,liu16ndss}) attempting to find stronger definitions that account for such correlations.
In some cases, these papers assert that such inferential threats are violations of privacy and not what people expect of \DP{}.
For example, Liu~et~al.'s abstract states that associations between data points can lead to ``degradation in expected privacy levels''~\cite{liu16ndss}.
The rest of this subsection provides details about these papers.

In 2011, Kifer and Machanavajjhala published a paper stating that the first popularized claim about \DP{} is that ``It makes no assumptions about how data are generated''~\cite[p.\,1]{kifer11sigmod}.
The paper then explains that ``a major criterion for a privacy definition is the following: can it hide the evidence of an individual's participation in the data generating process?''~\cite[p.\,2]{kifer11sigmod}.
It states~\cite[p.\,2]{kifer11sigmod}:
\begin{quote}
We believe that under any reasonable formalization of evidence of participation, such evidence can be encapsulated by exactly one tuple [as done by \DP{}] only when all tuples are independent (but not necessarily generated from the same distribution). We believe this independence assumption is a good rule of thumb when considering the applicability of differential privacy.
\end{quote}
For this reason, the paper goes on to say ``Since evidence of participation requires additional assumptions about the data (as we demonstrate in detail in Sections~3 and~4), this addresses the first popularized claim -- that differential privacy requires no assumptions about the data''~\cite[p.\,2]{kifer11sigmod}.  
From context, we take ``addresses'' to mean \emph{invalidates} since the paper states ``The goal of this paper is to clear up misconceptions about differential privacy''~\cite[p.\,2]{kifer11sigmod}.

In 2012, Kifer and Machanavajjhala published follow up work stating that ``we use [the Pufferfish framework] to formalize and prove the statement that differential privacy assumes independence between records''~\cite[p.\,1]{kifer12pods}.
It goes on to say ``Assumptionless privacy definitions are a myth: if one wants to publish useful, privacy-preserving sanitized data then one \emph{must} make assumptions about the original data and data-generating process''~\cite[p.\,1, emphasis in original]{kifer12pods}.
In 2014, Kifer and Machanavajjhala published a journal version of their 2012 paper, which makes a similar statement: ``Note that assumptions are absolutely necessary -- privacy definitions that can provide privacy guarantees without making any assumptions provide little utility beyond the default approach of releasing nothing at all''~\cite[p.\,3:5]{kifer14database}.
However, this version is, overall, more qualified.
For example, it states ``The following theorem says that if we have any correlations between records, then some differentially private algorithms leak more information than is allowable (under the odds ratio semantics in Section 3.1)''~\cite[3:12--13]{kifer14database}, which makes it clear that the supposed shortcoming of \DP{}
in the face of correlated data points %
is relative to a particular notion of privacy presented in that paper, roughly, reducing uncertainty about some sensitive fact about a person.

Also in 2014, He~et~al.\@ published a paper building upon the Pufferfish framework~\cite{he14sigmod}.
Referring to the conference version~\cite{kifer12pods}, He~et~al.\@ states~\cite[p.\,1]{he14sigmod}:
\begin{quote}
[Kifer and Machanavajjhala] showed that differential privacy is equivalent to a specific instantiation of the Pufferfish framework, where (a) every property about an individual's record in the data is kept secret, and (b) the adversary assumes that every individual is independent of the rest of the individuals in the data (no correlations). We believe that these shortcomings severely limit the applicability of differential privacy to real world scenarios that either require high utility, or deal with correlated data.
\end{quote}
and ``Recent work [by Kifer and Machanavajjhala] showed that differentially private mechanisms could still lead to an inordinate disclosure of sensitive information when adversaries have access to publicly known constraints about the data that induce correlations across tuples''~\cite[p.\,3]{he14sigmod}.

In 2013, Li~et~al.\@ published a paper that states ``differential privacy's main assumption is independence''~\cite[p.\,2]{li13ccs}.
Similar, to the papers by Kifer and Machanavajjhala, this paper assumes a technical definition of privacy, \emph{positive membership privacy}, and makes this assertion since independence is required for \DP{} to imply it.
The paper also claims that ``the original definition of differential privacy assumes that the adversary has precise knowledge of all the tuples in the dataset''~\cite[p.\,10]{li13ccs}, which we take as a reference to the strong adversary assumption.

Chen~et~al.'s 2014 paper is the first of three attempting to provide an associative version of privacy, motivated by Pufferfish, in the face of correlated data~\cite{chen14vldbj}.
It states ``$\epsilon$-differential privacy fails to provide the claimed privacy guarantee in the correlated setting''~\cite[p.\,2]{chen14vldbj} and ``$\epsilon$-differential privacy is built on the assumption that all underlying records are independent of each other''~\cite[p.\,7]{chen14vldbj}.

The second paper, Zhu~et~al.'s paper, published in 2015, provides a more accurate accounting of correlations~\cite{zhu15tifs}.
It states~\cite[p.\,229]{zhu15tifs}: 
\begin{quote}
An adversary with knowledge on correlated information will have higher chance of obtaining the privacy information, and violating the definition of differential privacy. Hence, how to preserve rigorous differential privacy in a correlated dataset is an emerging issue that needs to be addressed.
\end{quote}
It further asserts~\cite[p.\,231]{zhu15tifs}:
\begin{quote}
In the past decade, a growing body of literature has been published on differential privacy. Most existing work assumes that the dataset consists of independent records.
\end{quote}
and ``a major disadvantage of traditional differential privacy is the overlook of the relationship among records, which means that the query result leaks more information than is allowed''~\cite[p.\,232]{zhu15tifs}.

The third paper, by Liu~et~al.\@ in 2016, provides an even more accurate accounting of correlations~\cite{liu16ndss}.
A blog post by one of the authors, Mittal, announcing the paper states ``To provide its guarantees, DP implicitly assumes that the data tuples in the database, each from a different user, are all independent.''~\cite{mittal16blog}.
In five comments on this blog post, McSherry posted a summary of his concerns about their paper and blog post.
McSherry also treats the paper at length in a blog post~\cite{mcsherry16blog2}.
McSherry highlights three statements made by the paper that he finds false~\cite{mcsherry16blog2}:
(1) ``For providing this guarantee, differential privacy mechanisms assume independence of tuples in the database''~\cite[p.\,1]{liu16ndss},
(2) ``To provide its guarantees, DP mechanisms assume that the data tuples (or records) in the database, each from a different user, are all independent.''~\cite[p.\,1]{liu16ndss}, and
(3) ``However, the privacy guarantees provided by the existing DP mechanisms are valid only under the assumption that the data tuples forming the database are pairwise independent''~\cite[p.\,2]{liu16ndss}.

A somewhat different tack is taken in a 2016 paper by Cuff and Yu, which instead focuses on the strong adversary assumption~\cite[p.\,2]{cuff16ccs}:
\begin{quote}
The definition of $(\epsilon, \delta)$-DP involves a notion of neighboring database instances.  Upon examination one realizes that this has the affect of assuming that the adversary has already learned about all but one entry in the database and is only trying to gather additional information about the remaining entry.  We refer to this as the strong adversary assumption, which is implicit in the definition of differential privacy.
\end{quote}

Yang~et~al.\@'s 2015 paper allows either assumption~\cite[\S1.2]{yang15sigmod}:
\begin{quote}
Differential privacy is designed to preserve the privacy in the face of intrusions by the strongest adversary who exactly knows everything about all individual entities except the object of its attack.
[\dots]
In fact, as we will show in Section~3, differential privacy does guarantee privacy against intrusion by any adversary when all the entities in the database are independent. 
\end{quote}

\subsection{Responses}

In addition to the aforementioned blog post by McSherry~\cite{mcsherry16blog2},
other works by those promoting the original view of \DP{} have also re-asserted that \DP{} was never intended to prevent all inferential privacy threats and that doing so is impossible~\cite{bassily13focs,kasiviswanathan14jpc,mcsherry16blog1}. %
In a different blog post, 
McSherry goes the furthest, questioning whether wholesale inferential privacy is the normal meaning of ``privacy'' or even an appealing concept~\cite{mcsherry16blog1}.
He calls it ``forgettability'', invoking the European Union's right to be forgotten, and points out that preventing inferences prevents people from using data and scientific progress.
He suggests that perhaps people should only have an expectation to the privacy of data they own, as provided by \DP{}, and not to the privacy of data about them.
He challenges the line of research questioning \DP{} (Appendix~\ref{app:history-camp2}) to justify the view that forgettability is a form of privacy.

We know no works explicitly responding to this challenge.

\section{Counterexample Involving Zero Probability for Strong Adversary D.P.}
\label{app:chr-csb}

Consider Definition~\ref{view:dp-utc} modified to look at one distribution $\mc{P}$, which represents the actual distribution of the world.
\begin{chr}
\label{chr:dp-csb}
A randomized algorithm $\mc{A}$ is said to be
\emph{$ε$-Strong Adversary Differentially Private for One Distribution $\mc{P}$}
if for all databases $d,d' ∈ \mc{D}^n$ at Hamming distance at most $1$,
and for all output values $o$,
if
$\Fr[D{=}d] > 0$ and
$\Fr[D{=}d'] > 0$
then
\begin{align}
\Fr_{\mc{P},\mc{A}}[O{=}o \given D{=}d] ≤ e^ε * \Fr_{\mc{P},\mc{A}}[O{=}o \given D{=}d']
\end{align}
where $O = \mc{A}(D)$ and $D = \langle D_1, \ldots, D_n\rangle$.
\end{chr}

To prove that this
does not imply Definition~\ref{dfn:dp}, consider the case of a database holding a single data point whose value could be $0$, $1$, or $2$.
Suppose the population $\mc{P}$ is such that $\Fr_{\mc{P}}[D_1{=}2] = 0$.
Consider an algorithm $\mc{A}$ such that for the given population $\mc{P}$,
\begin{align}
\Fr_{\mc{A}}[\mc{A}(0){=}0] &= 1/2 & \Fr_{\mc{A}}[\mc{A}(0){=}1] &= 1/2 \\
\Fr_{\mc{A}}[\mc{A}(1){=}0] &= 1/2 & \Fr_{\mc{A}}[\mc{A}(1){=}1] &= 1/2 \\
\Fr_{\mc{A}}[\mc{A}(2){=}0] &= 1   & \Fr_{\mc{A}}[\mc{A}(2){=}1] &= 0
\end{align}
The algorithm does not satisfy Definition~\ref{dfn:dp} due to its behavior on the input $2$.
However, using~\eqref{eqn:conditioning-on-all},
\begin{align}
\Fr_{\mc{P},\mc{A}}[O{=}0 \given D_1{=}0] &= 1/2 & \Fr_{\mc{P},\mc{A}}[O{=}1 \given D_1{=}0] &= 1/2 \ndss{\notag}\\
\Fr_{\mc{P},\mc{A}}[O{=}0 \given D_1{=}1] &= 1/2 & \Fr_{\mc{P},\mc{A}}[O{=}1 \given D_1{=}1] &= 1/2 \ndss{\notag}
\end{align}
While~\eqref{eqn:conditioning-on-all} says nothing about $D_1{=}2$ since that has zero probability,
this is sufficient to show that the algorithm satisfies
Definition~\ref{chr:dp-csb} since it only applies to data points of non-zero probability.
Thus, the algorithm satisfies Definition~\ref{chr:dp-csb} but not Definition~\ref{dfn:dp}.

\section{Details of Causation}
\label{app:causation-details}

\ndss{We use a slight modification of Pearl's models.  The models we use are suggested by Pearl for handling ``inherent'' randomness~\cite[p.~220]{pearl09book} and differs from the model he typically uses (his Definition~7.1.6) by allowing randomization in the structural equations $F_V$.  We find this randomization helpful for modeling the randomization within the algorithm $\mc{A}$.}

Formally, let $\llb \mc{M} \rrb(\vec{x}).\vec{Y}$ be the joint distribution over values for the variables $\vec{Y}$ that results from the background variables $\vec{X}$ taking on the values $\vec{x}$ (where these vectors use the same ordering).
That is, $\llb \mc{M} \rrb(\vec{x}).\vec{Y}(\vec{y})$ represents the probability of $\vec{Y} = \vec{y}$ given that the background variables had values $\vec{X} = \vec{x}$.
Since the SEM is non-recursive this can be calculated in a bottom up fashion.
We show this for the model $\mc{M}^{\mc{A}}$ with $D_i := R_i$ for all $i$, $D := \langle D_1, \ldots, D_n\rangle$, and $O := \mc{A}(D)$:
\begin{fw}
\begin{align}
\llb \mc{M}^{\mc{A}} \rrb(r_1, \ldots, r_n).R_i(r_i) &= 1 \ndss{\notag}
\end{align}
\begin{align}
\llb \mc{M}^{\mc{A}} \aln\rrb(r_1, \ldots, r_n).D_i(r_i) \brkalnaln{}= \Fr_{F_{D_i}}[F_{D_i}(R_i){=}r_i] = \Fr_{F_{D_i}}[R_i{=}r_i] = 1 \ndss{\notag}
\end{align}
\begin{align}
\llb \mc{M}^{\mc{A}} \aln\rrb(r_1, \ldots, r_n).D(\langle r_1, \ldots, r_n\rangle)
  \brkalnaln{} = \Fr_{F_{D}}[F_{D}(D_1,\ldots,D_n){=}\langle r_1,\ldots,r_n\rangle] \ndss{\notag}\\
  &= \Fr_{F_{D}}[F_{D}(F_{D_1}(R_1),\ldots,F_{D_n}(R_n)){=}\langle r_1,\ldots,r_n\rangle] \ndss{\notag}\\
  &= \Fr_{F_{D}}[F_{D}(R_1,\ldots,R_n){=}\langle r_1,\ldots,r_n\rangle] \ndss{\notag}\\
  &= \Fr_{F_{D}}[\langle R_1,\ldots,R_n\rangle{=}\langle r_1,\ldots,r_n\rangle] = 1 \ndss{\notag}
\end{align}
and
\begin{align}
\llb \mc{M}^{\mc{A}} \rrb(r_1, \ldots, r_n).O(o)
  &= \Fr_{F_O}[F_O(D){=}o] \brkaln{}
  = \Fr_{\mc{A}}[\mc{A}(\langle r_1, \ldots, r_n\rangle){=}o] \ndss{\notag}
\end{align}
\end{fw}

We can raise the calculations above to work over $\mc{P}$ instead of a concrete assignment of values $\vec{x}$.
Intuitively, the only needed change is that, for background variables $\vec{X}$,
\begin{align}
\Fr_{\olangle \mc{M}^{\mc{A}}, \mc{P}\orangle}[\vec{Y}{=}\vec{y}] &= \sum_{\vec{x} \in \vec{\mc{X}}} \Fr_{\mc{P}\mc{A}}[\vec{X}{=}\vec{x}] * \llb \mc{M}^{\mc{A}} \rrb(\vec{x}).\vec{Y}(\vec{y}) \ndss{\notag}
\end{align}
where $\vec{X}$ are all the background variables.\footnote{This is Pearl's equation (7.2) raised to work on probabilistic structural equations $F_V$~\cite[p.~205]{pearl09book}.}
The following lemma will not only be useful, but will illustrate the above general points on the model $\mc{M}^{\mc{A}}$ that concerns us.
\begin{lem}
\label{lem:doing-on-all}
For all algorithms $\mc{A}$, $\mc{P}$, all $o$, and all $d_1, \ldots, d_n$,
\begin{align}
\aln\Fr_{\olangle \mc{M}^{\mc{A}}, \mc{P}\orangle}[O{=}o \given \doo(D_1{:=}d_1, \ldots, D_n{:=}d_n)] \brkalnaln{XXXXXXXXXXXXXX}= \Fr_{\mc{A}}[\mc{A}(d_1,\ldots,d_n){=}o] \ndss{\notag} %
\end{align}
\end{lem}
\begin{proof}
Let $F_{d_i}()$ represent the constant function with no arguments that always returns $d_i$.
The structural equation for $D_i$ is $F_{d_i}$ in $\mc{M}^{\mc{A}}[D_1{:=}d_1]\cdots[D_n{:=}d_n]$.
As before, we compute bottom up, but this time on the modified SEM:
\begin{fw}
\begin{align}
\llb \mc{M}^{\mc{A}}[D_1{:=}d_1]\cdots[D_n{:=}d_n] \rrb(r_1, \ldots, r_n).R_i(r_i) &= 1 \ndss{\notag}
\end{align}
\begin{multline}
\llb \mc{M}^{\mc{A}}[D_1{:=}d_1]\cdots[D_n{:=}d_n] \rrb(r_1, \ldots, r_n).D_i(d_i) \brk = \Fr_{F_{d_i}}[F_{d_i}(){=}d_i] = 1
\end{multline}
\begin{align}
\llb \mc{M}^{\mc{A}}[D_1{:=}d_1]\aln\cdots[D_n{:=}d_n] \rrb(r_1, \ldots, r_n).D(\langle d_1, \ldots, d_n\rangle)
  \brkalnaln{}
   = \Fr_{F_{D}}[F_{D}(D_1,\ldots,D_n){=}\langle d_1,\ldots,d_n\rangle] \ndss{\notag}\\
  &= \Fr_{F_{D}}[F_{D}(F_{D_1}(),\ldots,F_{D_n}()){=}\langle d_1,\ldots,d_n\rangle] \ndss{\notag}\\
  &= \Fr_{F_{D}}[F_{D}(d_1,\ldots,d_n){=}\langle d_1,\ldots,d_n\rangle] \ndss{\notag}\\
  &= \Fr_{F_{D}}[\langle d_1,\ldots,d_n\rangle{=}\langle d_1,\ldots,d_n\rangle] = 1 \ndss{\notag}
\end{align}
\begin{multline}
\llb \mc{M}^{\mc{A}}[D_1{:=}d_1]\cdots[D_n{:=}d_n] \rrb(r_1, \ldots, r_n).O(o)
  \brk = \Fr_{F_O}[F_O(D){=}o]
  = \Fr_{\mc{A}}[\mc{A}(\langle d_1, \ldots, d_n\rangle){=}o] \ndss{\notag}
\end{multline}
\end{fw}
Thus,
\begin{fw}
\begin{align}
\aln\Fr_{\olangle \mc{M}^{\mc{A}}, \mc{P}\orangle}[O{=}o \given \doo(D_1{:=}d_1, \ldots, D_n{:=}d_n)]
\brkalnaln{} = \Fr_{\olangle \mc{M}^{\mc{A}}[D_1{:=}d_1]\cdots[D_n{:=}d_n], \mc{P}\orangle}[O{=}o]  \ndss{\notag}\\
&= \sum_{\vec{r} \in \mc{R}^n} \Fr_{\mc{P}}[\vec{R}{=}\vec{r}] * \llb \mc{M}^{\mc{A}}[D_1{:=}d_1]\cdots[D_n{:=}d_n] \rrb(\vec{r}).O(o) \ndss{\notag}\\
&= \sum_{\vec{r} \in \mc{R}^n} \Fr_{\mc{P}}[\vec{R}{=}\vec{r}] * \Fr_{\mc{A}}[\mc{A}(\langle d_1, \ldots, d_n\rangle){=}o] \ndss{\notag}\\
&= \Fr_{\mc{A}}[\mc{A}(\langle d_1, \ldots, d_n\rangle){=}o] * \sum_{\vec{r} \in \mc{R}^n} \Fr_{\mc{P}}[\vec{R}{=}\vec{r}] \ndss{\notag}\\
&= \Fr_{\mc{A}}[\mc{A}(\langle d_1, \ldots, d_n\rangle){=}o] * 1 \ndss{\notag}\\
&= \Fr_{\mc{A}}[\mc{A}(\langle d_1, \ldots, d_n\rangle){=}o] \ndss{\notag}
\end{align}
\end{fw}
\end{proof}

\begin{lem}
\label{lem:doing-on-one}
For all algorithms $\mc{A}$, $\mc{P}$, $o$, $j$, and $d'_j$,
\begin{fw}
\begin{align}
&\Fr_{\olangle \mc{M}^{\mc{A}}, \mc{P}\orangle}[O{=}o \given \doo(D_j{=}d'_j)] \ndss{\notag}\\
&= \bigsum{\langle r_1,\ldots,r_{j-1},r_{j+1},\ldots,r_{n}\rangle \in \mc{R}^{n-1}} \Fr_{\mc{P}}\left[\wedge_{i\in\{1,\ldots,j-1,j+1,\ldots,n\}} R_i{=}r_i \right] \brkaln{XXXXXX} * \Fr_{\mc{A}}[\mc{A}(r_1,\ldots,r_{j-1},d'_j,r_{j+1},\ldots,r_n){=}o] \ndss{\notag}\\
&= \bigsum{\langle d_1,\ldots,d_{j-1},d_{j+1},\ldots,d_{n}\rangle \in \mc{D}^{n-1}} \Fr_{\olangle \mc{M}^{\mc{A}}, \mc{P}\orangle}\left[\wedge_{i\in\{1,\ldots,j-1,j+1,\ldots,n\}} D_i{=}d_i \right] \brkaln{XXXXXX} * \Fr_{\mc{A}}[\mc{A}(d_1,\ldots,d_{j-1},d'_j,d_{j+1},\ldots,d_n){=}o] \ndss{\notag}
\end{align}
\end{fw}
\end{lem}
\begin{proof}
With out loss of generality, assume $j$ is $1$.
Let $F_{d'_1}()$ represent the constant function with no arguments that always returns $d'_1$.
The structural equation for $D_1$ is $F_{d'_1}$ in $\mc{M}^{\mc{A}}[D_1{:=}d'_1]$.
As before, we compute bottom up, but this time on the modified SEM:
\begin{fw}
\begin{align}
\llb \mc{M}^{\mc{A}}[D_1{:=}d'_1] \rrb(r_1, \ldots, r_n).R_i(r_i) &= 1 \ndss{\notag}
\end{align}
holds as before.
The behavior of $D_i$ varies based on whether $i = 1$:
\begin{align}
\llb \mc{M}^{\mc{A}}[D_1{:=}d'_1] \rrb(r_1, \ldots, r_n).D_1(d'_1) &= \Fr_{F_{d'_1}}[F_{d'_1}(){=}d'_1] = 1 \ndss{\notag}\\
\llb \mc{M}^{\mc{A}}[D_1{:=}d'_1] \rrb(r_1, \ldots, r_n).D_i(r_i) &= \Fr_{F_{D_i}}[F_{D_i}(R_i){=}r_i] \brkaln{} = \Fr_{F_{D_i}}[R_i{=}r_i] = 1 \tufte{&&\text{for all $i \neq 1$}} \ndss{\notag}
\end{align}
\ndss{for all $i \neq 1$.}
Thus,
\begin{align}
\aln\llb \mc{M}^{\mc{A}}[D_1{:=}d'_1] \rrb(r_1, \ldots, r_n).D(\langle d'_1, r_2, \ldots, r_n\rangle)
  \brkalnaln{} = \Fr_{F_{D}}[F_{D}(D_1,D_2\ldots,D_n){=}\langle d'_1,r_2,\ldots,r_n\rangle] \ndss{\notag}\\
  &= \Fr_{F_{D}}[F_{D}(F_{d'_1}(),F_{D_2}(R_2),\ldots,F_{D_n}(R_n)){=}\langle d'_1,r_2,\ldots,r_n\rangle] \ndss{\notag}\\
  &= \Fr_{F_{D}}[F_{D}(d'_1,r_2,\ldots,r_n){=}\langle d'_1,r_2,\ldots,r_n\rangle] \ndss{\notag}\\
  &= \Fr_{F_{D}}[\langle d'_1,r_2,\ldots,d_n\rangle{=}\langle d'_1,r_2,\ldots,r_n\rangle] = 1 \ndss{\notag}
\end{align}
and
\begin{multline}
\llb \mc{M}^{\mc{A}}[D_1{:=}d'_1]\rrb(r_1, \ldots, r_n).O(o)
  \brk = \Fr_{F_O}[F_O(D){=}o]
  = \Fr_{\mc{A}}[\mc{A}(\langle d'_1, r_2, \ldots, r_n\rangle){=}o] \ndss{\notag}
\end{multline}
\end{fw}
Thus,
\begin{fw}
\begin{align}
&\Fr_{\olangle \mc{M}^{\mc{A}}, \mc{P}\orangle}[O{=}o \given \doo(D_1{:=}d'_1)] \ndss{\notag}\\
&= \Fr_{\olangle \mc{M}^{\mc{A}}[D_1{:=}d'_1], \mc{P}\orangle}[O{=}o]  \ndss{\notag}\\
&= \sum_{r_1, \ldots, r_n \in \mc{R}^n} \Fr_{\mc{P}}[R_1{=}r_1, \ldots, R_n{=}r_n] \brkaln{XXXXXXXXX} *  \llb \mc{M}^{\mc{A}}[D_1{:=}d'_1]\rrb(r_1, \ldots, r_n).O(o) \ndss{\notag}\\
&= \sum_{r_1, \ldots, r_n \in \mc{R}^n} \Fr_{\mc{P}}[R_1{=}r_1, \ldots, R_n{=}r_n] \brkaln{XXXXXXXXX} *  \Fr_{\mc{A}}[\mc{A}(\langle d'_1, r_2, \ldots, r_n\rangle){=}o] \ndss{\notag}\\
&= \sum_{r_1, \ldots, r_n \in \mc{R}^n} \Fr_{\mc{P}}[R_1{=}r_1 \given R_2{=}r_2, \ldots, R_n{=}r_n] \brkaln{XXXXXXXXX} *  \Fr_{\mc{P}}[R_2{=}r_2, \ldots, R_n{=}r_n] \brkaln{XXXXXXXXX} *  \Fr_{\mc{A}}[\mc{A}(\langle d'_1, r_2, \ldots, r_n\rangle){=}o] \ndss{\notag}\\
&= \sum_{r_2, \ldots, r_n \in \mc{R}^n}\: \sum_{r_1 \in \mc{R}} \Fr_{\mc{P}}[R_1{=}r_1 \given R_2{=}r_2, \ldots, R_n{=}r_n] \brkaln{XXXXXXXXX} *  \Fr_{\mc{P}}[R_2{=}r_2, \ldots, R_n{=}r_n] \brkaln{XXXXXXXXX} *  \Fr_{\mc{A}}[\mc{A}(\langle d'_1, r_2, \ldots, r_n\rangle){=}o] \ndss{\notag}\\
&= \sum_{r_2, \ldots, r_n \in \mc{R}^n} \Fr_{\mc{P}}[R_2{=}r_2, \ldots, R_n{=}r_n] \brkaln{XXXXXXXXX} *  \Fr_{\mc{A}}[\mc{A}(\langle d'_1, r_2, \ldots, r_n\rangle){=}o] \brkaln{XXXXXXXXX} *  \sum_{r_1 \in \mc{R}} \Fr_{\mc{P}}[R_1{=}r_1 \given R_2{=}r_2, \ldots, R_n{=}r_n] \ndss{\notag}\\
&= \sum_{r_2, \ldots, r_n \in \mc{R}^n} \Fr_{\mc{P}}[R_2{=}r_2, \ldots, R_n{=}r_n] \brkaln{XXXXX} *  \Fr_{\mc{A}}[\mc{A}(\langle d'_1, r_2, \ldots, r_n\rangle){=}o] *  1 \ndss{\notag}\\
&= \sum_{r_2, \ldots, r_n \in \mc{R}^n} \Fr_{\mc{P}}[R_2{=}r_2, \ldots, R_n{=}r_n] \brkaln{XXXXXXXXX} *  \Fr_{\mc{A}}[\mc{A}(\langle d'_1, r_2, \ldots, r_n\rangle){=}o] \ndss{\notag}\\
&= \sum_{d_2, \ldots, d_n \in \mc{D}^n} \Fr_{\mc{P}}[D_2{=}d_2, \ldots, D_n{=}d_n] \brkaln{XXXXXXXXX} *  \Fr_{\mc{A}}[\mc{A}(\langle d'_1, d_2, \ldots, d_n\rangle){=}o] \ndss{\notag}
\end{align}
\end{fw}
where the last line follows since $D_i = R_i$ for $i \neq 1$.
\end{proof}

\bibliographystyle{IEEEtran}
\begin{fw}

\end{fw}

\section*{Colophon}

The authors typeset this document using \LaTeX{} with the \texttt{tufte-handout} document class.
We configured the class with the \texttt{nobib}, \texttt{nofonts}, and \texttt{justified} options.
We also altered the appearance of section headers.

The varying line widths in the document are a purposeful attempt at balancing two competing concerns in typesetting.
On the one hand, people find reading long lines of text difficult, which argues for short line lengths.
On the other hand, series of equations are easier to follow when the individual equations are not broken up across lines, which argues for line lengths long enough to hold the longest equation.
To balance these two concerns, we use a short line length for text, but exceed that length as needed for wide equations.

This compromise sacrifices the consistency of line lengths.
We welcome comments on whether this sacrifice is too high a price to pay for balancing the two aforementioned concerns.

\end{document}